\RequirePackage{silence}
\documentclass[11pt]{article}
\usepackage[in]{fullpage}
\usepackage[normalem]{ulem}
\usepackage{enumitem}
\usepackage{etoolbox}
\usepackage{alglib-tomer}
\usepackage{amsmath}
\usepackage{amssymb}
\usepackage{amsthm}
\usepackage{amsfonts}
\usepackage{thmtools, thm-restate}
\usepackage[table]{xcolor}
\usepackage{tabularx}
\usepackage{xspace}
\usepackage{calc}
\usepackage{parskip}
\usepackage{multirow}
\usepackage{mathtools}
\usepackage[T1]{fontenc}
\usepackage{bbm}
\usepackage{mfirstuc}
\usepackage{hyperref}
\hypersetup{
    colorlinks,
    citecolor=blue,
    filecolor=black,
    linkcolor=black,
    urlcolor=black,
    linktoc=all
}

\newcolumntype{Y}{>{\centering\arraybackslash}X}

\newtheorem{lemma}{Lemma}[section]
\newtheorem{theorem}[lemma]{Theorem}

\newtheorem{corollary}[lemma]{Corollary}
\newtheorem{observation}[lemma]{Observation}

\theoremstyle{definition}\newtheorem{definition}[lemma]{Definition}

\newcommand{\accept}[0]{\textsc{accept}\xspace}
\newcommand{\reject}[0]{\textsc{reject}\xspace}
\newcommand\esttoolow{\textup{\textsc{low}}\xspace}

\newcommand{\abs}[1]{\left|{#1}\right|}
\newcommand{\cond}{\middle |}
\newcommand{\floor}[1]{{\left\lfloor{#1}\right\rfloor}}
\newcommand{\ceil}[1]{{\left\lceil{#1}\right\rceil}}

\newcommand{\E}{\mathop{{\rm E}\/}}
\newcommand{\Var}{\mathop{{\rm Var}\/}}
\newcommand{\Ct}{\mathop{{\rm Ct}\/}}
\newcommand{\eps}{\varepsilon}

\newcommand{\poly}{\mathrm{poly}}

\newcommand{\Bin}{\mathrm{Bin}}
\newcommand{\Geo}{\mathrm{Geo}}
\newcommand{\CDF}{\mathrm{CDF}}
\newcommand{\median}{\mathrm{median}}

\newcommand{\dhist}{\ensuremath{D_\mathrm{H}}}
\newcommand{\dtv}{\ensuremath{d_\mathrm{TV}}}

\title{Optimal mass estimation in the conditional sampling model}
\author{Tomer Adar\thanks{Technion - Israel Institute of Technology, Israel. Email: \href{mailto:tomer-adar@campus.technion.ac.il}{tomer-adar@campus.technion.ac.il}.} \and Eldar Fischer\thanks{Technion - Israel Institute of Technology, Israel. Email: \href{mailto:eldar@cs.technion.ac.il}{eldar@cs.technion.ac.il}. Research supported by an Israel Science Foundation grant number 879/22.} \and Amit Levi\thanks{University of Haifa, Israel. Email: \href{mailto:alevi@cs.haifa.ac.il}{alevi@cs.haifa.ac.il}.}}

\newcommand\procnameZmainZsingle{\textsf{Estimate-element}\xspace}
\newcommand\procnameZmainZsingleHREF{\hyperref[fig:alg:estimation-task]{\procnameZmainZsingle}\xspace}

\newcommand\procnameZsaturationZawareZest{\textsf{SA-Est}\xspace}
\newcommand\procnameZsaturationZawareZestHREF{\hyperref[fig:alg:saturation-aware-estimator]{\procnameZsaturationZawareZest}\xspace}

\newcommand\procnameZtargetZtestZexplicit{\textsf{Target-test-explicit}\xspace}
\newcommand\procnameZtargetZtestZexplicitHREF{\hyperref[fig:alg:target-test-explicit]{\procnameZtargetZtestZexplicit}\xspace}

\newcommand\procnameZtargetZtest{\textsf{Target-test}\xspace}
\newcommand\procnameZtargetZtestHREF{\hyperref[fig:alg:target-test]{\procnameZtargetZtest}\xspace}

\newcommand\procnameZtargetZtestZgross{\textsf{Target-test-gross}\xspace}
\newcommand\procnameZtargetZtestZgrossHREF{\hyperref[fig:alg:target-test-gross]{\procnameZtargetZtestZgross}\xspace}

\newcommand\procnameZpreamble{\textsf{Reference-estimation}\xspace}
\newcommand\procnameZpreambleHREF{\hyperref[fig:alg:est-mu-x-s-x]{\procnameZpreamble}\xspace}

\newcommand\procnameZuncertainZcomparator{\textsf{Weak-uncertain-comparator}\xspace}
\newcommand\procnameZuncertainZcomparatorHREF{\hyperref[fig:alg:uncertain-comparator-new]{\procnameZuncertainZcomparator}\xspace}

\newcommand\procnameZfindZgoodZalpha{\textsf{Find-good-$\alpha$}\xspace}
\newcommand\procnameZfindZgoodZalphaHREF{\hyperref[fig:alg:find-good-alpha-new]{\procnameZfindZgoodZalpha}\xspace}

\newcommand\procnameZestZexpectedZbeta{\textsf{Estimate-$\E[\beta_{x,\alpha}]$}\xspace}
\newcommand\procnameZestZexpectedZbetaHREF{\hyperref[fig:alg:est-beta]{\procnameZestZexpectedZbeta}\xspace}

\newcommand\procnameZinitializeZnewZVx{\textsf{Initialize-new-$V_x$}\xspace}
\newcommand\procnameZinitializeZnewZVxHREF{\hyperref[fig:alg:draw-secret-vx]{\procnameZinitializeZnewZVx}\xspace}

\newcommand\procnameZVxZquery{\textsf{$V_x$-Query}\xspace}
\newcommand\procnameZVxZqueryHREF{\hyperref[fig:alg:query-lazy-vx]{\procnameZVxZquery}\xspace}

\newcommand\procnameZestZsingleZbeta{\textsf{Estimate-$\beta_{x,\alpha}$}\xspace}
\newcommand\procnameZestZsingleZbetaHREF{\hyperref[fig:alg:est-beta-single]{\procnameZestZsingleZbeta}\xspace}

\newcommand\procnameZestZexpectedZhZbeta{\textsf{Estimate-scaled-result}\xspace}
\newcommand\procnameZestZexpectedZhZbetaHREF{\hyperref[fig:alg:est-h-beta]{\procnameZestZexpectedZhZbeta}\xspace}

\newcommand\procnameZstrictZbinaryZsearch{\textsf{Uncertain-binary-search}\xspace}
\newcommand\procnameZstrictZbinaryZsearchHREF{\hyperref[fig:alg:strict-binary-search]{\procnameZstrictZbinaryZsearch}\xspace}

\newcommand\procnameZlearnZhistogramZbuckets{\textsf{Learn-histogram-buckets}\xspace}
\newcommand\procnameZlearnZhistogramZbucketsHREF{\hyperref[fig:alg:peekaboo-learn-histogram-buckets]{\procnameZlearnZhistogramZbuckets}\xspace}

\newcommand\procnameZestimateZboundedZratio{\textsf{Estimate-bounded-ratio}\xspace}
\newcommand\procnameZestimateZboundedZratioHREF{\hyperref[fig:alg:est-bounded-ratio]{\procnameZestimateZboundedZratio}\xspace}

\newcommand\procnameZinitializeZcondZsimulator{\textsf{Initialize-COND-simulator}\xspace}
\newcommand\procnameZsimulateZcond{\textsf{Sample-COND-simulator}\xspace}

\newcommand\paperZtargeterrZpreamble{{\ensuremath{\frac{1}{4}c\eps}}}
\newcommand\paperZtargeterrZevbeta{{\ensuremath{\frac{1}{10^9}}}}
\newcommand\paperZtargeterrZestsingle{{\ensuremath{\frac{\eps^5}{10^{20} (\ln \eps^{-1})^5}}}}
\newcommand\paperZhardcodeZdelta{\ensuremath{\frac{1}{200}}}

\newcommand\paperZvalofZkappa{{\ensuremath{10^{-9} / 45}}}

\newcommand\branchZlbndZexpZdiv{\ensuremath{{\log N}}}
\newcommand\branchZubndZexp{\ensuremath{{8 \log^3 N}}}
\newcommand\branchZubndZq{\ensuremath{{\log \log N - 2 \log \log \log N}}}

\newcommand\branchZfracZbadZkZchoices{\ensuremath{{\frac{1}{\log \log \log N}}}}
\newcommand\branchZfracZbadZleafZforZgoodZk{\ensuremath{{\frac{\log \log \log N}{10 \log \log N}}}}

\begin{document}

\begin{titlepage}
    \maketitle
    \thispagestyle{empty}
    \pagestyle{empty}
    
    \begin{abstract}
        The conditional sampling model, introduced by Canonne, Ron and Servedio (SODA 2014, SIAM J.\ Comput.\ 2015) and independently by Chakraborty, Fischer, Goldhirsh and Matsliah (ITCS 2013, SIAM J.\ Comput.\ 2016), is a common framework for a number of studies concerning strengthened models of distribution testing. A core task in these investigations is that of estimating the mass of individual elements. The above mentioned works, and the improvement of Kumar, Meel and Pote (AISTATS 2025), provided polylogarithmic algorithms for this task.

        In this work we shatter the polylogarithmic barrier, and provide an estimator for the mass of individual elements that uses only $O(\log \log N) + O(\poly(1/\eps))$ conditional samples. We complement this result with an $\Omega(\log\log N)$ lower bound.
        
        We then show that our mass estimator provides an improvement (and in some cases a unifying framework) for a number of related tasks, such as testing by learning of any label-invariant property, and distance estimation between two (unknown) distributions. In light of some known lower bounds for common restricted models, our results imply that the full power of the conditional model is indeed required for the doubly-logarithmic upper bound.
        
        Finally, we exponentially improve the previous lower bound on testing by learning of label-invariant properties from double-logarithmic to $\Omega(\log N)$ conditional samples, whereas our testing by learning algorithm provides an upper bound of $O(\poly(1/\eps)\cdot\log N \log \log N)$.
    \end{abstract}

    \newpage
    \setcounter{tocdepth}{2}
    \tableofcontents
\end{titlepage}

\section{Introduction}
\label{sec:intro}

The property testing framework~\cite{GGR98,RS96} deals with approximate decision making in situations where the input data cannot be read in its entirety. Instead, the algorithm is only allowed to read a very small fraction of the data and deduce some global property based on the observed information.

A well-investigated area of property testing focuses on examining the properties of distributions. In this context, the algorithm can access independent samples from a discrete distribution over $\{1,\ldots,N\}$, and must determine whether to accept or reject the input based on these samples. Specifically, the algorithm receives a parameter $\eps > 0$ and is required to accept any input that meets the property to be tested (with high probability), while rejecting any input that is $\eps$-far (in terms of total variation) from any distribution that fulfills the property (again, with high probability). This model was explicitly defined in \cite{batu-FRSW2000,batuFFKRW2001,gr11} and has garnered considerable attention over the past few decades. 

Somewhat unsurprisingly, a typical sample complexity for distribution testing algorithms is $\widetilde{O}(N^\delta)$ for some constant $\delta<1$. Even for testing whether a distribution is uniform, one of the most basic and simple distribution properties, a tight bound of $\Theta(\sqrt{N}/\eps^2)$ is known~\cite{paninski2008coincidence,gr11}.  When studying distributions supported over extremely large domains, this sample complexity effectively makes testing intractable. To circumvent this problem, several competing approaches were considered.

The first approach involves restricting the class of input distributions (e.g., restricting the input distribution to be monotone~\cite{RS09} or a product distribution~\cite{baynets17,daskalakis2019testing}). The second approach considers a model equipped with a more relaxed distance metric (usually coupled with an even weaker query model), such as the Huge Object Model~\cite{GR23,chakraborty2023testing,adar2024refining,supp24,chakraborty2024testing}, which uses the earth-mover distance metric (as defined in those works). The third approach, which is the main focus of this work, investigates stronger query models.

One of the earliest models suggested to tackle the scaling problem is the \emph{conditional sampling} model. This model was introduced independently by Chakraborty, Fischer, Goldhirsh, and Matsliah~\cite{CFGM16}, and Canonne, Ron, and Servedio~\cite{ crs15}. The conditional model allows more general queries: namely, the algorithm may specify an arbitrary subset of the domain and request a sample from the distribution conditioned on it belonging to the subset. In many cases, the conditional sampling model circumvents sample-complexity lower bounds. Since its introduction, there has been significant study into the complexity of testing a number of properties of distributions under conditional samples, in both adaptive and non-adaptive settings~\cite{canonne2020survey, falahatgar2015faster, ACK15b, BCG19, BC18, KT19, FLV19}.
Beyond distribution testing, this model of conditional sampling has found applications in sublinear algorithms~\cite{GTZ17},  group testing~\cite{ACK15a}, and crowdsourcing~\cite{GTZ18}.

In this work we concentrate on a core task that is useful to many investigations of distribution testing. Consider a task that, given $x\in\{1,\ldots,N\}$, attempts to provide a multiplicative approximation of the probability of drawing $x$ according to the input distribution $\mu$, denoted by $\mu(x)$. It was first described as the \emph{evaluation oracle} in \cite{RS09}. If we are able to do it efficiently for all but a small probability set of the possible elements, then we can solve other tasks. Algorithms that simulate the multiplicative estimation task appear in \cite{crs15} and in \cite{CFGM16}. Both works independently define an implementation with $\poly(\log(N),1/\eps)$ sample-complexity of the evaluation oracle and use it to show their results (equivalence testing in \cite{crs15}, a universal tester for label-invariant properties in \cite{CFGM16}). Our main contribution is a radically improved algorithm for this task, that uses only $\log(\log(N))$ many samples, with a polynomial dependency on $\eps$ and an additional approximation parameter that we de-couple from $\eps$ and specify later.

As an example application, consider the task of \emph{distance estimation} in the conditional sampling model. In this task the algorithm receives (conditional) sampling access to two unknown distributions $\mu$ and $\tau$ supported over $\{1,\ldots, N\}$, and is required to estimate their total-variation distance within an additive error parameter. In the standard sampling model, tight bounds of $\Theta(N/\log N)$ are known even for estimating the distance from uniformity \cite{vv11, vv10a, vv10b}.

When considering the conditional sampling model, one can do much better. In \cite{crs15, subcubetolerant23}, an algorithm using $O(\poly(\log N)/\poly(\eps))$ conditional samples was established. Later, \cite{falahatgar2015faster} improved the complexity in the easier task of equivalence testing, which is distinguishing between zero distance and a distance greater than the approximation parameter, to $\tilde{O}(\log \log N / \eps^5)$, and \cite{chakraborty2024tight} obtained a corresponding lower bound of $\tilde{\Omega}(\log \log N)$.

In this paper we apply our improved distribution approximator to drastically improve the upper bound for estimating the distance between two unknown distributions using conditional samples to $O(\log \log N / \eps^2 + 1 / \eps^7)\cdot\poly\log(\eps^{-1})$. Based on our core approximator, we also improve the polynomial $\eps$-factors of \cite{falahatgar2015faster} and eliminate the polynomial triple-logarithmic $N$-factors. We also use this enahnced estimation module to approximate the histogram of an unknown distribution, being optimal up to poly-double-logarithmic $N$-factors and polynomial $\eps$-factors, and use it to obtain a universal tester for every label-invariant property at this cost.

To complement the picture, we show that there exists a label-invariant property that requires $\Omega(\log N / \eps)$ samples to test, which implies that the above-mentioned universal tester is optimal up to polynomial double logarithmic $N$-factors and polynomial $\eps$-factors. We also show that the core approximation task in itself is nearly optimal in its number of samples.

\subsection{Summary of our results}
\label{subsec:summary::of:intro}

\paragraph{Table of results} The following table summarizes our results, except for the two lower bounds marked by ``$(*)$'' which are due to \cite{chakraborty2024tight}. The following paragraphs provide more details. In the estimation task of $\mu(x)$, a correct output (with high probability) is guaranteed for every $x$ in some set $G \subseteq \Omega$ satisfying $\mu(G) > 1 - c$. This task has two sample-complexity upper bounds: the first holds for every $x$ in the domain of $\mu$ (even if it does not belong to $G$), and the second is the expected sample complexity when $x$ is unconditionally drawn from $\mu$.

\begin{center}
    \begin{tabularx}{\textwidth}{ |l|c|lY| }
        \hline
        \multicolumn{1}{|c|}{Task} & Lower Bound & \multicolumn{2}{c|}{Upper Bound} \\
        \hline
        \multirow{ 2}{\widthof{$\mu=\tau$ vs.\ $\dtv(\mu,\tau) > \eps$}}{Estimate $(1 \pm \eps)\mu(x)$} & \multirow{2}{*}{$\Omega(\log \log N)$} & 
        all $x$ & $O(\log \log N) + \tilde{O}(\frac{1}{\eps^2 c} + \frac{1}{\eps^5})$ \\
        & &  $x \sim \mu$ & $O(\log \log N) + O(\frac{1}{\eps^4}) \cdot \mathop\mathrm{polylog} (c^{-1}, \eps^{-1})$ \\
        \hline
        $\mu=\tau$ vs.\ $\dtv(\mu,\tau) > \eps$ & $(*)~\tilde{\Omega}(\log \log N)$ & & {$O\left(\frac{\log \log N}{\eps} + \frac{1}{\eps^5}\right) \cdot \mathop\mathrm{polylog} \eps^{-1}$} \\
        \hline
        Estimate $\dtv(\mu,\tau) \pm \eps$ & $(*)~\tilde{\Omega}(\log \log N)$ & & {$O\left(\frac{\log \log N}{\eps^2} + \frac{1}{\eps^7}\right) \cdot \mathop\mathrm{polylog} \eps^{-1}$} \\
        \hline
        Learn histogram of $\mu$ & $\Omega(\log N / \eps)$ & &
        \multirow{3}{*}{$\tilde{O}(1 / \eps^7) \cdot \log N \log \log N$} \\
        \cline{1-2}
        \multirow{2}{\widthof{$\mu=\tau$ vs.\ $\dtv(\mu,\tau) > \eps$}}{Label-invariant universal tester} & $\Omega(\log N / \eps)$ & & \\
        & (worst property) & & \\
        \hline
    \end{tabularx}
\end{center}

\paragraph{Single-element mass estimation} This is our core contribution. We use a new approach to show a tight upper bound for estimating the mass of a given element within $(1 \pm \eps)$-factor in the fully conditional (adaptive) model, where the mass of the set of eligible elements is at least $1-c$ for a second approximation parameter $c$ (all prior works implicitly set $c=O(\eps)$, which is sufficient for most applications). The old approaches (for example descending the tree of dyadic intervals as in \cite{CFGM16}) all have poly-logarithmic factors, which are unavoidable since they are also implementable in more restricted conditional models (such as interval conditioning \cite{crs15} and subcube conditioning \cite{BC18}) in which equivalence testing is known to be poly-logarithmic hard.

\begin{theorem}[Informal statement of Theorem \ref{th:estimation-task}] \label{th:summary:estimation-task}
    Let $\mu$ be a distribution over $\Omega = \{1,\ldots,N\}$ that is accessible through the fully conditional oracle, and let $\eps,c>0$ be the given approximation parameters. There exists a set $G \subseteq \Omega$ of mass $\mu(G) > 1 - c$ such that for every $x \in G$, we can algorithmically estimate $\mu(x)$ within $(1 \pm \eps)$-factor with probability $2/3$. The sample complexity is bounded by $O(\log \log N) + \tilde{O}\left(\frac{1}{\eps^2 c} + \frac{1}{\eps^5}\right)$. If $x$ is drawn from $\mu$, then the expected sample complexity is only $O(\log \log N) + O\left(\frac{1}{\eps^4}\right) \cdot \poly(\log c^{-1}, \log \eps^{-1})$.
\end{theorem}


The ``$O(\log\log N)$'' part comes from a binary search performed over a range of size $O(\log N)$. By a reduction from a binary search problem back to an estimation task, we also show that this part of the bound is tight (even if we only need to estimate a ``typical'' element from its support).

\begin{theorem}[Informal statement of Theorem \ref{th:lbnd-tight-c-eps-task}]
    There exists $\eps > 0$, so that every algorithm, that can with probability at least $2/3$ estimate the probability mass of an element drawn (unconditionally) from $\mu$ within $(1 \pm \eps)$-multiplicative factor, must draw at least $\Omega(\log \log N)$ conditional samples in expectation.
\end{theorem}

We leverage Theorem \ref{th:summary:estimation-task} to obtain the following upper-bound results.

\paragraph{Equivalence $\eps$-testing} We improve the state-of-the-art upper bound of \cite{falahatgar2015faster} that uses $\tilde{O}(\log \log N / \eps^5)$ for $\eps$-testing equivalence of two distributions over $\{1,\ldots,N\}$ in the fully conditional sampling model. We provide two asymptotical speedups: first, the $1/\eps^5$-factor becomes additive rather than multiplicative, whereas $\log \log N$ is only multiplied by $1/\eps$. This is an $\eps^4$-speedup over the previous best for $N$ that is large enough with respect to $\eps$. Second, we remove the poly-triple-logarithmic factors over $N$, leaving a clean $O(\log \log N)$ dependency on the domain size.

\begin{restatable}[Almost-tight upper bound for equivalence testing]{theorem}{thZubndZnontolZequivalence}
\label{th:ubnd-nontol-equivalence}
    Let $\mu$, $\tau$ be two distributions over $\Omega = \{1,\ldots,N\}$ and $\eps > 0$. There exists an algorithm for distinguishing between the case where $\mu=\tau$ and the case where $\dtv(\mu,\tau) > \eps$, using $O((\log \log N / \eps + 1/\eps^5) \cdot \poly(\log \eps^{-1}))$ conditional samples.
\end{restatable}

\paragraph{Distance estimation between two distributions} We show an almost-tight upper bound for estimating the total variation distance between two distributions over $\{1,\ldots,N\}$ in the fully conditional model. The surprising aspect of this result is the almost-tight gap between the lower bound of the equivalence testing problem and the upper bound of the (harder) distance estimation problem. Having only a small gap depending on the domain size between testing and estimation tasks is relatively uncommon, and in various models there exist examples for polynomial gaps (for example, uniformity in the sampling model \cite{paninski2008coincidence,vv10a, vv11}), and even examples for constant-cost tests with non-constant cost corresponding estimation tasks (for example, \cite{FF06} in the string model, later improved in \cite{pcpp19}).

\begin{restatable}[Almost-tight upper bound for distance estimation]{theorem}{thZubndZestZdtv}
\label{th:ubnd-est-dtv}
    Let $\mu$, $\tau$ be two distributions over $\Omega = \{1,\ldots,N\}$ and $\eps > 0$. There exists an algorithm for estimating $\dtv(\mu,\tau)$ within $\eps$-additive error using $O((\log \log N / \eps^2 + 1/\eps^7) \cdot \poly(\log \eps^{-1}))$ conditional samples.
\end{restatable}

Since the distance estimation task cannot be easier than equivalence testing, whose lower bound is $\tilde{\Omega}(\log \log N)$ \cite{chakraborty2024tight}, this upper bound for distance estimation is optimal, up to poly-triple-logarithmic factors of $N$ and polynomial $\eps$-factors.

\paragraph{Learning the histogram of $\mu$} We show an upper bound for learning the histogram of a given distribution $\mu$ up to a threshold parameter $\eps$ using a quasi-logarithmic in $N$ number of conditional samples. Since histogram learning directly implies testing of any label-invariant property without additional samples, this improves over \cite{CFGM16}, whose label-invariant tester only guarantees a polylogarithmic upper bound. A nearly-matching lower bound for this task follows from specially constructed label-invariant properties (see below).

\begin{theorem}[Informal statement of Theorem~\ref{th:ubnd-histogram}]
    There exists an algorithm that approximates the histogram of an input distribution $\mu$ with accuracy $\eps$, using $\tilde{O}(1/\eps^7) \cdot \log N \log \log N $ conditional samples.
\end{theorem}

\begin{theorem}[Informal statement of Corollary~\ref{cor:lbnd-histogram}]
    For every sufficiently small $\eps > 0$, every algorithm that approximates the histogram of its input distribution with accuracy $\eps$ must draw at least $\Omega(\log N / \eps)$ conditional samples.
\end{theorem}

\paragraph{Almost tight label-invariant testing} The histogram learning algorithm immediately implies a corresponding test for any label-invariant property: one can just perform the histogram approximation up to a distance of $\eps/2$, and then accept if and only if this histogram corresponds to a distribution that is $\eps/2$-close to satisfying the property. We complement this with an existence proof of label-invariant properties with a nearly matching lower bound on the number of required samples, an exponential improvement over the $\tilde{\Omega}(\log \log N)$ bound recently shown in \cite{chakraborty2024tight}.

\begin{theorem}[Informal statement of Corollary~\ref{cor:universal-tester-label-invariant}]
    There exists a universal tester for $\eps$-testing every label-invariant property using $\tilde{O}(1 / \eps^7) \cdot \log N \log \log N$ conditional samples.
\end{theorem}

\begin{theorem}[Informal statement of Theorem \ref{th:lbnd-label-invariant-testing-new}]
    For every sufficiently small $\eps > 0$, there exists a label-invariant property $\mathcal P$ such that every $\eps$-testing algorithm for $\mathcal P$ draws at least $\Omega(\log N / \eps)$ conditional samples.
\end{theorem}

\subsection{Related work}
Closely related to the distance estimation problem is the problem of \emph{equivalence testing}, which asks to determine whether two unknown distributions are equal or far from each other. In the standard sampling model, the sample complexity of the problem was pinned down to $\Theta(N^{2/3}/\eps^{4/3}+\sqrt{N}/\eps^2)$~\cite{batu-FRSW2000,valiant2008testing,cdvv14}. In the conditional sampling model, Canonne Ron and Servedio~\cite{crs15} designed a testing algorithm with query complexity $\tilde O (\log^5N/\eps^4)$. This was later improved to $\tilde O(\log \log N/\eps^5)$ by \cite{falahatgar2015faster}, and complemented with an almost matching lower bound of $\tilde\Omega(\log \log N)$ \cite{chakraborty2024tight}. That lower bound can be used for a relatively easy derivation of a $\tilde\Omega(\log\log N)$ lower bound on the $(c,\eps)$-estimation task for small enough $\eps>0$ and $c>0$, but we directly prove a clean $\Omega(\log\log N)$ bound.

One interesting special case of the distance estimation problem is the case where one of the distributions is explicitly given to the algorithm. In this setting, \cite{crs15} showed that one can estimate the distance to the known distribution using $\tilde O(\log^5N/\eps^5)$ conditional queries, which was later improved by ~\cite{narayanan2020distribution} to $\tilde{O}(1/\eps^4)$. In contrast, in the standard sampling model, estimating the distance to the uniform distribution requires at least $\Omega(N/\log N)$ samples~\cite{vv11, vv10a, vv10b}. 

A special type of conditional access which gained popularity in recent years is the \emph{subcube} conditioning model~\cite{BC18,crs15}. In this model, the distributions are given over a product set $\{0,1\}^n$, and the algorithm can query subcube subsets, which are sets of the form $\prod_{i=1}^n D_i$ where $D_i\subseteq \{0,1\}$ for every $1\le i\le n$ (note that here $N=2^n$).
In this model, uniformity can be tested using $\tilde{\Theta}(\sqrt{n}/\eps^2) = \tilde\Theta\left(\sqrt{\log N}/\eps^2\right)$ samples \cite{ccklw21}, and the best-known test for equivalence is $\tilde{O}(n/\eps^2)$ \cite{adar2024improved}, with a lower bound of $\Omega(n^{3/4} / \eps + \sqrt{n} / \eps^2)$~\cite{baynets17}. Other properties studied under the subcube conditional model include monotonicity~\cite{chakrabarty2025monotonicity}, and having a probability density function supported on a low-dimensional subspace~\cite{chen2021learning}.

A related line of work aims to circumvent the polynomial dependency in the domain's size by considering restricted classes of input distributions. Some of the cases studied are those where the distribution is known to be monotone~\cite{RS09,chakrabarty2025monotonicity}, a low-degree Bayesian Network \cite{baynets17,daskalakis2017square,acharya2018learning}, a Markov Random Field~\cite{daskalakis2019testing,gheissari2018concentration,bezakova2020lower}, or having a ``histogram by intervals'' structure~\cite{DKP19}. Considering a distribution having a histogram structure, a learning algorithm was given in \cite{FLV19} under several sampling models (for such distributions there is little difference between learning the histogram and learning the entire distribution).

\section{Overview}
\label{sec:technical-overview}

\subsection{Technical overview}
\subsubsection*{The core result}

Given a distribution $\mu$ over $\Omega = \{1,\ldots,N\}$, an element $x \in \Omega$ and two estimation parameters $\eps, c \in (0,1)$, our task is to estimate $\mu(x)$ within a $(1 \pm \eps)$-factor or to indicate that it is among the smallest elements, whose cumulative mass is at most $c$. Note that if $\mu(x) = \Omega(c)$ then it can easily be estimated directly using unconditional sampling, and therefore, this overview focuses on the case where $\mu(x) = O(c)$ (with an appropriate hidden constant factor).

At top level, our algorithm looks for a reference set $R$ whose probability mass (as an event) is both comparable to $\mu(x)$ (that is, $\Pr_\mu\left[x | R \cup \{x\}\right]$ lies in a reasonable range, between two constants) and estimable with high accuracy. This way we can arithmetically estimate $\mu(x)$ using estimations of $\Pr_\mu\left[x | R \cup \{x\}\right]$ and $\mu(R)$. For estimating $\Pr_\mu\left[x | R \cup \{x\}\right]$, it should be possible to efficiently draw samples from $R\cup\{x\}$. This can be directly done under the conditional model if we know $R$ in its entirety, but as we describe below there are ways around this problem when we do not have complete access to $R$.

Our construction refers to two sets: the \emph{target set} $V_x$, which is a set that includes all elements with mass smaller than $\mu(x)$ and no element whose mass is significantly higher than $\mu(x)$, and the \emph{filter set} $A_\alpha$, which is the result of independently choosing every $z \in \Omega$ with probability $\alpha$. The intersection $V_x \cap A_\alpha$ is a good reference set whenever the order of $\alpha$ is about the quotient of $\mu(x)$ and $\mu(V_x)$. The algorithm spends most of its effort on finding a good $\alpha$. In fact, if a good $\alpha$ is already given, then the rest of the algorithm can complete its estimation of $\mu(x)$ using a number of samples that does not depend on $N$ at all.

The target set $V_x$, whose mass is $\Omega(c)$ (for $x$ whose cumulative mass is at least $c$ and for which $\mu(x) = O(c)$), is estimable directly by virtue of having a high mass. Since it does not contain elements much heavier than $x$, we can use a large deviation inequality to deduce that the mass of $V_x \cap A_\alpha$ is highly concentrated around $\alpha\mu(V_x)$.

However, $V_x$ cannot be found explicitly. We can only construct a ``membership-oracle'' by comparing the weight of potential elements to the weight of $x$ using pair conditionals. In particular, $V_x$ is probabilistic, but for any element $y$ with strict demands (either lighter than $x$ or much heavier than $x$) there is a very small probability to misclassify the membership of $y$. For medium-weight elements, which are only slightly heavier than $x$, our analysis embraces the probabilistic nature of belonging to $V_x$.

Another problematic consequence of the oracle-membership implicit-construction of $V_x$ (and thereby of $R = V_x \cap A_\alpha$) is the inability to use it as a condition, since we can only restrict to explicit sets. Instead, we restrict to $A_\alpha$ (whose construction uses internal randomness and no samples from $\mu$), and use rejection-sampling to simulate the restriction to $V_x \cap A_\alpha$. Since $V_x$ has a globally high weight but contains no elements whose mass is too high, the relative weight of $V_x \cap A_\alpha$ as a subset of $A_\alpha$ is usually high as well. This allows a lazy construction of $V_x$, where we only query candidate elements drawn from $A_\alpha$ for belonging to $V_x$, instead of drawing $V_x$ in its entirety in advance.

To find an $\alpha$ of the correct magnitude, we first observe that it suffices to consider powers of $\frac{1}{2}$ that lie between $1$ and $\frac{1}{O(N)}$. This observation reduces the search range to $O(\log N)$ possible choices. To reduce the needed work to $O(\log \log N)$, we show a monotone estimable function of $\alpha$ that can characterize the range of good $\alpha$s based on their respective values of this function. This allows a binary search, but since the estimation of the function is probabilistic, we construct a binary search scheme that allows the comparator to be wrong with small fixed probability. Our binary search scheme removes the triple-logarithmic penalty required by the straight-forward approach of amplifying the success probability of the comparator to the point that even a single error is unlikely to occur during the binary search.

The result of the binary search is a $\Theta(1)$-approximation of the ideal choice of $\alpha$. Referring to $R$ constructed using such an $\alpha$, a simple arithmetical function of $\mu(R)/\mu(R\cup\{x\})$ (which can be approximated by inspecting a sequence of samples from $R\cup\{x\}$) gives us the missing factor that allows us to calculate the approximation of $\mu(x)$. Our procedure uses the roughly-estimated $\alpha$ to estimate this expectation within $1 \pm O(\eps)$-factor, which we then use to obtain a $(1 \pm \eps)$-factor estimation of $\mu(x)$.

\subsubsection*{The applications}

We present three applications of our core estimator. All of them are the result of plugging our estimator (each time with different parameters and under different circumstances) into an algorithm that achieves the corresponding task when it has some access to the actual values of the distribution function $\mu$.

For the task of histogram learning, knowing the exact value of $\mu(x)$ for each $x$ that was received as an unconditional sample would have allowed us to just approximate the weight of each ``bucket'' $B_i$ that contains all $x$ of weight between $(1-\eps)^{i-1}$ and $(1-\eps)^i$ (ignoring buckets with $i > O(\eps^{-1} \log N)$). From the bucket weights one can then write down a distribution that is an approximation of a permutation of $\mu$, providing the histogram of $\mu$. Receiving only approximate values can cause some ``bucket shift'' to (say) $i\pm 2$, but the resulting error is not significant.

Estimating the distance between two distributions, for which element mass estimations are provided, is generally achievable by (unconditionally) drawing elements from the distributions, and for each drawn element examining the ratio of its masses according to the two distributions. To implement this in the conditional sampling model, we plug in our estimator. Since we need to use it also for elements which were not drawn from the distribution to be queried, this increases the dependency on $\eps$ to one that follows from the ``all $x$'' bound. However, when we only want to solve the equivalence testing task, we show that we can still use the ``$x\sim \mu$'' bound as long as both distributions are identical to the same $\mu$, which allows us to automatically reject if the algorithm happens to require more samples than that bound. For both tasks, the lower bound is that of \cite{chakraborty2024tight}.

\subsubsection*{The lower bounds}

The tight lower bound for the estimation task is essentially an ad-hoc reduction from the task of finding an unknown value $k$, whose range of possible values has size $O(\log N)$, through binary search. Such a value can be ``encoded'' by a uniform distribution over a subset of $\{1,\dots,N\}$ whose size is $(1\pm o(1))2^{-k}\cdot N$, and then retrieved by successfully approximating the mass of any of the support elements. Since the bottleneck of the upper bound algorithm is a binary search task as well, this implies that a binary search task (in an appropriate range) is indeed a crucial component of the estimation task.

The demonstration that the full conditional model is essential for a doubly logarithmic algorithm follows from using the framework of some of our applications ``in the other direction'': a doubly logarithmic solution to the estimation task in a weak model would have implied a solution to a testing task that contradicts a known lower bound.

A label-invariant property with a logarithmic in $N$ lower bound is constructed by encoding maximally hard to test linear codes as histograms, and proving that a test for such a code can be converted to a classical string property test in this case.

\subsection{Organization of the paper}
Section \ref{sec:prelims} describes the sampling model and the notation scheme that we use throughout this paper. Within it, Subsection \ref{sec:prelims:subsec:paper-specific-notations} defines the quantities and constructed sets used for our algorithm. Appendix \ref{apx:notation-table} provides a concise table for these, for the reader's convenience.

Sections \ref{sec:our-algorithm} through \ref{sec:est-mu-x-using-alpha} contain the proof of our core result. Appendix \ref{apx:call-chart} provides a simplified chart of the calling structure and dependencies between the various procedures defined in these sections.

Section \ref{sec:our-algorithm} provides the top layer of the algorithm, which is the procedure \procnameZmainZsingleHREF.

Section \ref{sec:estimation-procs} provides the implementation of the target test scheme (\procnameZtargetZtestHREF in Subsection \ref{subsec:target-test::of-sec:estimation-procs}), along with a few algorithmic tools to assess the target sets, and specifically tools related to the quantity $\beta_{x,\alpha} = \mu(R)/\mu(R\cup\{x\})$: a ``cheaper'' \procnameZtargetZtestZgrossHREF in Subsection \ref{subsec:target-test::of-sec:estimation-procs}, a virtualization of the target set (\procnameZinitializeZnewZVxHREF, \procnameZVxZqueryHREF) in Subsection \ref{subsec:individual-drawing-V-x::of-sec:estimation-procs}, \procnameZestZexpectedZbeta in Subsection \ref{subsec:est-E-beta::of-sec:estimation-procs}, and \procnameZestZsingleZbetaHREF in Subsection \ref{subsec:est-individual-beta::of-sec:estimation-procs} (for a virtual single instance of $V_x$). The estimator here, while satisfying the optimal asymptotic guarantees, has an unrealistic numerical constant factor. In Appendix \ref{apx:target-test-galactic} we show how to reduce this by adjusting the target test, at the cost of a small asymptotic penalty which carries over to the estimator.

The following sections provide the three main components of the \procnameZmainZsingleHREF procedure: \procnameZpreambleHREF in Section \ref{sec:preamble}, \procnameZfindZgoodZalphaHREF in Section \ref{sec:find-good-alpha} and \procnameZestZexpectedZhZbetaHREF (an estimator for $\E\left[\frac{\beta_{x,\alpha}}{1-\beta_{x,\alpha}}\right]$) in Section \ref{sec:est-mu-x-using-alpha}. The estimator for $\E\left[\frac{\beta_{x,\alpha}}{1 - \beta_{x,\alpha}} \right]$, whose precise definition is $\E_{A_\alpha, V_x}\left[\frac{\beta_{x,\alpha}\left(A_\alpha, V_x\right)}{1 - \beta_{x,\alpha}\left(A_\alpha, V_x\right)}\right]$, uses the procedures for draws of $V_x$ and estimations of $\beta_{x,\alpha}$ (for a provided $A_\alpha$) of Section \ref{sec:estimation-procs}.

We provide the uncertain-comparator binary search used in \procnameZfindZgoodZalphaHREF, which is a probabilistic variant of binary search that can use a comparator which is allowed to be wrong with a small fixed probability, in Subsection \ref{subsec:binary-search::of-sec:find-good-alpha}.

In Section \ref{sec:applications} we provide applications of our core result in the fully conditional model: in Subsection \ref{subsec:learning-histograms::of-sec:applications} we provide an almost-tight histogram learning algorithm, and as a corollary a universal $\eps$-testing algorithm for label-invariant properties. In Subsection \ref{subsec:dtv-est::of-sec:applications} we provide an almost tight $\eps$-estimation for total-variation distance, and in Subsection \ref{subsec:eps-test-equivalence::of-sec:applications} we provide an improved $\eps$-test for equivalence. These applications use a trio of general-purpose application lemmas, Lemmas \ref{lemma:generic-application-single-distribution}, \ref{lemma:generic-application-k-distributions} and \ref{lemma:generic-application-k-distributions-testing}. In Appendix \ref{apx:one-more-lemma} we provide another lemma of this type, for which we hope to be useful in future applications.

In Section \ref{sec:lbnds} we provide lower bounds for the tasks discussed in this paper. In Subsection \ref{subsec:tight-lbnd-c-eps-estimation::of-sec:lbnds} we prove the tight lower bound on the $\mu(x)$ estimation task (as a function of $N$). Then, in Subsection \ref{subsec:lbnd-c-eps-est::of-sec:lbnds} we provide quick lower bounds for this task under more restricted conditional models, mostly derived from known lower bounds on equivalence testing in conjunction with interim algorithms from Section \ref{sec:applications}. In Subsection \ref{subsec:lbnd-lbl-inv::of-sec:lbnds} we construct a specific label-invariant property, for which we prove an almost-tight testing lower bound.

The most technical (and mechanical) proofs across the paper are deferred to Appendix \ref{apx:tech-SHEET} (bounds relevant to $\E[\beta_{x,\alpha}]$ and $\E[\beta_{x,\alpha}/(1 - \beta_{x,\alpha})]$) and Appendix \ref{apx:tldr} (miscellaneous ad-hoc proofs).

\section{Preliminaries}
\label{sec:prelims}

\subsection{Distribution access oracles and tasks}

In this work we consider algorithms that access an input \emph{indirectly} through oracles. In particular, the complexity of our algorithms is measured in terms of the number of calls to the provided oracle. In all oracles, the output distribution is determined by the distribution $\mu$ and the arguments of the call, and is completely independent of past calls and any other algorithmic behavior.

The following is the weakest oracle, the one allowed in the traditional distribution testing model.

\begin{definition}[Sampling oracle]
    Let $\mu$ be an input distribution over a set $\Omega$. The \emph{sampling oracle} for $\mu$ has no additional input, and outputs an element $x \in \Omega$ that distributes like $\mu$.
\end{definition}

The following oracle corresponds to the algorithms that we analyze in this paper.
\begin{definition}[Conditional sampling oracle]
    Let $\mu$ be an input distribution over a set $\Omega$. The \emph{conditional sampling oracle} for $\mu$ gets a set $A \subseteq \Omega$ as input, and outputs an element $x\in A$ that distributes like $\mu$ when conditioned on $x$ belonging to $A$.
\end{definition}

The above definition still leaves open the question as to what happens when a sample conditioned on $A$ is requested for a probability zero set $A$ (two common variants are the oracle returning a special error symbol \cite{crs15} or the oracle returning a value uniformly drawn from $A$ \cite{CFGM16}). Our estimation algorithm is designed to never ask for such a sample, and hence works for all variants of this model.

We next define some notions of distances and approximations.

\begin{definition}[Total variation distance]
    Let $\mu$ and $\tau$ be two distributions over $\Omega$. Their \emph{total variation distance} is defined as
    \[ \dtv(\mu, \tau)
    \overset{\mathrm{def}}= \frac{1}{2} \sum_{x \in \Omega} \abs{\mu(x) - \tau(x)}
    = \sum_{x \in \Omega : \mu(x) > \tau(x)} (\mu(x) - \tau(x))
    = \max_{E \subseteq \Omega} \left(\mu(E) - \tau(E)\right) \]
\end{definition}

\begin{definition}[$\eps$-test]
    Let $\mathcal R$ be a metric space and let $\mathcal P$ be a closed set, which we call a \emph{property}. For an input element $x \in \mathcal R$ and a parameter $\eps > 0$, the goal of an \emph{$\eps$-test} for $\mathcal P$ is to distinguish between the case where $x \in \mathcal P$ and the case where $d(x,y) > \eps$ for every $y \in \mathcal P$.
\end{definition}

The following notion is a major building block in our algorithms.

\begin{definition}[Saturation-aware estimator]
    Let $f : [0,1] \to [0,1]$ be a non-decreasing monotone function. An algorithm is an $(\eps;p_\ell,p_m)$-$f$-saturation-aware estimator of an unknown probability $p$ if the following hold:
    \begin{itemize}
        \item If $f(p) \le p_\ell$, then with probability at least $2/3$ the output is a special value \esttoolow.
        \item If $p_\ell < f(p) < p_m$, then with probability at least $2/3$ the output is either the special value \esttoolow or in the range $(1 \pm \eps)p$.
        \item If $f(p) \ge p_m$, then with probability at least $2/3$ the output is in the range $(1 \pm \eps)p$.
    \end{itemize}
    We use ``$(\eps;p_\ell,p_m)$-saturation-aware estimator'' to denote the case where $f$ is the identity function.
\end{definition}
Usually, a test can distinguish between the first and the third cases, and an estimator can guarantee the third part. The saturation-aware estimator also requires correctness in the ``middle'' case.

Recall that our main goal is to estimate the probability mass of individual elements. The exact mass cannot be provided since it can be any value in the continuous range $[0,1]$. Moreover, the effort needed to estimate the mass of an extremely rare element is unbounded. Hence, every estimation algorithm must allow a small mass of elements whose probabilities cannot be approximated at all, and in this paper these are characterized by the notion of the \emph{cumulative distribution function}.

\begin{definition}[Cumulative distribution function]
    Let $\mu$ be a distribution over $\Omega$. The \emph{cumulative distribution function} of $\mu$ is the function $\CDF_\mu : \Omega \to [0,1]$ defined as $\CDF_\mu(x) = \Pr_{y \sim \mu}[\mu(y) \le \mu(x)]$.
\end{definition}

\begin{definition}[The $(\eps,c)$-estimation task] \label{def:eps-c-estimation-task}
    Let $\eps>0$ and $c>0$ be our parameters. For a distribution $\mu$ over a finite domain $\Omega$, let $\mathcal A$ be an algorithm that gets $x \in \Omega$ and outputs some $\hat p$. The goal is an $(\eps;0,c)$-$\CDF_\mu$-saturation-aware estimation of $\mu(x)$.
\end{definition}

Our main result is an algorithm that solves the estimation task, whose dependency on $N$ is doubly-logarithmic for fixed $\eps$ and $c$. Note that in particular the set $G=\{x\in\Omega:\CDF_\mu(x)\geq c\}$ has mass strictly larger than $1-c$. We next describe a major application of our estimator.

\begin{definition}[Histogram divergence $\dhist(\cdot ; \cdot)$] \label{def:histogram-divergence}
    Let $\mu$ and $\tau$ be two distributions over $\Omega$, and let $S(\Omega)$ denote the set of all permutations over $\Omega$. The \emph{histogram divergence} of $\mu$ and $\tau$ is defined as:
    \[ \dhist\left(\mu ; \tau\right) = \min \{ \eps \ge 0 : \min_{\pi\in S(\Omega)}\; \Pr_{x \sim \mu} \left[ \mu(x) \notin (1 \pm \eps) \tau(\pi(x)) \right] \le \eps \} \]
\end{definition}

The following lemma, which we prove in Appendix \ref{apx:tldr}, states that distributions with low histogram divergence are close up to a permutation of the labels.

\begin{restatable}{lemma}{lemmaZdtvZpermutationZleZtwiceZdhist} \label{lemma:dtv-permutation-le-twice-dhist}
    For every two distributions $\mu$, $\tau$ over $\Omega$ there exists a permutation $\pi$ over $\Omega$ for which $\dtv(\mu, \pi \tau) \le 2 \dhist(\mu ; \tau)$.
\end{restatable}

\begin{definition}[The $\eps$-histogram learning task] \label{def:eps-histogram-learning-task}
    Let $\mu$ be a distribution over $\Omega$. The \emph{$\eps$-histogram learning task} requires finding a distribution $\tau$ for which $\dhist(\mu ; \tau) \le \eps$.
\end{definition}

\subsection{Paper-specific notations}
\label{sec:prelims:subsec:paper-specific-notations}
To describe our estimation algorithm we need various ad-hoc notations. Most of them involve $x \in \Omega$, $c \in (0,1)$ and $\eps \in (0,1)$, and some additionally involve $0 < \alpha \le 1$. We usually use short-form notations ignoring $c$ and $\eps$, but never ignore $x$ and $\alpha$. See Appendix \ref{apx:notation-table} for a concise table that summarizes these notations.

Given $x\in\Omega$ for which we would like to assess $\mu(x)$, our proofs rely heavily on a categorization of $\Omega\setminus\{x\}$ by masses.

\begin{definition}[The three scale-sets with respect to $x$]\label{def:LMH-x}
    Let $x \in \Omega$. We divide the rest of the domain $\Omega$ according to their probability masses as compared to $\mu(x)$ as follows:
    \begin{itemize}
        \item The $x$-light set is $L_x = \{ y \in \Omega \setminus \{x\} : \mu(y) \le \mu(x) \}$.
        \item The $x$-medium set is $M_x = \{ y \in \Omega : \mu(x) < \mu(y) < 1.2 \mu(x) \}$.
        \item The $x$-heavy set is $H_x = \{ y \in \Omega : \mu(y) \ge 1.2 \mu(x) \}$.
    \end{itemize}
\end{definition}

Distinguishing between $L_x$-elements and $H_x$-elements cannot be certain since it uses random samples. We require the probability to be very small. The affect of the target error on our algorithm's complexity is logarithmic.

\begin{definition}[$\eta_{c,\eps}$, the target error] \label{def:target-error}
    The \emph{target error} is $\eta_{c,\eps} = \min\left\{ \paperZtargeterrZpreamble, \paperZtargeterrZevbeta, \paperZtargeterrZestsingle \right\}$.
\end{definition}

We define the constraints of the categorization algorithm based on the target error.

\begin{definition}[An $(x,c,\eps)$-target test] \label{def:target-test}
    An algorithm $\mathcal T$ is an \emph{($x,c,\eps$)-target-test} if:
    \begin{itemize}
        \item The probability to accept $y \in \Omega \setminus \{x\}$ only depends on $x$ and $y$ (and $\mu$), and in particular is independent of past executions.
        \item For every $y \in L_x$, the probability to accept $y$ is greater than $1 - \eta_{c,\eps}$.
        \item For every $y \in H_x$, the acceptance probability is less than $\eta_{c,\eps}$.
    \end{itemize}
\end{definition}

\begin{definition}[An $(x,c,\eps)$-target-test scheme] \label{def:target-test-scheme}
    A mapping from every triplet $(x,c,\eps)$ to an $(x,c,\eps)$-target-test $\mathcal T_{x,c,\eps}$ is called an \emph{$(x,c,\eps)$-target-test scheme}.
\end{definition}

The target-test scheme is a crucial part in our proof, and can be implemented algorithmically.

\begin{lemma}[Informal statement of Lemma \ref{lemma:alg:target-function-correct}]
    Procedure \procnameZtargetZtest (Algorithm \ref{fig:alg:target-test} in Section \ref{sec:estimation-procs}), with parameters $(\mu,c,\eps;x,y)$, is an $(x,c,\eps)$-target-test scheme.
\end{lemma}

We prove this lemma in Section \ref{sec:estimation-procs}, right after the implementation of \procnameZtargetZtest. Based on this lemma, we have a canonical target-test for every triplet $(x,c,\eps)$, and can omit the recurring parameter of ``the specific $(x,c,\eps)$-test we refer to'', encapsulating it as a function.

\begin{definition}[$f_{x,c,\eps}$, Target function] \label{def:target-function}
    The target function $f_{x,c,\eps} : \Omega \setminus \{x\}$ is the probability of the (canonical) $(x,c,\eps)$-target-test to accept $y$.
\end{definition}

The ``ideal reference set'' is defined similarly to the output of the target test, only here we do not allow any error with respect to the ``important'' sets $L_x$ and $H_x$.

\begin{definition}[$V_{x,c,\eps}$, the target set] \label{def:V-x}
    Let $x \in \Omega$. The target set $V_{x,c,\eps} \subseteq \Omega \setminus \{x\}$ is a random set that contains $L_x$, is disjoint from $H_x$, and contains every element $y \in M_x$ with probability $f_{x,c,\eps}(y)$, independently.
\end{definition}

We define two masses based on the expected mass of the target set corresponding to $x$, with or without $x$ itself. Both masses play a role.

\begin{definition}[$s_{x,c,\eps}$, the scale mass] \label{def:s-x}
    The scale mass with respect to $x$ is denoted by $s_{x,c,\eps} = \E[\mu(V_{x,c,\eps})] = \mu(L_x) + \sum_{y \in M_x} \mu(y) f_{x,c,\eps}(y)$. The expectation is over the choice of $V_{x,c,\eps}$ as a random set.
\end{definition}

\begin{definition}[$w_{x,c,\eps}$, the weight of $x$] \label{def:w-x}
    The weight of $x$ is denoted by $w_x = \mu(x) + s_{x,c,\eps}$.
\end{definition}

As mentioned in the technical overview, we also define a filter set that is randomly constructed according to a parameter $\alpha$. Since the filter set only depends on the internal coin-tosses of the algorithm, we can fully characterize it and use it for drawing conditional samples. Our reference set is the intersection of the target set and the filter set.

\begin{definition}[$A_\alpha$, the $\alpha$-filter set] \label{def:A-alpha}
    Let $0 < \alpha \le 1$. The $\alpha$-filter set, $A_\alpha$, is a random set where every element in $\Omega$ belongs to $A_\alpha$ with probability $\alpha$, independently.
\end{definition}

\begin{definition}[$V_{x,c,\eps,\alpha}$, the $\alpha$-filtered target set] \label{def:V-x-alpha}
    The $\alpha$-filtered target set is denoted by $V_{x,c,\eps,\alpha} = V_{x,c,\eps} \cap A_\alpha$.
\end{definition}

The next definition, $\gamma_{x,c,\eps}$, describes the ``best'' value of $\alpha$ with respect to $x$, $c$ and $\eps$, which our algorithm looks for. The rest of the algorithm works by first finding a suitable $\alpha=\Theta(\gamma_{x,c,\eps})$

\begin{definition}[$\gamma_{x,c,\eps}$, the goal magnitude] \label{def:gamma-x}
    The goal magnitude of the filtering parameter $\alpha$ is denoted by $\gamma_{x,c,\eps} = \mu(x) / \E[\mu(V_{x,c,\eps})]$.
\end{definition}

For a good $\alpha$, the distribution of the following probability is concentrated around a value bounded away from both $0$ and $1$.

\begin{definition}[$\beta_{x,c,\eps,\alpha}$, the filtered density] \label{def:beta-x-alpha}
    Let $x \in \Omega$ and $0 < \alpha \le 1$. The filtered density of $x$, with respect to the choices of $V_{x,c,\eps}$ and $A_\alpha$, is $\beta_{x,c,\eps,\alpha} = \Pr_\mu[\neg x | V_{x,c,\eps,\alpha} \cup \{x\}] = \frac{\mu(V_{x,c,\eps,\alpha})}{\mu(x) + \mu(V_{x,c,\eps,\alpha})}$.
\end{definition}

The algorithm first finds a good $\alpha$ by performing a binary search, where for values too close to $0$ or $1$ the expectation of $\beta_{x,c,\eps,\alpha}$ are in the ``too low'' and ``too high'' ranges respectively. A good $\alpha$ allows us to complete our assessment of $\mu(x)$ using $\beta_{x,c,\eps,\alpha}$.

\begin{observation} \label{obs:key-observation}
    For every $0 < \alpha \le 1$, $\mu(x) = \alpha \E[\mu(V_{x,c,\eps})] / \E\left[\frac{\beta_{x,c,\eps,\alpha}}{1 - \beta_{x,c,\eps,\alpha}}\right]$.
\end{observation}

\begin{proof} For every fixed $x$,
    \[
        \E\left[\frac{\beta_{x,c,\eps,\alpha}}{1 - \beta_{x,c,\eps,\alpha}}\right]
        = \E\left[\frac{\frac{\mu(V_{x,c,\eps,\alpha})}{\mu(x) + \mu(V_{x,c,\eps,\alpha})}}{1 - \frac{\mu(V_{x,c,\eps,\alpha})}{\mu(x) + \mu(V_{x,c,\eps,\alpha})}}\right]
        = \E\left[\frac{\mu(V_{x,c,\eps,\alpha})}{\mu(x)}\right]
        = \frac{\E\left[\mu(V_{x,c,\eps,\alpha})\right]}{\mu(x)}
        = \frac{\alpha \E\left[\mu(V_{x,c,\eps})\right]}{\mu(x)}
    \qedhere \]
\end{proof}


\subsection{Technical lemmas}

\begin{definition}[Binomial distribution, $\Bin(n,p)$]
    The distribution of the sum of $n$ independent trials with success probability $p$ is denoted by $\Bin(n,p)$. Explicitly, $\Pr_{\mathrm{Bin}(n,p)}[k] = \binom{n}{k} p^k (1-p)^{n-k}$.
\end{definition}

\begin{lemma}[Additive Chernoff bound]
    Let $X \sim \Bin(n,p)$. For every $t > 0$, $\Pr[X - \E[X] > t] \le e^{-2 t^2/n}$ and $\Pr[X - \E[X] < -t] \le e^{-2 t^2/n}$.
\end{lemma}

\begin{lemma}[Multiplicative Chernoff bound]
    Let $X \sim \Bin(n,p)$. For every $0 < r \le 1$, $\Pr[X > (1 + r) \E[X]] \le e^{-\frac{1}{3} r^2 \E[X]}$ and $\Pr[X < (1 - r) \E[X]] \le e^{-\frac{1}{3} r^2 \E[X]}$.
\end{lemma}

\begin{definition}[Geometric distribution, $\Geo(p)$]
    The distribution of the number of independent trials, each with success probability $p$, until the first success (including the successful trial itself) is denoted by $\Geo(p)$. Explicitly, $\Pr_{\Geo(p)}[k] = (1 - p)^{k-1} p$.
\end{definition}
\begin{lemma}[Well-known] \label{lemma:geo-mgf}
    Let $X$ be a random variable that is geometrically distributed with parameter $p$. Then $\E[X] = p^{-1}$, $\Var[X] = \frac{1-p}{p^2}$, and $\E[e^{\lambda X}] = \frac{p e^\lambda}{1 - (1-p)e^\lambda}$ for $\lambda < -\ln (1-p)$.
\end{lemma}

\begin{observation}[Folklore] \label{obs:non-decreasing-set-monotone}
    Let $A_\alpha \subseteq \Omega$ be a random set such that, given $\alpha$, every element $y \in \Omega$ belongs to $A_\alpha$ with probability $p_{y,\alpha}$, where $p_{y,\alpha}$ is non-decreasing monotone with respect to $\alpha$ (but possibly not the same for different choices of $y$). Let $f : 2^\Omega \to \mathbb R$ be a non-decreasing monotone function (that is, $U \subseteq V \to f(U) \le f(V)$). In this setting, the mapping $\alpha \to \E[f(A_\alpha)]$ is non-decreasing monotone as well.
\end{observation}

\begin{observation}[Generic]  \label{obs:set-exp-value-continuity}
    Let $f : 2^\Omega \to [a, b]$ be a bounded function. Assume that $A_\alpha \subseteq \Omega$ is drawn such that every element $y \in U$ is drawn with probability $p_{y,\alpha}$, which is continuous with respect to a parameter $\alpha$ (but possibly not the same for different choices of $y$). The mapping $\alpha \to \E[f(A_\alpha)]$ is continuous.
\end{observation}
\begin{proof}
    Let $\alpha_1 < \alpha_2$. Then:
    \[\abs{\E[f(A_{\alpha_2})] - \E[f(A_{\alpha_1})]}
        \le \left (\max_U f(U)-\min_U f(U)\right) \sum_{y \in \Omega} \abs{p_{y,\alpha_2} - p_{y,\alpha_1}}
        \le (b - a) \cdot \sum_{y \in \Omega} \abs{p_{y,\alpha_2} - p_{y,\alpha_1}}
    \]
    This expression tends to zero for $\alpha_2 \to \alpha_1$ since all $p_{y,\alpha}$s are continuous and $a$ and $b$ are fixed.
\end{proof}

\begin{restatable}[Median amplification]{observation}{obsZmedianZamplification} \label{obs:median-amplification}
    Let $X$ be a random variable, and $[a,b]$ be a range such that $\Pr\left[X \in [a, b]\right] \ge 2/3$. We use ``median-of-$M$'' to denote the process of drawing $M$ independent samples of $X$ and taking their median value. Then:
    \begin{enumerate}[label=(\alph*)]
        \item Median-of-$9$ amplifies the probability of obtaining a value in $[a,b]$ to $5/6$. \label{median-amplification:9-to-5/6}
        \item Median-of-$13$ amplifies it to $8/9$. \label{median-amplification:13-to-8/9}
        \item Median-of-$47$ amplifies it to $99/100$. \label{median-amplification:47-to-99/100}
        \item Median-of-$\ceil{30 \ln c^{-1}}$ amplifies it to $1 - \frac{1}{2}c$ for $c < 1/3$. \label{median-amplification:log-to-2c}
        \item Median-of-$\ceil{30 \ln c^{-1}}$ amplifies it to $1 - \frac{1}{24}c$ for $c < 1/150$. \label{median-amplification:log-to-24c}
    \end{enumerate}
\end{restatable}

We prove Observation \ref{obs:median-amplification} in Appendix \ref{apx:tldr}.

\begin{lemma} \label{lemma:one-over-cdf-sum-bound}
    Let $r > 0$. For every distribution $\mu$, $\E_{x \sim \mu}\left[\frac{1}{w_x + r}\right] \le \E_{x\sim \mu}\left[\frac{1}{\CDF_\mu(x) + r}\right] = O(\log r^{-1})$.
\end{lemma}
\begin{proof}
    By definition, $w_x \ge \CDF_\mu(x)$ for every $x \in \Omega$.
    \begin{eqnarray*}
        \E\left[\frac{1}{w_x + r}\right]
        \le \E\left[\frac{1}{\CDF_\mu(x) + r}\right]
        &\le& \sum_{x \in \Omega} \mu(x) \cdot \frac{1}{\max\{\CDF_\mu(x),r\}} \\
        &\le& \Pr[\CDF_\mu(x) \le r] \cdot \frac{1}{r} + \sum_{t=0}^{\floor{\log_2 r^{-1}}} \Pr[\CDF_\mu(x) \le 2^{-t}] \cdot 2^t  \\
        &\le& r \cdot \frac{1}{r} + \sum_{t=0}^{\floor{\log_2 r^{-1}}} 2^{-t} \cdot 2^t \\
        &=& 1 + \sum_{t=0}^{\floor{\log_2 r^{-1}}} 1
        = O(\log r^{-1})
    \end{eqnarray*}
\end{proof}

Due to the length of some expressions in our proofs, we use here the contribution notation introduced in \cite{supp24}:
\begin{definition}[Contribution of $X$ over $B$] \label{def:contribution}
    Let $X$ be a random variable and $B$ be an event. We denote the \emph{contribution of $X$ over $B$} by $\Ct[X | B] = \sum_{x \in B} \Pr[x] \cdot X(x) = \Pr[B] \E[X | B]$.
\end{definition}

We quickly summarize some equalities of the contribution notation:
\begin{itemize}
    \item $\Ct[\alpha X + \beta Y | B] = \alpha \Ct[X | B] + \beta \Ct[Y | B]$.
    \item If $B_1 \cap B_2 = \emptyset$ then $\Ct[X | B_1 \cup B_2] = \Ct[X | B_1] + \Ct[X | B_2]$.
    \item If $\Pr[X = Y | B] = 1$ then $\E[X] - \E[Y] = \Ct[X - Y | \neg B]$.
\end{itemize}

\section{Our algorithm}
\label{sec:our-algorithm}

Our upper-bound statements assume that $\eps$ and $c$ are ``sufficiently small''. More concretely, $\eps < \frac{1}{10}$ and $c < \frac{1}{16}$. The acronym ``SP'' appearing in some of the algorithms refers to ``Success Probability''.

We state the main theorem of this paper, which refers to the correctness of the procedure \procnameZmainZsingle.

\begin{theorem} \label{th:estimation-task}
    For every individual $x \in \Omega$, Algorithm \ref{fig:alg:estimation-task} solves the $(\eps,c)$-estimation task with expected sample complexity $O(\log \log N) + O\left(\log \frac{1}{\eps c} \cdot \left(\frac{1}{\eps^2 (w_x + c)} + \frac{\log^5 \eps^{-1}}{\eps^4\left(w_x + \eps / \log \eps^{-1}\right)}\right)\right)$ (the expectation is over the random choices of the algorithm), where $N = \abs{\Omega}$ is the size of the domain of $\mu$.
\end{theorem}

\begin{corollary} \label{cor:worst-case-complexity-of::th:estimation-task}
    The expected complexity of Algorithm \ref{fig:alg:estimation-task} is $O(\log \log N) + O\left(\log \frac{1}{\eps c} \!\cdot\! \left(\frac{1}{\eps^2 c} + \frac{\log^6 \eps^{-1}}{\eps^5}\right)\right)$ for the worst-case choice of $x \in \Omega$.
\end{corollary}
\begin{proof}
    The worst case is trivially $w_x = 0$.
\end{proof}

\begin{corollary} \label{cor:expected-complexity-of::th:estimation-task}
    The expected complexity of Algorithm \ref{fig:alg:estimation-task}, where $x$ is the result of an unconditional sample from $\mu$, is $O(\log \log N) + O\left(\log \frac{1}{\eps c} \cdot \left(\frac{\log c^{-1}}{\eps^2} + \frac{\log^6 \eps^{-1}}{\eps^4}\right)\right)$.
\end{corollary}
\begin{proof}
    By Lemma \ref{lemma:one-over-cdf-sum-bound}, the expected value of $\frac{1}{w_x + c}$ is bounded by $O(\log c^{-1})$ and the expected value of $\frac{1}{w_x + \eps/\log \eps^{-1}}$ is $O(\log \eps^{-1})$.
\end{proof}

The algorithmic demonstration of Theorem \ref{th:estimation-task} (Algorithm \ref{fig:alg:estimation-task} below) relies on three core subroutines, whose interface is stated in the following lemmas.

The first lemma, proved in Section \ref{sec:preamble}, provides an estimation of the expected mass of the target set. Additionally, for the edge-case of elements with very high mass, it estimates this mass directly.

\begin{restatable}[\procnameZpreamble]{lemma}{lemmaZmainZalgZpreamble} \label{lemma:main-alg-preamble}
    For every $x \in \Omega$, Algorithm \ref{fig:alg:est-mu-x-s-x} is a joint estimator of $(\mu(x), s_x)$ which is:
    \begin{itemize}
        \item An $\left(\eps; \max\left\{\frac{1}{400}c,\frac{1}{400}s_x\right\}, \max\left\{c,\frac{1}{4}s_x\right\}\right)$-saturation-aware estimator for $\mu(x)$.
        \item A $\left(\frac{1}{3}\eps; \max\left\{\frac{1}{400}c,\frac{1}{400}\mu(x)\right\}, \max\left\{c,\frac{1}{4}\mu(x)\right\}\right)$-saturation-aware estimator for $s_x$.
    \end{itemize}
    Its expected cost is $O\left(\log \frac{1}{\eps c} \cdot \frac{1}{\eps^2 (w_x + c)}\right)$ samples.
\end{restatable}

The second lemma, proved in Section \ref{sec:find-good-alpha}, provides us with a parameter $\alpha$ that would be usable for comparing $\mu(x)$ with the mass of the filtered target set $V_{x,\alpha}$.

\begin{restatable}[\procnameZfindZgoodZalpha]{lemma}{lemmaZmainZalgZfindZgoodZalpha} \label{lemma:main-alg-find-good-alpha}
    Assume that $\mu(x) \le \frac{1}{4}s_x$. The output of Algorithm \ref{fig:alg:find-good-alpha-new} is a random variable $\alpha$ for which, with probability $2/3$, $\gamma_x \le \alpha \le 41\gamma_x$, at the cost of $O(\log \log N)$ samples at worst case.
\end{restatable}

The third lemma, proved in Section \ref{sec:est-mu-x-using-alpha}, produces the relative estimation of the ratio of $\mu(x)$ to the expected mass of the filtered target set, provided we have started with a suitable $\alpha$ (essentially a very rough approximation of the ratio between $\mu(x)$ and the scale mass).

\begin{restatable}[\procnameZestZexpectedZhZbeta]{lemma}{lemmaZmainZalgZuseZalpha} \label{lemma:main-alg-use-alpha}
    Let $0 < \alpha \le 1$ be an explicitly given input, and assume that $\gamma_x \le \alpha \le 50\gamma_x$. The output of Algorithm \ref{fig:alg:est-h-beta}  is a random variable whose value, with probability $2/3$, is $(1 \pm \eps/2)\alpha s_x / \mu(x)$, at the expected cost of $O\left(\log \frac{1}{\eps c} \cdot \frac{\log^5 \eps^{-1}}{\eps^4 \left(w_x + {\eps}/{\log \eps^{-1}}\right)}\right)$ samples.
\end{restatable}

At this point we provide Algorithm \ref{fig:alg:estimation-task}, and prove its correctness, which implies Theorem \ref{th:estimation-task}.

\begin{algo}[ht!]
    \label{fig:alg:estimation-task}
    \procname{$\procnameZmainZsingle(\mu,c,\eps;x)$}
    \alginput{$x$}
    \algoutput{$\hat{p} \in \{\esttoolow\} \cup (1 \pm \eps)\mu(x)$}
    \begin{code}
        \algitem Let $M \gets 13$. \algcomment For median amplification (Observation \ref{obs:median-amplification}\ref{median-amplification:13-to-8/9}).
        \begin{For}{$i$ from $1$ to $M$}
            \algitem Let $(\hat{p}_i, \hat{s}_i) \gets \procnameZpreambleHREF(\mu,c,\eps;x)$.
        \end{For}
        \algitem Set $\hat{p} \gets \median(\hat{p}_1,\ldots,\hat{p}_M)$ and $\hat{s} \gets \median(\hat{s}_1,\ldots,\hat{s}_M)$. \algcomment SP: $\frac{8}{9}$
        \begin{If}{$\hat{p} = \esttoolow$ and $\hat{s} \ne \esttoolow$}
            \begin{For}{$i$ from $1$ to $M$}
                \algitem Let $\alpha_i \gets \procnameZfindZgoodZalphaHREF(\mu,c,\eps;x)$.
            \end{For}
            \algitem Let $\alpha \gets \median(\alpha_1,\ldots,\alpha_M)$. \algcomment SP: $\frac{8}{9}$
            \begin{For}{For $i$ from $1$ to $M$}
                \algitem Let $\hat{b}_i \gets \procnameZestZexpectedZhZbetaHREF(\mu, c, \eps; x)$.
            \end{For}
            \algitem Let $\hat{b} \gets \median(\hat{b}_1,\ldots,\hat{b}_M)$. \algcomment SP: $\frac{8}{9}$
            \algitem Set $\hat{p} \gets \alpha \hat{s} / \hat{b}$.\label{fig:alg:estimation-task::step:set-p-hat}
        \end{If}
        \item Return $\hat{p}$. \algcomment \textbf{Total success probability: $\mathbf{\frac{2}{3}}$}
    \end{code}
\end{algo}

\begin{proof}[Proof of Theorem \ref{th:estimation-task}]
    For the definition of the $(\eps,c)$-estimation task with the required saturation-awareness bounds $p_1$ and $p_2$, let $p_1 = \min_x \!\left(\{c\} \cup \left\{ \mu(x) : w_x \!\ge\! c \right\}\right)$, and $p_2 = \max_x \!\left\{ \mu(x) : w_x \!\le\! \frac{1}{300}c \right\}$ with a fallback of $p_2 = \frac{1}{2}p_1$ if this set is empty. Also, let $G = \{ x : w_x \ge c \} \supseteq \{ \CDF_\mu(x) \ge c \}$ be the set of ``good'' $x$s.

    For a given $x$, Algorithm \ref{fig:alg:estimation-task} draws $O\left(\log \log N + \log \frac{1}{\eps c} \cdot \left(\frac{1}{\eps^2 (w_x + c)} + \frac{\log^5 \eps^{-1}}{\eps^4 \left(w_x + {\eps}/{\log \eps^{-1}}\right)} \right)\right)$ samples in expectation:
    \begin{itemize}
        \item $O(1)$ calls to \procnameZpreamble (Lemma \ref{lemma:main-alg-preamble}) at the expected cost of $O\left(\log \frac{1}{\eps c} \cdot \frac{1}{\eps^2 (w_x+c)}\right)$.
        \item $O(1)$ calls to \procnameZfindZgoodZalpha (Lemma \ref{lemma:main-alg-find-good-alpha}) at the cost of $O(\log \log N)$.
        \item $O(1)$ calls to \procnameZestZexpectedZhZbeta (Lemma \ref{lemma:main-alg-use-alpha}), at total cost of $O\left(\log \frac{1}{\eps c} \cdot \frac{\log^5 \eps^{-1}}{\eps^4 \left(w_x + \frac{\eps}{\log \eps^{-1}}\right)}\right)$ in expectation.
    \end{itemize}

    Correctness (main case): assume that we obtain $\hat{p}$, $\hat{s}$ where at least one of them is not \esttoolow, and each of them which is not \esttoolow is a $(1 \pm \eps/3)$-factor estimation of its goal ($\mu(x)$ or $s_x$). This happens with probability at least $8/9$ if $x \in G$ ($w_x \ge c$).

    If $\hat{p}$ is not \esttoolow, then we just return it as a correct estimation. This happens with probability $8/9$ if $\mu(x) \ge \max\{c, \frac{1}{4}s_x\}$.
    
    If $\hat{p}$ is \esttoolow, then we have $\hat{s} = (1 \pm \eps/3)s_x$. This happens with probability at least $8/9$ if $s_x \ge \max\{c, \frac{1}{4}\mu(x)\}$, which is a superset of the constraint $(x\in G) \wedge (\mu(x) < \frac{1}{4}s_x)$. In this case, $\alpha$ is correct with probability $8/9$ and $\hat{b}$ is correct with probability $8/9$ as well. Overall, with probability $2/3$, we have $\hat{s} = (1 \pm \eps/3)s_x$ and $\hat{b} = (1 \pm \eps/2) \alpha s_x / \mu(x)$, which means that Step \ref{fig:alg:estimation-task::step:set-p-hat} sets $\hat{p} = (1 \pm \eps/3)(1 \pm \eps/2) \mu(x) = (1 \pm \eps)\mu(x)$ as desired.
    
    Correctness (reject case): assume that we obtain $\hat{p}$, $\hat{s}$ which are both \esttoolow. This happens with probability at least $8/9$ if $w_x \le \frac{1}{100}c$. In this case, we return \esttoolow, which is a correct output for $x \notin G$.
    
    Correctness (middle case): the saturation-awareness of Lemma \ref{lemma:main-alg-preamble} guarantees that, with probability $8/9$, we obtain a pair $(\hat{p}, \hat{s})$ that correctly matches the main case or the reject case.
\end{proof}

\paragraph{Avoiding probability zero sets} In Section \ref{sec:preamble} we explain how $\procnameZpreambleHREF(\mu,c,\eps;x)$ in itself takes samples only from sets that include an element that was already sampled (unconditionally) from $\mu$, thus ensuring that they have positive probability. We also explain there why, in the case where $\mu(x)=0$, $\procnameZpreambleHREF(\mu,c,\eps;x)$ always answers $(\esttoolow,\esttoolow)$, causing $\procnameZmainZsingle(\mu,c,\eps;x)$ to skip the next steps and immediately return $\esttoolow$. The other procedures used by $\procnameZmainZsingle(\mu,c,\eps;x)$ only take samples from sets that include $x$ itself, and hence if they are invoked they only take samples from positive probability sets.

\section{Target test assessment}
\label{sec:estimation-procs}

\subsection{The target test}
\label{subsec:target-test::of-sec:estimation-procs}

We formulate the algorithm whose (randomized) output is used as part of the definition of the target function $f_x$. Procedure \procnameZtargetZtestHREF uses $O(\log \frac{1}{\eps c})$ pair conditional samples to distinguish between $\mu(y) \le \mu(x)$ and $\mu(y) \ge 1.2\mu(x)$ with probability at least $1 - \eta_{c,\eps}$.

First, we provide \procnameZtargetZtestZexplicit (Algorithm \ref{fig:alg:target-test-explicit}) as a common logic for the actual target test, and a cheaper approximation of the target test that is used for finding a good $\alpha$. For this implementation we use a hard-coded tuning parameter $\kappa = \paperZvalofZkappa$. In Appendix \ref{apx:target-test-galactic} we provide an alternative implementation of the target test with reasonable constant factors (removing the dependency on $\kappa$) but with an additional asymptotic penalty, which carries over to the estimator.

\begin{algo}[ht!]
    \procname{$\textsf{Target-test-explicit}(\mu, \eta; x, y)$}
    \label{fig:alg:target-test-explicit}
    \alginput{$y \in \Omega$}
    \algoutput{\accept or \reject}
    \begin{code}
        \algpushcomment{Technical guarantee}
        \begin{If}{$y = x$}
            \algitem \reject
        \end{If}
        \algitem Let $\kappa = \paperZvalofZkappa$.
        \algitem Let $\ell \gets \ceil{\ln \eta^{-1} / 2\kappa^2}$.
        \algitem Draw $t \sim [\frac{1}{2} + \kappa, \frac{6}{11} - \kappa]$ uniformly.
        \algitem Draw $z_1,\ldots,z_\ell$ independent samples from $\mu$ conditioned on $\{x, y\}$.
        \algitem Let $Y = \abs{\{i : z_i = y \}}$.
        \begin{If}{$Y < t \cdot \ell$}
            \algitem \accept.
        \end{If}
        \begin{Else}
            \algitem \reject.
        \end{Else}
    \end{code}
\end{algo}

We use the common logic to define the target test and the approximate target test.

\begin{algo}
    \procname{$\procnameZtargetZtest(\mu, c, \eps; x, y)$}
    \label{fig:alg:target-test}
    \alginput{$y \in \Omega$}
    \algoutput{\accept or \reject}
    \begin{code}
        \algitem Call $\procnameZtargetZtestZexplicitHREF(\mu, \eta_{c,\eps}; x, y)$ and return its answer.
    \end{code}
\end{algo}

\begin{algo}
    \procname{$\procnameZtargetZtestZgross(\mu; x, y)$}
    \label{fig:alg:target-test-gross}
    \alginput{$y \in \Omega$}
    \algoutput{\accept or \reject}
    \begin{code}
        \algitem Call $\procnameZtargetZtestZexplicitHREF(\mu, 10^{-9}; x, y)$ and return its answer.
    \end{code}
\end{algo}

We provide some essential bounds of the explicit test.

\begin{lemma} \label{lemma:target-test-explicit-base-calc}
    In Algorithm \ref{fig:alg:target-test-explicit}, $\Pr\left[\procnameZtargetZtestZexplicit(\mu,\eta;x,y) = \accept \cond t \notin \frac{\mu(y)}{\mu(x) + \mu(y)} \pm \kappa \right]$ is at least $1 - \eta$ if $\frac{\mu(y)}{\mu(x) + \mu(y)} \le t - \kappa$ and at most $\eta$ if $\frac{\mu(y)}{\mu(x) + \mu(y)} \ge t + \kappa$.
\end{lemma}
\begin{proof}
    If $t \ge \frac{\mu(y)}{\mu(x) + \mu(y)} + \kappa$, then by Chernoff bound,
    \[ \Pr[Y \ge t \ell | t]
    \le \Pr\left[\Bin\left(\ell, \frac{\mu(y)}{\mu(x) + \mu(y)}\right) \ge \left(\frac{\mu(y)}{\mu(x) + \mu(y)} + \kappa\right)\ell\right]
    \le e^{-2\kappa^2 \ell} \le e^{-\ln \eta^{-1}} = \eta \]
    
    If $t \le \frac{\mu(y)}{\mu(x) + \mu(y)} - \kappa$, then by Chernoff bound,
    \[ \Pr[Y < t \ell | t]
    \le \Pr\left[\Bin\left(\ell, \frac{\mu(y)}{\mu(x) + \mu(y)}\right) < \left(\frac{\mu(y)}{\mu(x) + \mu(y)} - \kappa\right)\ell\right]
    \le e^{-2\kappa^2 \ell} \le e^{-\ln \eta^{-1}} = \eta \qedhere \]
\end{proof}

\begin{lemma} \label{lemma:target-test-explicit-bad-event}
    In the setting of Algorithm \ref{fig:alg:target-test-explicit}, $\Pr\left[t \in \frac{\mu(y)}{\mu(x) + \mu(y)} \pm \kappa\right] \le 10^{-9}$.
\end{lemma}
\begin{proof}
    $t$ is uniformly drawn in $[ \frac{1}{2} + \kappa, \frac{6}{11} - \kappa]$. Hence, the probability that the segment $t \pm \kappa$ contains a given number $p$ is at most $\frac{2\kappa}{6/11 - 1/2 - 2\kappa} \le 10^{-9}$.
\end{proof}

These bounds allow us to prove the following two lemmas.

\begin{lemma}[\procnameZtargetZtest] \label{lemma:alg:target-function-correct}
    Procedure \procnameZtargetZtestHREF uses $O(\log \frac{1}{\eps c})$ conditional samples, accepts with probability at least $1 - \eta_{c,\eps}$ if $y \in L_x$ and rejects with probability at least $1 - \eta_{c,\eps}$ if $y \in H_x$.
\end{lemma}
\begin{proof}
    For $y \in L_x$: $\mu(y) \le \mu(x)$ and hence $t \ge \frac{1}{2} + \kappa \ge \frac{\mu(y)}{\mu(x)+\mu(y)} + \kappa$ with probability $1$. By Lemma \ref{lemma:target-test-explicit-base-calc}, $\procnameZtargetZtest(\mu;x,y)$ accepts with probability at least $1 - \eta_{c,\eps}$. That is, the difference from $\Pr[y \in V_x] = 1$ is bounded by $\eta_{c,\eps}$.

    For $y \in H_x$: $\mu(y) \ge 1.2\mu(x)$ and hence $t \le \frac{6}{11} - \kappa \le \frac{\mu(y)}{\mu(x)+\mu(y)} - \kappa$ with probability $1$. By Lemma \ref{lemma:target-test-explicit-base-calc}, $\procnameZtargetZtest(\mu;x,y)$ accepts with probability at most $\eta_{c,\eps}$. That is, the difference from $\Pr[y \in V_x] = 0$ is bounded by $\eta_{c,\eps}$.
\end{proof}

\begin{lemma}[\procnameZtargetZtestZgross] \label{lemma:target-test-approx}
    $\abs{\Pr\left[y \in V_x\right] - \Pr\left[ \procnameZtargetZtestZgrossHREF(\mu;x,y) = \accept \right]} \le 10^{-8}$ for every $x,y \in \Omega$.
\end{lemma}
\begin{proof}
    We consider every $y$ individually.
    
    For $y \in L_x$: then $\Pr[y \in V_x] = 1$ and $\procnameZtargetZtestZgrossHREF$ accepts $(x,y)$ with probability at least $1 - 10^{-9}$. The difference is bounded by $10^{-9}$.

    For $y \in H_x$: then $\Pr[y \in V_x] = 0$ and $\procnameZtargetZtestZgrossHREF$ accepts $(x,y)$ with probability at most $10^{-9}$. The difference is bounded by $10^{-9}$.
    
    For $y \in M_x$: then $\Pr[y \in V_x]$ is the probability of $\procnameZtargetZtestHREF$ to accept $(x,y)$. By Lemma \ref{lemma:target-test-explicit-bad-event}, the probability that $t \in p_y \pm \kappa$ is bounded by $10^{-9}$. In this case, our best bound for the variation in behaviors is $1$. Otherwise, $t \notin p_y \pm \kappa$, and hence the variation in behaviors is bounded by $\max\{\eta_{\eps,c}, 10^{-9}\}$ by Lemma \ref{lemma:target-test-explicit-base-calc}.
    
    Combined, for every $x, y \in \Omega$, the difference between the accept probability of $\procnameZtargetZtestHREF$ and $\procnameZtargetZtestZgrossHREF$ is bounded by $10^{-9} + \max\{\eta_{c,\eps}, 10^{-9}\} \le 10^{-9} + 10^{-9} \le 10^{-8}$.
\end{proof}

\subsection{Individual drawing of $V_x$}
\label{subsec:individual-drawing-V-x::of-sec:estimation-procs}

The goal of this subsection is to provide membership query access to a drawing of a set $V_x$ through the following interface:
\begin{itemize}
    \item $\procnameZinitializeZnewZVx(c,\eps;x,q)$ -- draws a secret set $V$ according to a distribution that is $\eta_{c,\eps} q$-close to the correct distribution of $V_x$, and returns an object that supports up to $q$ queries.
    \item $\procnameZVxZquery(\mathit{obj},y)$ -- reports whether $y \in V$ or not, where $V$ is the set being held by the object $\mathit{obj}$. If the initialization parameter of the object is $q$, then only the first $q$ calls are guaranteed to be meaningful. A query may affect the contents of $\mathit{obj}$.
    \item The overall sample complexity of the initialization followed by at most $q$ queries is bounded by $O(\log \frac{1}{\eps c}) \cdot q$.
\end{itemize}

The implementation of the interface is straightforward: during initialization, we initialize an empty list of ``historical records'', which we denote by $\mathit{hist}$. In every query of an element we first look for it in the list. If it exists there, then we report (again) the recorded result, and if it is missing, then we run the \procnameZtargetZtestHREF procedure to determine whether it belongs to the set, and record the answer in the list.

\begin{algo}
    \label{fig:alg:draw-secret-vx}
    \procname{$\procnameZinitializeZnewZVx(c,\eps;x,q)$}
    \begin{code}
        \algitem Return $(c, \eps, x, \mathit{hist})$, where $\mathit{hist}$ is an empty list.
    \end{code}
\end{algo}

\begin{algo}
    \label{fig:alg:query-lazy-vx}
    \procname{$\procnameZVxZquery(\mathit{obj},y)$}
    \alginput{An object $\mathit{obj}$ created by \procnameZinitializeZnewZVxHREF representing a subset of $\Omega$, and $y\in\Omega$}
    \algoutput{Whether or not $y$ belongs to the set represented by the object}
    \algcontract{Side effects}{the $\mathit{hist}$ component of $\mathit{obj}$ may change}
    \begin{code}
        \algitem Let $c$, $\eps$, $x$, $\mathit{hist}$ be the components of $\mathit{obj}$ as a $4$-tuple.
        \begin{If}{$y = x$}
            \algitem Return \reject. \algcomment ($x \in V_x$ never happens)
        \end{If}
        \algpushcomment{($y$ was queried before)}
        \begin{If}{$\mathit{hist}$ contains $(y,\mathit{ans})$ for any $\mathit{ans}$}
        \algitem Return $\mathit{ans}$.
        \end{If}
        \algpushcomment{(new $y$)}
        \begin{Else}
            \algitem Let $\mathit{ans} \gets \procnameZtargetZtest(c,\eps;x,y)$.
            \algitem Add $(y, \mathit{ans})$ to $\mathit{hist}$.
            \algitem Return $\mathit{ans}$.
        \end{Else}
    \end{code}
\end{algo}

In the rest of this subsection we show that Algorithm \ref{fig:alg:draw-secret-vx} and Algorithm \ref{fig:alg:query-lazy-vx} implement their desired interface guarantees.

Lemma \ref{lemma:local-simulation} below states that, from the caller's perspective, drawing $V_x$ while answering one query at a time is logically equivalent to drawing it all at once.
\begin{lemma}[The local-simulation lemma] \label{lemma:local-simulation}
    Given $p_z\in [0,1]$ for every $z\in\Omega$, Consider the following three methods of drawing a random set $U \subseteq \Omega$.
    \begin{enumerate}
        \item Every $z\in U$ is chosen to be in $U$ with probability $p_z$, independently of all other members of $U$.
        \item We start with an empty $U$. For up to $q$ iterations, in the $i$th iteration we receive $z_i$ (which may depend on the results of previous iterations), and then with probability $p_{z_i}$ (independently of all previous results) add $z_i$ to $\Omega$. After this phase is over, every $y\in\Omega\setminus\{z_1,\ldots,z_q\}$ is added to $U$ independently with probability $p_y$.
        \item Same as Item 2, only here in the first phase we use $p'_{z_i}$ instead of $p_{z_i}$, where $|p_{z_i}-p'_{z_i}|\leq \delta$ (in the second phase we still use the original $p_y$).
    \end{enumerate}

    The distribution of the sets as drawn in the first item or the second item (both phases) are identical, and are $\delta q$-close to the distribution of the set as drawn in the third item.
\end{lemma}
\begin{proof}
    Methods 1 and 2 are identical since, regardless of the order of choices, every element $z$ belongs to $U$ with probability $p_z$ independently of the others.

    In every step $1 \le i \le q$, if the first $i-1$ steps in the first phase of Method 2 and Method 3 were the same, then the probability of the $i$th step of Method 2 to deviate from the $i$th step of Method 3 is bounded by $\delta$. By the union bound (and the second phase of both methods being the same), the distributions of a set drawn by Method 2 and a set drawn by Method 3 are are $\delta q$-close to each other.
\end{proof}

\begin{lemma}[\procnameZinitializeZnewZVx, \procnameZVxZquery] \label{lemma:Vx-oracle-new}
    The distribution of the output of a sequence starting with a single call to \procnameZinitializeZnewZVxHREF (Algorithm \ref{fig:alg:draw-secret-vx}), followed by $q$ calls to \procnameZVxZqueryHREF (Algorithm \ref{fig:alg:query-lazy-vx}) over the produced object with the queries $y_1,\ldots,y_q \in \Omega$, is $\eta_{c,\eps} q$-close to the distribution of the output of a sequence that draws $V \sim V_{x,c,\eps}$ and determines whether $y_i$ belongs to $V$ or not for every $1 \le i \le q$. This bound holds also for an adaptive choice of every $y_i$ based on the answers to the queries $y_1,\ldots,y_{i-1}$. The call to \procnameZinitializeZnewZVxHREF has no sample cost, and each call to \procnameZVxZqueryHREF costs $O(\log \frac{1}{\eps c})$ conditional samples.
\end{lemma}
\begin{proof}
    Observe that Algorithm \ref{fig:alg:draw-secret-vx} and Algorithm \ref{fig:alg:query-lazy-vx} together implement Method 3 of Lemma \ref{lemma:local-simulation} for drawing $V_x$ with error parameter $\delta = \eta_{c,\eps}$. Hence, its behavior is $\eta_{c,\eps}$-close to an object that is initialized with an explicit drawing of $V_x$ as a whole and then answers all membership queries (whether $y \in V_x$ for some $y$) with the same deviation probability as Method 3.
\end{proof}

\subsection{Estimation of $\E[\beta_{x,\alpha}]$}
\label{subsec:est-E-beta::of-sec:estimation-procs}

Recall that, given $A_\alpha$ and $V_x$, we define $\beta_{x,\alpha} = \Pr_\mu\left[\neg x \cond V_{x,\alpha} \cup \{x\} \right]$. 

Algorithm \ref{fig:alg:est-beta} estimates $\E[\beta_{x,\alpha}]$ as the expected value of the following indicator: first we draw $A_\alpha$, and then we repeatedly draw $y$ from $\mu$ conditioned on $A_\alpha \cup \{x\}$, until we hit an instance where $y=x$ or Algorithm \procnameZtargetZtest accepts, or until we exceed a pre-defined iteration limit.

\begin{algo}[ht!]
    \label{fig:alg:est-beta}
    \procname
    {\procnameZestZexpectedZbeta$(\mu,c,\eps;x,\alpha)$}
    \algoutput{$\hat{b} \in \E[\beta_{x,\alpha}] \pm \paperZhardcodeZdelta$}
    \begin{code}
        \algitem Let $M \gets 70000$.
        \algitem Set $m \gets 0$.
        \begin{For}{$i$ from $1$ to $M$}
            \algitem Draw $A_\alpha$ according to its definition.
            \algitem Set $b_i \gets 0$.
            \begin{For}{$10000$ iterations or until explicitly terminated\label{fig:alg:est-beta:step:indicator-oracle}}
                \algitem Draw $y$ from $\mu$, conditioned on $A_\alpha \cup \{x\}$.
                \begin{If}{$y=x$}
                    \item Exit FOR loop.
                \end{If}
                \begin{If}{$\procnameZtargetZtestZgrossHREF(\mu;x,y)$ accepts\label{fig:alg:est-beta:step:target-test-call}}
                    \item Set $b_i \gets 1$. \label{fig:alg:est-beta:step:increase-bi}
                    \item Exit FOR loop.
                \end{If}
            \end{For}
        \end{For}
        \algitem Let $\hat{b} \gets \frac{1}{M} \sum_{i=1}^M b_i$.
        \algitem Return $\hat{b}$.
    \end{code}
\end{algo}

\begin{lemma}[\procnameZestZexpectedZbeta]
\label{lemma:est-E-beta-x-alpha-pm-eps}
    Given $0 < \alpha \le 1$, Algorithm \ref{fig:alg:est-beta} estimates $\E[\beta_{x,\alpha}]$ within $\paperZhardcodeZdelta$-additive error with probability $2/3$ using $O(1)$ conditional samples.
\end{lemma}
\begin{proof}
    The worst-case cost of the algorithm is $O(1)$ calls to \procnameZtargetZtestZgrossHREF, each costing $O(1)$ samples.

    Consider the following hypothetical variants of Algorithm \ref{fig:alg:est-beta}:
    \begin{itemize}
        \item Variant A. Algorithm \ref{fig:alg:est-beta} as written (the realizable variant).
        \item Variant B. A variant where instead of $\procnameZtargetZtestZgross(\mu;x,y)$ we use a hypothetical (non-realistic) procedure that accepts $y$ with probability equal to $\Pr[y \in V_x]$.
        \item Variant C. A variant where additionally to the hypothetical procedure used in Variant B, we also remove the iteration limit of the loop in Step \ref{fig:alg:est-beta:step:indicator-oracle}, running it until it is explicitly terminated.
    \end{itemize}
    Let $r_A$, $r_B$ and $r_C$ be the expected values of each of $b_1,\ldots,b_M$ (which are all distributed the same) in each respective variant. By definition (and $\hat b$ being the average of $b_1,\ldots,b_M$ and hence having the same expectation), $r_A$ is also the expected output of Algorithm \ref{fig:alg:est-beta}, and $r_C = \E[\beta_{x,\alpha}]$.

    By the triangle inequality, $|\hat{b} - \E[\beta_{x,\alpha}]| = |\hat{b} - r_C| \le |\hat{b} - r_A| + |r_B - r_A| + |r_C - r_B|$.
    
    For every $1 \le i \le M$, $b_i$ is an indicator and hence its variance is bounded by $\frac{1}{4}$. Since the $b_i$s are independent, the variance of $\hat{b}$ is bounded by $\frac{1}{4M}$, and by Chebyshev's inequality,
    \[
        \Pr\left[\abs{\hat{b} - r_A} > \frac{1}{300}\right]
        = \Pr\left[\abs{\hat{b} - \E[\hat{b}]} > \frac{1}{300}\right]
        \le \frac{\frac{1}{4M}}{(1/300)^2}
        = \frac{22500}{70000}
        < \frac{1}{3}
    \]

    By the union bound over $10000$ iterations of the loop in Step \ref{fig:alg:est-beta:step:indicator-oracle} of Variant A and the corresponding one of Variant B, $\abs{r_B - r_A}
    \le 10000 \abs{\Pr[y \in V_x] - \Pr[\procnameZtargetZtestZgrossHREF(x,y) = \accept]}$.

    We use Lemma \ref{lemma:target-test-approx} to obtain that $\abs{\Pr[y \in V_x] - \Pr[\procnameZtargetZtestZgrossHREF(x,y) = \accept]}$ is bounded by $\frac{1}{10^8}$. Hence, $\abs{r_B - r_A} \le 10000 \cdot \frac{1}{10^8} < \frac{1}{5400}$.
    
    The behaviors of Variant B and Variant C with respect
    to the definitions differ only when $10000$ iterations of Step \ref{fig:alg:est-beta:step:indicator-oracle} are exceeded, and always $r_B \le r_C$.
    
    Let $r$ be the probability to explicitly terminate the loop in an individual iteration. Observe that $r \ge r_C$ since we always terminate after writing $b_i \gets 1$ (which happens with probability $r_C$), but we can also terminate when $y=x$. The number of iterations in Variant C distributes geometrically with parameter $r$, hence the total variation distance between a run of this loop in Variant B and in Variant C is also bounded by $\Pr[\Geo(r) > 10000]$. Combined with $0 \le r_B \le r_C \le r$, we obtain $\abs{r_C - r_B} \le \min\{ r, \Pr[\Geo(r) > 10000] \}$.

    We use $\frac{1}{675}$ as an approximate break-even parameter. If $r \ge \frac{1}{675}$, then
    \begin{eqnarray*}
        \abs{r_B - r_C} \le \Pr\left[\Geo(1/675) \ge 10000\right]
        &\le& \Pr\left[\Geo(1/675) \ge 14 \cdot 675 \right] \\
        &=& \Pr\left[e^{\frac{1}{675}\Geo(1/675)} \ge e^{14}\right] \\
        \text{[Markov]} &\le& e^{-14} \E\left[e^{\frac{1}{675}\Geo(1/675)}\right] \\
        \text{[Lemma \ref{lemma:geo-mgf}]} &=& e^{-14} \frac{\frac{1}{675} e^{1/675}}{1 - \left(1 - \frac{1}{675}\right) e^{1/675}}
        \le \frac{1}{675}
    \end{eqnarray*}
    Hence, $\abs{r_C - r_B} \le \min\left\{ \frac{1}{675}, \Pr\left[\Geo(1/675) \ge 10000\right]\right\} = \frac{1}{675}$.
    
    Overall, with probability at least $\frac{2}{3}$, $\abs{\hat{b} - \E[\beta_{x,\alpha}]} \le \frac{1}{300} + \frac{1}{5400} + \frac{1}{675} = \frac{1}{200}$.
\end{proof}

\subsection{Individual estimation of $\beta_{x,\alpha}$}
\label{subsec:est-individual-beta::of-sec:estimation-procs}

The previous subsection provided a rough approximation of $\E[\beta_{x,\alpha}]$ where the expectation is taken over random choices of $A_\alpha$ and $V_x$, but in Section \ref{sec:est-mu-x-using-alpha} we define some function $h$ and need to approximate $\E[h(\beta_{x,\alpha})]$. Hence, we also need to estimate $\beta_{x,\alpha}$ for \emph{specific} draws of $A_\alpha$ and $V_x$, where the later is given through an output of a call to \procnameZinitializeZnewZVxHREF. In the following we show how to estimate $\beta_{x,\alpha}$ with respect to $A_\alpha$ and the set $V_x$ that is virtually held by the $V_x$-object.

\begin{algo}[ht!]
    \label{fig:alg:est-beta-single}
    \procname {$\procnameZestZsingleZbeta(\mu,c,\eps;x,\alpha,\delta,A_\alpha,\mathit{obj})$}
    \alginput{$\delta < \frac{1}{4}$}
    \alginput{$\mathit{obj}$ represents a subset of $\Omega$ that is the output of \procnameZinitializeZnewZVxHREF}
    \algoutput{$\hat{b} \in \beta_{x,\alpha}(A_\alpha,V_x) \pm \delta$}
    \algcomplexity{At most $25 \frac{\ln (6/\delta)}{\delta^3}$ conditioned samples and one $\procnameZVxZquery(\mathit{obj}, \cdot)$ call per sample}
    \begin{code}
        \algitem Let $M \gets \ceil{8 / \delta^2}$.
        \algitem Set $m \gets 0$.
        \begin{For}{$M$ times}
            \begin{For}{$\ceil{3 \ln (6/\delta) / \delta}$ iterations or until explicitly terminated}
                \algitem Draw $y$ from $\mu$, conditioned on $A_\alpha \cup \{x\}$.
                \begin{If}{$y=x$}
                    \algitem Exit FOR loop.
                \end{If}
                \begin{If}{$\procnameZVxZqueryHREF(\mathit{obj}, y)$ accepts}
                    \algitem Set $m \gets m + 1$.
                    \algitem Exit FOR loop.
                \end{If}
            \end{For}
        \end{For}
        \algitem Let $\hat{b} \gets m/M$.
        \algitem Return $\hat{b}$.
    \end{code}
\end{algo}

\begin{restatable}{lemma}{lemmaZestZbetaZexpectedZtime} \label{lemma:est-beta-expected-time}
    $\E\left[\frac{\mu(A_\alpha \cup \{x\})}{\mu\left((V_x \cap A_\alpha) \cup \{x\}\right)}\right] \le \frac{20}{w_x}$.
\end{restatable}
We prove Lemma \ref{lemma:est-beta-expected-time} in Appendix \ref{apx:tech-SHEET}.

\begin{lemma}[\procnameZestZsingleZbeta] \label{lemma:est-beta-single}
    Let $0 < \delta < \frac{1}{4}$. Algorithm \ref{fig:alg:est-beta-single} outputs an estimation of $\beta_{x,\alpha}(A_\alpha,V_x)$ within additive error $\pm \delta$ with probability at least $\frac{2}{3}$. Its expected query complexity is bounded by $O\left(\log \frac{1}{\eps c} \cdot \frac{1}{\delta^2 \left(w_x + \left(\delta/ \log \delta^{-1}\right)\right)}\right)$. Additionally, the number of $\procnameZVxZqueryHREF(\mathit{obj},\cdot)$-calls is bounded by $25 \frac{\ln (6/\delta)}{\delta^3}$.
\end{lemma}
\begin{proof}
    The algorithm makes at most one $\procnameZVxZqueryHREF(\mathit{obj},\cdot)$-call per sample. The number of samples is bounded by $\ceil{8/\delta^2} \cdot \ceil{3 \ln (6/\delta)/\delta^3}$, which is at most $25 \frac{\ln (6/\delta)}{\delta^3}$ for every $\delta < \frac{1}{4}$. Recall that every $\procnameZVxZqueryHREF$-call costs $O(\log \frac{1}{\eps c})$ samples.
    
    Clearly, the worst case cost of Algorithm \ref{fig:alg:est-beta-single} is $O(1/\delta^2) \cdot O(\ln \delta^{-1} / \delta) \cdot O\left(\log \frac{1}{\eps c}\right) = O\left(\log \frac{1}{\eps c} \cdot \frac{\log \delta^{-1}}{\delta^3}\right)$.

    By Lemma \ref{lemma:est-beta-expected-time}, the expected number of inner-loop iterations is bounded by $O(1/w_x)$ and the expected complexity is bounded by $O\left(\log \frac{1}{\eps c} \cdot \frac{1}{\delta^2 w_x} \right)$. Since the expected cost cannot exceed the worst-case cost, we can reformulate the expected cost as $O\left(\log \frac{1}{\eps c} \cdot \frac{1}{\delta^2 \left(w_x + \delta / \log \delta^{-1} \right)}\right)$.

    If the algorithm does not terminate the inner loop after $\ceil{3 \ln (6/\delta)/\delta}$ iterations, then the expected gain in $m$ in every inner iteration would be exactly $\E\left[\frac{\mu\left((V_x \cap A_\alpha) \cup \{x\}\right)}{\mu(A_\alpha \cup \{x\})}\right]$ for the given $A_\alpha$ and the set being held by $\mathit{obj}$. The actual gain is less, since we must consider the possibility to terminate the loop after $\ceil{3 \ln (6/\delta)/\delta}$. Let $\hat\beta$ be the expectation of this gain.

    If $\E\left[\frac{\mu\left((V_x \cap A_\alpha) \cup \{x\}\right)}{\mu(A_\alpha \cup \{x\})}\right] \le \frac{1}{3}\delta$, then the additive penalty is at most $\frac{1}{3}\delta$. Otherwise, the distance between $\hat{\beta}$ and $\beta_{x,\alpha}$ is bounded by $\left(1 - \frac{1}{3}\delta\right)^\ceil{3 \ln (6/\delta) / \delta}
    \le \left(1 - \frac{1}{3}\delta\right)^{3 \ln (6/\delta) / \delta}
    \le e^{-\ln (6/\delta)}
    = \frac{1}{6}\delta \le \frac{1}{3}\delta$ as well.

    Note that $m$ is the sum of $M$ independent indicators, hence its variance is bounded by $M/4$. Since $M \ge 8/\delta^2$, Chebyshev's inequality implies that with probability at least $\frac{2}{3}$, $m \in \E[m] \pm \frac{2}{3}\delta M$, and in this case, $\hat{b} = \E[\hat{b}] \pm \frac{2}{3}\delta = \beta_{x,\alpha}(A_\alpha, V_x) \pm \delta$.
\end{proof}

\section{The reference estimator}
\label{sec:preamble}

In this section we prove Lemma \ref{lemma:main-alg-preamble} and provide an algorithm that demonstrates it. Towards this proof, we first provide a general saturation-aware estimator for the expectation of an indicator variable.

\subsection{Estimation of an unknown probability}

We implement here a $\left(\delta; \frac{1}{12}a, a\right)$-saturation-aware estimator for $p$, where $p$ is only accessible through an oracle that draws an indicator random variable with expected value $p$. Algorithm \ref{fig:alg:saturation-aware-estimator} draws independent samples of this indicator and then stops after seeing $O(1/\delta^2)$ occurrences of $1$ or until reaching the limit of $O(1/(\delta^2 a))$ samples.

\begin{algo}[ht!]
    \label{fig:alg:saturation-aware-estimator}
    \procname{$\procnameZsaturationZawareZest(a;\mathcal A,\delta)$}
    \alginput{An oracle $\mathcal A$ for sampling a binary variable}
    \algoutput{$\hat{p} \in \{\esttoolow\} \cup (1 \pm \delta) \E[\mathcal A]$}
    \begin{code}            
        \algitem Let $M \gets \ceil{48 / \delta^2}$.
        \algitem Let $L \gets \floor{6M/a}$.
        \algitem Set $m \gets 0$, $\ell \gets 0$.
        \begin{While}{$m < M$ and $\ell < L$}
            \algitem Set $\ell \gets \ell + 1$.
            \algitem Draw $b \sim \mathcal A$. \algcomment ($b \in \{0,1\}$).
            \algitem Set $m \gets m + b$.
        \end{While}
        \algpushcomment{(Sufficiently many occurrences observed)}
        \begin{If}{$m = M$}
            \algitem Set $\hat{p} \gets M / \ell$.
        \end{If}
        \algpushcomment{(Give-up limit reached)}
        \begin{Else}
            \algitem Set $\hat{p} \gets \esttoolow$.
        \end{Else}
        \algitem Return $\hat{p}$.
    \end{code}
\end{algo}

\begin{lemma}[\procnameZsaturationZawareZest] \label{lemma:sat-aware-est-of-black-box-indicator}
    Assume that we have an oracle $\mathcal A$ that draws $1$ with probability $p$ and $0$ with probability $1-p$, where every call is independent of past calls. For every $0 < a < 1$, Algorithm \ref{fig:alg:saturation-aware-estimator} with parameters $(a; \mathcal A,\delta)$ is a $(\delta;a/12, a)$-saturation-aware estimator for $p = \E[\mathcal A]$, at the expected cost of $O\left(\frac{1}{\delta^2 (p + a)}\right)$ oracle calls. Moreover, $\E[\hat{p}^{-1} | \hat{p} \ne \esttoolow] \le p^{-1}$, where $\hat{p}$ is the output of the estimator.
\end{lemma}
\begin{proof}
    Let $\ell'$ be the theoretical value of $\ell$ if we would not terminate the loop after $L$ iterations but let it continue until $m=M$. Observe that $\ell'$ distributes the same as the sum of $M$ independent geometric variables with parameter $p$. Also, $\Pr[\ell' = \ell | m = M] = 1$ (since we terminate the loop due to $m=M$ regardless of whether $\ell$ exceeds $L$ or not) and $\Pr[\ell \le \ell'] = 1$ (since $\ell = \min\{\ell', L\}$).

    Observe that the events ``$m = M$'' and ``$\ell' \le 6M/a$'' are equivalent: in the last iteration, we increased $m$ from $M-1$ to $M$ and also increased $\ell$ and $\ell'$ together (regardless whether the result reached $L$ or not).

    Using standard facts about geometric variables, $\E[\ell'] = M/p$ and $\Var[\ell'] \le M/p^2$. Thus the expected number of samples taken by the algorithm is $O\left(\frac{1}{\delta^2p}\right)$, and the introduction of the hard bound $\ell\leq L$ brings it down to the required $O\left(\frac{1}{\delta^2 (p + a)}\right)$.
    
    By Chebyshev's inequality,
    \begin{eqnarray*}
        \Pr\left[M/\ell' \notin (1 \pm \delta)p\right]
        = \Pr\left[\ell' \notin \frac{1}{1 \pm \delta}M/p\right]
        &\le& \Pr\left[\ell' \notin \left(1 \pm \frac{1}{2} \delta\right)\E[\ell']\right] \\
        &\le& \frac{M/p^2}{(\delta/2)^2 (M/p)^2}
        = \frac{4}{\delta^2 M}
        \le \frac{4}{\delta^2 (48 / \delta^2)}
        < 1/6
    \end{eqnarray*}

    We now go over the cases in the definition of a saturation-aware approximation.

    Case I. $p \ge a$: by Markov's inequality, the probability to stop due to the give-up limit is bounded by $\E[\ell] / \ceil{6M/a} \le (M/p) / (6M/p) = 1/6$, and the overall probability to have the correct output is at least $1 - \Pr[m = M] - \Pr[M/\ell' \notin (1 \pm \delta)p] \ge 1 - 1/6 - 1/6 = 2/3$.

    Case II. $a/12 < p < a$: the probability to return the wrong output is $\Pr[m = M \wedge M/\ell' \notin (1 \pm \delta)p] \le \Pr[M/\ell' \notin (1 \pm \delta)p] < \frac{1}{6} < \frac{1}{3}$.
    
    Case III. $p \le a/12$: the ``bad event'' is $m=M$,
    which is equivalent to $\ell' \le 6M/a$. Hence, the probability to return any real number instead of \esttoolow is:
    \begin{eqnarray*}
        \Pr[M = m]
        = \Pr[\ell' \le 6M/a]
        &=& \Pr[\ell' - \E[\ell'] \le 6M/a - M/p] \\
        &\le& \Pr[\ell' - \E[\ell'] \le M/(2p) - M/p] \\
        &=& \Pr[\ell' - \E[\ell'] \le -M/(2p)] \\
        \text{[Chebyshev]} &\le& \frac{M/p^2}{(M/(2p))^2}
        = \frac{4}{M}
        \le \frac{4}{48 / \delta^2}
        < \frac{1}{3}
    \end{eqnarray*}
    Hence we correctly return \esttoolow with probability at least $2/3$.

    Lastly, observe that:
    \[ \E\left[\hat{p}^{-1} \cond \hat{p} \ne \esttoolow\right]
    = \E\left[\ell'/M \cond m = M \right]
    = \frac{1}{M} \E\left[\ell' \cond \ell' \le 6M/a \right]
    \le \frac{1}{M} \E\left[\ell'\right]
    = p^{-1}
    \qedhere \]
\end{proof}

\subsection{The reference estimation procedure}
We first roughly estimate $w_x = \mu(x) + s_x$, and based on the result, we estimate $\mu(x)$ and $s_x$ using the magnitude of $w_x$ as a reference.

\begin{algo}[ht!]
    \label{fig:alg:est-mu-x-s-x}
    \procname{$\procnameZpreamble(\mu, c, \eps; x)$}
    \algoutput{$\hat{p}, \hat{s}$}
    \begin{code}
        \algitem Let $M \gets 13$. \algcomment For median amplification (Observation \ref{obs:median-amplification}\ref{median-amplification:13-to-8/9}).
        \begin{For}{$i$ from $1$ to $M$}
            \algitem Let $\hat{w}_i \gets \procnameZsaturationZawareZestHREF(a=c-\eta_{c,\eps};\text{Oracle},\delta=1/3)$.
            \begin{enumerate}
                \item Oracle: draw $y \sim \mu$. Return $1$ if $y=x$ or $\procnameZtargetZtestHREF(\mu,c,\eps;x,y)$ accepts. \label{fig:alg:est-mu-x-s-x:step:target-oracle}
            \end{enumerate}
        \end{For}
        \algitem Let $\hat{w} \gets \median(\hat{w}_1,\ldots,\hat{w}_M)$. \algcomment SP: $\frac{8}{9}$
        \begin{If}{$\hat{w} = \esttoolow$}
            \algitem Return $(\esttoolow, \esttoolow)$.
        \end{If}
        \begin{For}{$i$ from $1$ to $M$}
            \algitem Let $\hat{s}_i \gets \procnameZsaturationZawareZestHREF(a=\hat{w}/9;\text{Oracle},\delta=\eps/6)$.
            \begin{enumerate}
                \item Oracle: draw $y \sim \mu$. Return $1$ if $y\ne x$ and $\procnameZtargetZtestHREF(\mu,c,\eps;x,y)$ accepts. \label{fig:alg:est-mu-x-s-x:step:sx-oracle}
            \end{enumerate}
            \algitem Let $\hat{p}_i \gets \procnameZsaturationZawareZestHREF(a=\hat{w}/9; \text{Oracle}, \delta=\eps)$.
            \begin{enumerate}
                \item Oracle: draw $y \sim \mu$. Return $1$ if $y=x$. \label{fig:alg:est-mu-x-s-x:step:mu-x-oracle}
            \end{enumerate}
        \end{For}
        \algitem Let $\hat{p} \gets \median(\hat{p}_i)$. \algcomment SP: $\frac{8}{9}$
        \algitem Let $\hat{s} \gets \median(\hat{s}_i)$. \algcomment SP: $\frac{8}{9}$
        \algitem Return $(\hat{p}, \hat{s})$. \algcomment \textbf{Total success probability: $\mathbf{\frac{2}{3}}$}
    \end{code}
\end{algo}

We recall Lemma \ref{lemma:main-alg-preamble} and then prove it.

\lemmaZmainZalgZpreamble*
\begin{proof}
    The expected value of the oracle in step \ref{fig:alg:est-mu-x-s-x:step:target-oracle} (in Algorithm \ref{fig:alg:est-mu-x-s-x}) is:
    \begin{eqnarray*}
        \mu(x) + \sum_{y \in \Omega_x} \mu(y) f_x(y)
        &=& \mu(x) + \sum_{y \in L_x} \mu(y) f_x(y) + \sum_{y \in M_x} \mu(y) f_x(y) + \sum_{y \in H_x} \mu(y) f_x(y) \\
        (*) &=& \mu(x) + (1 \pm \eta_{c,\eps})\mu(L_x) + \sum_{y \in M_x} \mu(y) f(y) \pm \eta_{c,\eps} \mu(H_x) \\
        &=& \mu(x) + \mu(L_x) + \sum_{y \in M_x} \mu(y) f_x(y) \pm \eta_{c,\eps} (\mu(L_x) + \mu(H_x))
        = w_x \pm \eta_{c,\eps}
    \end{eqnarray*}
    $(*)$: since $1 - \eta_{c,\eps} \le f(y) \le 1$ for every $y \in L_x$ and $0 \le f(y) \le \eta_{c,\eps}$ for every $y \in H_x$.
    
    Also, by Lemma \ref{lemma:sat-aware-est-of-black-box-indicator} (\procnameZsaturationZawareZestHREF), the sample complexity of this estimation is $O\left(\frac{1}{\eps^2 (\hat{w} + (c - \eta_{c,\eps}))}\right) = O\left(\frac{1}{\eps^2(w_x + c)}\right)$ oracle calls, since (again by Lemma \ref{lemma:sat-aware-est-of-black-box-indicator}) $\E[\hat{w}^{-1} | \hat{w} \ne \esttoolow] \le \sum_{i=1}^M \E[\hat{w}^{-1}_i | \hat{w}_i \ne \esttoolow] \le \frac{M}{w_x} = O\left(\frac{1}{w_x}\right)$.

    If $w_x \ge c$, then the oracle's expected value is in the range $w_x \pm \eta_{c,\eps} = w_x \pm \paperZtargeterrZpreamble = (1 \pm \eps/2)w_x$, and hence with probability $8/9$, $\hat{w} = (1 \pm 1/3)(1 \pm \eps/2)w_x = (1 \pm 1/2)w_x$.

    If $s_x \ge \max\{ c, \frac{1}{4}\mu(x) \}$, then $c \le s_x \le w_x \le 5 s_x$. Hence, with probability at least $8/9$, $\hat{w} = (1 \pm 1/2)w_x \le 8s_x$. The expected value of the oracle in step \ref{fig:alg:est-mu-x-s-x:step:sx-oracle} is $s_x \pm \eta_{c,\eps} = s_x \pm \frac{1}{4}\eps c = (1 \pm \eps/4)s_x \ge (1 - 1/40)s_x$, which is more than $\hat{w}/9 \le \frac{8}{9}(s_x - \eta_{c,\eps})$. In this case, the estimation outputs an $(1 \pm \eps/6)$-estimation of $s_x \pm \eta_{c,\eps}$ with probability at least $2/3$. This estimation is in the range $(1 \pm \eps/6)(1 \pm \eps/4)s_x = (1 \pm \eps)s_x$. The cost of this estimation is bounded by $\E\left[\hat{w}^{-1} \cond \hat{w} \ne \esttoolow\right] = O\left(\frac{1}{w_x}\right) = O\left(\frac{1}{w_x + c}\right)$ oracle calls.
    
    If $s_x \le \frac{1}{400}\max\{ c, \mu(x) \}$, then either $\hat{w}$ is \esttoolow (and then we use $\hat{s} = \esttoolow$ as well), or $\hat{w} \ge \frac{1}{2}w_x = \frac{1}{2} \cdot (s_x + \mu(x)) \ge \frac{1}{2} \cdot 401 s_x > 108 s_x$. In this case, for $a = \hat{w}/9$, we have $s_x < \frac{1}{108}w_x = \frac{1}{12}a$ and hence the estimation outputs \esttoolow for $s_x$ with probability $2/3$.

    If $\frac{1}{400}\max\{ c, \mu(x) \} \le s_x \le \max\{ c, \frac{1}{4}\mu(x) \}$, then with probability at least $8/9$, either $\hat{w}$ is \esttoolow (and then we correctly return $\hat{s} = \esttoolow$) or the saturation-aware estimation of $s_x$ returns, with probability at least $2/3$, either \esttoolow or an answer in the range $(1 \pm \eps/6)s_x$. As seen before, the latter is in the range $(1 \pm \eps)s_x$.
    
    The analysis for $\mu(x)$ and $\hat{p}$ is analogous (and even a bit stricter, since the indicator for $\mu(x)$ is exact and does not have the $\pm \eta_{c,\eps}$ additive penalty).
\end{proof}

\paragraph{Avoiding probability zero sets} Note that all calls made by $\procnameZpreambleHREF(\mu, c, \eps; x)$ to $\procnameZtargetZtestHREF(\mu,c,\eps;x,y)$ involve an element $y$ that was sampled (unconditionally) from $\mu$. Since $\procnameZtargetZtestHREF(\mu,c,\eps;x,y)$ is based only on samples drawn from $\{x,y\}$, this means that no probability zero sets are involved. Additionally, if $\mu(x)=0$ then $y=x$ never happens, and the conditioning on $\{x,y\}$ causes $\procnameZtargetZtestHREF(\mu,c,\eps;x,y)$ to reject $y$ with probability $1$, forcing the output of $\procnameZpreambleHREF(\mu, c, \eps; x)$ to be $(\esttoolow,\esttoolow)$.

\section{Finding $\alpha$}
\label{sec:find-good-alpha}

We  prove Lemma \ref{lemma:main-alg-find-good-alpha} in this section. We recall it here.
\lemmaZmainZalgZfindZgoodZalpha*

We look for $\alpha$ using a binary search, adapted to a probabilistic setting as defined below.
\begin{definition}[Uncertain comparator] \label{def:uncertain-comparator}
    Let $\mathcal A$ be an oracle to a probabilistic function from $\{1,\ldots,N\}$ to $\{\text{``low''}, \text{``good''}, \text{``high''}\}$. We say that $\mathcal A$ is an \emph{uncertain comparator} if:
    \begin{itemize}
        \item Conviction: there exists a function $f : \{1,\ldots,N\} \to \{\text{``low''}, \text{``good''}, \text{``high''}\}$ such that for every $1 \le i \le N$ and every event $E$ about past calls, $\Pr[\mathcal A(i) = f(i) | E] \ge 99/100$.
        \item Monotonicity: the above function $f$ is non-decreasing monotone with respect to the full order $\text{``low''} < \text{``good''} < \text{``high''}$.
    \end{itemize}
\end{definition}

\begin{definition}[Goal range of an uncertain comparator]
    Let $\mathcal A$ be an uncertain comparator over $\{1,\ldots,N\}$. The \emph{goal range} of $\mathcal A$ is the set $\{ 1 \le i \le N : \Pr[\mathcal A(i) = \text{``good''}] \ge 99/100 \}$ (due to monotonicity, this is always a segment).
\end{definition}

\begin{definition}[Appeasement of an uncertain comparator]
    An uncertain comparator $\mathcal A$ is \emph{appeasable} if its goal range is non-empty.
\end{definition}

The interface of the binary search is stated in the following lemma:

\begin{restatable}[\procnameZstrictZbinaryZsearch]{lemma}{lemmaZfewZliesZbinaryZsearchZstrict} \label{lemma:few-lies-binary-search-strict}
    Assume that we have access to an appeasable uncertain comparator $\mathcal A$. The output of Algorithm \ref{fig:alg:strict-binary-search} is in the goal range of $\mathcal A$ with probability at least $2/3$, at the cost of $O(\log N)$ oracle calls.
\end{restatable}

We present Algorithm \ref{fig:alg:strict-binary-search} and prove the above lemma in Subsection \ref{subsec:binary-search::of-sec:find-good-alpha}. Note that the $99/100$ could be substituted by any fixed constant strictly greater than $1/2$. We use it rather than the more standard $2/3$ bound to eliminate the need for amplification in the implementation of Algorithm \ref{fig:alg:strict-binary-search}.

\subsection{The uncertain-comparator for $\alpha$}

We provide here a procedure whose guarantees are weaker than those of an uncertain comparator, in the sense that it allows for some ``gray areas'' where there is more than one correct answer (and hence no probability guarantee). Later, in the proof of Lemma 5.3, we use it in a way that side-steps the issue with the gray areas.

\begin{lemma}[\procnameZuncertainZcomparator] \label{lemma:uncertain-comparator-new}
    There exist $\alpha_x \in \left(2.3 \gamma_x, 38\gamma_x\right)$ such that:
    \begin{itemize}
        \item If $\alpha \le 1 \cdot \gamma_x$, then the output of Algorithm \ref{fig:alg:uncertain-comparator-new} is ``low'' with probability at least $2/3$.
        \item If $\frac{1}{2}\alpha_x \le \alpha \le \alpha_x$, then the output of Algorithm \ref{fig:alg:uncertain-comparator-new} is ``good'' with probability at least $2/3$.
        \item If $\alpha \ge 41\gamma_x$, then the output of Algorithm \ref{fig:alg:uncertain-comparator-new} is ``high'' with probability at least $2/3$.
    \end{itemize}
    Moreover, the number of samples drawn by Algorithm \ref{fig:alg:uncertain-comparator-new} is $O(1)$.
\end{lemma}

We prove Lemma \ref{lemma:uncertain-comparator-new} in this subsection. We essentially show that, if $\alpha = \Theta(1) \cdot \gamma_x$, then the expectations $\E[\beta_{x,\alpha}]$ and $\E[\beta_{x,2\alpha}]$ lie inside a globally fixed range and are well-separated by an additive stride. Hence, it suffices to estimate $\E[\beta_{x,\alpha}]$ within half of this stride to implement the comparator for the binary search, possibly being wrong once at each side of the correct range. We formally state this sketch in one observation and two lemmas.

\begin{observation} \label{obs:beta-monotone-wrt-alpha}
    $\E[\beta_{x,\alpha}]$ is non-decreasing monotone with respect to the choice of $\alpha$.
\end{observation}
\begin{proof}
    We can apply Observation \ref{obs:non-decreasing-set-monotone}, since the expression that defines $\beta_{x,\alpha}$, which is based on $V_{x,\alpha}$, is non-decreasing monotone.
\end{proof}

\begin{restatable}[Effective bounds for \ensuremath{\E[\beta_{x,\alpha}]}]{lemma}{lemmaZgoodZalphasZnew}
\label{lemma:good-alphas-new}
    There exists $2.3 \gamma_x \le \alpha_x \le 38\gamma_x$ for which $\E[\beta_{x,\alpha_x}] = 0.91$. Additionally, if $\alpha \le 2 \gamma_x$ then $\E[\beta_{x,\alpha}] < 0.9$ and if $\alpha \ge 41\gamma_x$ then $\E[\beta_{x,\alpha}] > 0.92$.
\end{restatable}
We prove Lemma \ref{lemma:good-alphas-new} in Appendix \ref{apx:tech-SHEET}.

At this point we provide Algorithm \ref{fig:alg:uncertain-comparator-new} based on the sketch above, which implements the weak uncertain comparator.

\begin{algo}[ht!]
    \label{fig:alg:uncertain-comparator-new}
    \procname {$\procnameZuncertainZcomparator(\mu,c,\eps;x,\alpha)$}
    \algcontract{Inaccessible data}{$\gamma_x$, $\alpha_x$ (of Lemma \ref{lemma:good-alphas-new})}
    \algoutput{``low'' if $\alpha \le 1 \gamma_x$, ``good'' if $\frac{1}{2}\alpha_x \le \alpha \le \alpha_x$, ``high'' if $\alpha \ge 41\gamma_x$}
    \begin{code}
        \algitem Let $M \gets 9$. \algcomment For median amplification (Observation \ref{obs:median-amplification}\ref{median-amplification:9-to-5/6}).
        \begin{For}{$i$ from $1$ to $9$}
            \algitem Let $\hat{\ell}_i \gets \procnameZestZexpectedZbetaHREF(\mu,c,\eps;x,\alpha)$.
            \algitem Let $\hat{h}_i \gets \procnameZestZexpectedZbetaHREF(\mu,c,\eps;x,\min\{1,2\alpha\})$.
        \end{For}
        \algitem Let $\hat{\ell} \gets \median(\hat{\ell}_1,\ldots,\hat{\ell}_m)$. \algcomment SP: $\frac{5}{6}$
        \algitem Let $\hat{h} \gets \median(\hat{h}_1,\ldots,\hat{h}_m)$. \algcomment SP: $\frac{5}{6}$
        \begin{If}{$\hat{h} < 0.905$}
            \algitem Return ``low''. \algcomment (Inferring that $h < 0.91$)
        \end{If}
        \begin{If}{$\hat{\ell} > 0.915$}
            \algitem Return ``high''. \algcomment (Inferring that $\ell > 0.91$)
        \end{If}
        \algitem Return ``good''. \algcomment (Inferring that $\ell \le 0.91 \le h$)
    \end{code}
\end{algo}

\begin{proof}[Proof of Lemma \ref{lemma:uncertain-comparator-new}]
    Case I: if $\alpha \le \gamma_x$, then by Lemma \ref{lemma:good-alphas-new}, $\E\left[\beta_{x,2\alpha}\right] \le \E\left[\beta_{x,2\gamma_x}\right] < 0.9$. In this case, with probability at least $5/6$, $\hat{h} \le \E\left[\beta_{x,\alpha}\right] + \frac{1}{200} < 0.905$, and the algorithm outputs ``low''.

    Case II: if $\alpha \ge 41\gamma_x$, then by Lemma \ref{lemma:good-alphas-new}, $\E\left[\beta_{x,\alpha}\right] \ge \E\left[\beta_{x,41\gamma_x}\right] > 0.92$. In this case, with probability at least $5/6$, $\hat{\ell} \ge \E\left[\beta_{x,\alpha}\right] - \frac{1}{200} > 0.915$, and the algorithm outputs ``high''.

    Case III. if $\frac{1}{2}\alpha_x \le \alpha \le \alpha_x$, for $\alpha_x$ guaranteed by Lemma \ref{lemma:good-alphas-new}, then with probability at least $5/6$, $\hat{\ell} \le \E\left[\beta_{x,\alpha}\right] + \frac{1}{200} \le \E\left[\beta_{x,\alpha_{x}}\right] + \frac{1}{200} = 0.915$. Also, with probability at least $5/6$, $\hat{h} \ge \E\left[\beta_{x,2\alpha}\right] - \frac{1}{200} \ge \E\left[\beta_{x,\alpha_{x}}\right] - \frac{1}{200} = 0.905$. By the union bound, with probability at least $\frac{2}{3}$, the algorithm outputs ``good''.

    The cost of this procedure is the same as the cost of two estimations of $\E[\beta_{x,\alpha}]$ using Lemma \ref{lemma:est-E-beta-x-alpha-pm-eps}, which is $O(1)$.
\end{proof}

\subsection{Proof of Lemma \ref{lemma:main-alg-find-good-alpha}}

At this point we provide Algorithm \ref{fig:alg:find-good-alpha-new} and use it to prove Lemma \ref{lemma:main-alg-find-good-alpha}.

\begin{algo}[ht!]
    \label{fig:alg:find-good-alpha-new}
    \procname{$\procnameZfindZgoodZalpha(\mu,c,\eps;x)$}
    \algoutput{$\alpha \in (1 \gamma_x, 41 \gamma_x)$}
    \begin{code}
        \algitem Let $M_1 \gets 47$. \algcomment For median amplification (Observation \ref{obs:median-amplification}\ref{median-amplification:47-to-99/100}).
        \algitem Let $M_2 \gets 9$. \algcomment For median amplification (Observation \ref{obs:median-amplification}\ref{median-amplification:9-to-5/6}).
        \algitem Let $N' \gets 1 + \ceil{\log_2 N}$.
        \begin{For}{$r$ from $0$ to $5$}
            \algitem Let $I'_r = \{ 1, \ldots, \floor{\frac{1}{6}(N' - r)} + 1 \}$.
            \algitem Let $\textsf{Oracle}_r(i')$ be the procedure that takes the median of $M_1$ independent calls to $\procnameZuncertainZcomparatorHREF(\mu,c,\eps;x,\alpha = 2^{-(6(i'-1) + r)})$. \algcomment Oracle SP: $\frac{99}{100}$
            \begin{For}{$j$ from $1$ to $M_2$}
                \algitem Let $i'_{r,j} \gets \procnameZstrictZbinaryZsearchHREF(\abs{I'_r}; \textsf{Oracle}_r(\cdot))$.
                \algitem Let $i_{r,j} \gets 6(i'_{r,j} - 1) + r$. \algcomment (For analysis only)
            \end{For}
            \algitem Let $i'_r \gets \median(i'_{r,1},\ldots,i'_{r,M_1})$. \algcomment SP: $\frac{5}{6}$
            \algitem Let $i_r \gets 6(i'_r - 1) + r$ \algcomment ($=\mathrm{median}(i_{r,1},\ldots,i_{r,M_1})$)
            \begin{If}{$\textsf{Oracle}_r(i') = \text{``good''}$}
                \algitem Return $\alpha = 2^{-i_r}$ \label{fig:alg:find-good-alpha-new:step:regular-return}
            \end{If}
        \end{For}
        \algitem Return $2^{-N'}$. \algcomment A fallback output
    \end{code}
\end{algo}

\begin{lemma} \label{lemma:range-of-gamma-x}
    If $\mu(x) \le \frac{1}{4}s_x$, then $\frac{1}{2N} \le \gamma_x \le 1$.
\end{lemma}
\begin{proof}
    Recall that $\gamma_x = \frac{\mu(x)}{s_x}$, hence $\gamma_x \le \frac{1}{4} < 1$ immediately by $\mu(x) \le \frac{1}{4}s_x$.

    For the lower bound, observe that
    \[  s_x = \E[\mu(V_x)]
        \le \max_{L_x \subseteq V_x \subseteq \Omega \setminus (H_x \cup \{x\})} \mu(V_x)
        \le \max \abs{V_x} \cdot \max_{y \in V_x} \mu(y)
        \le (N - 1) \cdot 1.2 \mu(x)
        \le 1.2 N \mu(x)
    \]
    Hence $\gamma_x = \frac{\mu(x)}{s_x} \ge \frac{1}{1.2 N} > \frac{1}{2N}$.
\end{proof}

\begin{proof}[Proof of Lemma \ref{lemma:main-alg-find-good-alpha}]
    At top level, Algorithm \ref{fig:alg:find-good-alpha-new} looks for $\alpha = 2^{-i}$ in the range $\{0,\ldots,N'\}$, where $N' \ge \log_2 (2N)$. By Lemma \ref{lemma:range-of-gamma-x}, $2^{-N'} \le \gamma_x \le 1$. By Lemma \ref{lemma:good-alphas-new}, $\alpha_x \in (2.3\gamma_x, 41\gamma_x) \subseteq (1/N, 41)$ but also $\alpha_x \le 1$, hence both $\alpha_x$ and $\frac{1}{2}\alpha_x$ lie between the endpoints of our search range ($2^{-N'}$ and $1$).
    
    Since the guarantees of $\procnameZuncertainZcomparator$ are weaker than the constraints of the uncertain comparator that can be used in \procnameZstrictZbinaryZsearch, we use an interleaving technique: instead of searching the entire
    range, we partition it to six parts, so that in every part the indexes are arranged in skips of width six. Thus, every two consecutive indexes in each part are far enough apart to allow us to distinguish whether we are below or above the ``good'' indexes. Also, as together these parts cover the entire range, at least one of the parts contains an index which our comparator explicitly marks as ``good''.

    We apply the uncertain binary search to each part separately, adding an additional ``goodness check'' to the index resulting from this search. We then greedily select the first index that was both selected by the search and verified by the additional check.

    Observe that, for satisfying the requirements  of $\procnameZstrictZbinaryZsearchHREF$, we amplify the $2/3$-success probability of \procnameZuncertainZcomparator (for inputs where it is guaranteed) to $99/100$ using median-of-$M_1$ (Observation \ref{obs:median-amplification}\ref{median-amplification:47-to-99/100}).
    
    By Lemma \ref{lemma:uncertain-comparator-new}, the comparator satisfies the requirements of Definition \ref{def:uncertain-comparator} (uncertain comparator) when restricted to the union of the ranges $\alpha \le 1\gamma_x$, $\alpha \ge 41\gamma_x$ and $\frac{1}{2}\alpha_x \le \alpha \le \alpha_x$, with respect to an unknown $\alpha_x$ whose existence is guaranteed by Lemma \ref{lemma:good-alphas-new}. There exists some integer $0 \le r_\mathrm{hit} \le 5$ for which the range $(\frac{1}{2}\alpha_x, \alpha_x]$ intersects the set $2^{-(6\mathbb N + r_\mathrm{hit})}$. Let $\alpha_\mathrm{hit}$ be the only element in this intersection and let $i_\mathrm{hit} = -\log_2 \alpha_\mathrm{hit}$. Note that the comparator is both valid (with respect to Definition \ref{def:uncertain-comparator}) in the $r_\mathrm{hit}$th range and appeasable (since the majority answer in $i_\mathrm{hit}$ is ``good'').

    We define the following events:
    \begin{itemize}
        \item $G$: at the $r_\mathrm{hit}$th iteration, the algorithm successfully executes step \ref{fig:alg:find-good-alpha-new:step:regular-return} when $i_{r_\mathrm{hit}} = i_\mathrm{hit}$.
        \item $B_r$ ($0 \le r \le 5$): at the $r$th iteration, the algorithm successfully, but wrongly, executes step \ref{fig:alg:find-good-alpha-new:step:regular-return} when $2^{-i_r} \notin (1\gamma_x, 41\gamma_x)$.
    \end{itemize}
    Clearly, if $G$ happens and none of the $B_r$s do, then the value of $\alpha$ at the return statement is in the range $(1\gamma_x, 41\gamma_x)$.

    The good event: consider the iteration in which $r = r_\mathrm{hit}$. Due to the interleave gaps, all choices of $\alpha = 2^{-(6(i'-1)+r)}$ for $i' \in I_r$ are outside the range $(\gamma_x, 41\gamma_x)$ except for a single choice for which $6(i'-1) + r = i_\mathrm{hit}$. Therefore, in the $r_\mathrm{hit}$th iteration, \procnameZuncertainZcomparator has the required behavior guarantee for all elements, and in every individual inner-iteration we obtain $i_{r_\mathrm{hit},j} = i_\mathrm{hit}$ with probability at least $2/3$. The probability that $i_{r_\mathrm{hit}} = i_\mathrm{hit}$ is at least $5/6$ since we use the median-of-$M_2$ amplification (Observation \ref{obs:median-amplification}\ref{median-amplification:9-to-5/6}). With probability at least $99/100$, the additional call to the oracle returns ``good'', and then we return $\alpha_\mathrm{hit} \in (1\gamma_x, 41\gamma_x)$. To conclude, $\Pr[G] \ge \frac{5}{6} - \frac{1}{100}$.

    Bad events: consider some $0 \le r \le 5$ for which $2^{-i_r} \notin (1\gamma_x, 41\gamma_x)$. In this case, the oracle returns either ``low'' or ``high'' with probability at least $99/100$, hence $\Pr[B_r] \le \frac{1}{100}$.

    By the union bound, \[\Pr\left[\alpha \in (\gamma_x, 41\gamma_x)\right]
    \ge \Pr\left[G \wedge \bigwedge_{r=0}^5 \neg B_r\right]
    \ge \Pr[G] - \sum_{r=0}^5 \Pr[B_i]
    \ge \left(\frac{5}{6} - \frac{1}{100}\right) - 6 \cdot \frac{1}{100}
    > \frac{2}{3} \qedhere \]
\end{proof}

\subsection{The uncertain-comparator binary search}
\label{subsec:binary-search::of-sec:find-good-alpha}

In this subsection we prove the correctness of the uncertain-comparator binary search costing $O(\log n)$ uncertain-comparator calls (Lemma \ref{lemma:few-lies-binary-search-strict}). Note that a standard amplification of the uncertain comparator allows binary search at the cost of an additional $O(\log\log n)$ factor.

Recall Lemma \ref{lemma:few-lies-binary-search-strict}:
\lemmaZfewZliesZbinaryZsearchZstrict*

The binary search algorithm executes a Markov chain over a range tree in a way that can be seen as a random-walk over a \emph{line}. To prove the correctness of our binary search variant, we recap common definitions.

\begin{definition}[Dyadic range tree]
    A \emph{dyadic range tree} is a tree whose root holds the dyadic interval $\{1,\ldots,2^k\}$, in which every non-leaf node has two children, each holding half of the node's range.
\end{definition}

\begin{observation}
    In a dyadic range tree whose root range is $\{1,\ldots,2^k\}$, all nodes of depth $0 \le k' \le k$ (where the root's depth is $0$) hold dyadic ranges of length exactly $2^{k-k'}$. In particular, all leaves (which hold singletons) have the same depth, which is $k$.
\end{observation}
\begin{proof}
    Trivial for $k'=0$. By induction for $1 \le k' \le k$: an internal node in depth $k'-1$ holds some dyadic range $\{2^{k-k'+1} t + 1, \ldots, 2^{k-k'+1} t + 2^{k-k'+1}\}$. Its left child holds the dyadic range $\{2^{k-k'} (2t) + 1, \ldots, 2^{k-k'} (2t) + 2^{k-k'}\}$ and its right child holds the dyadic range $\{2^{k-k'} (2t+1) + 1, \ldots, 2^{k-k'} (2t+1) + 2^{k-k'}\}$.
\end{proof}

A deterministic binary search can be represented as a walk over a dyadic range tree, where we start with the widest considerable range and in every step we proceed to a narrower range until we reach a leaf, whose singleton range represents the result of the search.

In our setting the comparator is probabilistic, hence the search is a random walk. If we fully trust our comparator and only go forward, as in the deterministic version, then the worst-case guarantee for the success probability is $(99/100)^{\log_2 n}$ which is too low. Alternatively, we can amplify the confidence of the uncertain comparator by considering the majority vote of $O(\log \log n)$ independent calls. This reduces the error probability to $o(\log n)$ in each step, and hence by the union bound, the run of the algorithm is $o(1)$-close to the deterministic version. Though correct, this approach brings an $O(\log \log n)$-penalty which we want to avoid.

Instead of these forward-only tree walks, we define a local logic that uses three uncertain comparisons to choose the best edge to use in each step. This edge is possibly the parent edge, which allows us to correct errors that may occur in this setting, as opposed to a deterministic binary search.

To formulate the random walk of Algorithm \ref{fig:alg:strict-binary-search} as a search, we define a set of \emph{good leaves}, corresponding to the the values for which the comparator outputs ``good'' with high probability. A random walk of a predefined length is considered \emph{successful} if we reach one of the good leaves at the very last step. We insist that we only consider the last step: it does not suffice to pass through a good leaf during the walk.

\begin{algo}[ht!]
    \label{fig:alg:strict-binary-search}
    \procname {$\procnameZstrictZbinaryZsearch(n;\mathcal A)$}
    \alginput{An uncertain comparator $\mathcal A : \{1,\ldots,n\} \to \{\text{``low''}, \text{``good''}, \text{``high''} \}$}
    \algpromise{$\mathcal A$ is appeasable}
    \algoutput{$i$ for which $\mathrm{ans}(i) = \text{``good''}$}
    \algwlog{$n$ is a power of $2$}
    \algwlogphantom[noperiod]{}{(The comparator deterministically returns ``high'' $n+1 \le i \le 2^{\ceil{\log_2 n}}$)}
    \begin{code}
        \algitem Initialize $\mathit{stk} \gets \emptyset$. \algcomment (Empty backtrace stack)
        \algitem Initialize $L_1=1$, $R_1=n$. \algcomment (Current node is the root)
        \algpushcomment{(An integer by assumption)}
        \begin{For}{$i$ from $2$ to $20\log_2 n + 1$}
            \algpushcomment{(Currently in a leaf)}
            \begin{If}{$L_{i-1} = R_{i-1}$}
                \algitem Let $\mathit{ans} \gets \mathcal A(L_{i-1})$.
                \begin{If}{$\mathit{ans} = \text{``good''}$}
                    \item Let $(L_i, R_i) \gets (L_{i-1}, R_{i-1})$ \algcomment (Stay at leaf)
                \end{If}
                \begin{Else}
                    \algitem Let $(L_i, R_i) \gets \mathrm{pop}(stk)$ \algcomment (Move to parent)
                \end{Else}
            \end{If}
            \algpushcomment{(Currently in an inner node)}
            \begin{Else}
                \algitem Let $M_{i-1} \gets \frac{1}{2}\left(L_{i-1} + R_{i-1} - 1\right)$. \algcomment (Always an integer)
                \algitem Let $\mathit{ans}_L \gets \mathcal A(L_{i-1})$.
                \algitem Let $\mathit{ans}_M \gets \mathcal A(M_{i-1})$.
                \algitem Let $\mathit{ans}_R \gets \mathcal A(R_{i-1})$.
                \algpushcomment{(Consistent answers)}
                \begin{If}{$\mathit{ans}_L \le \mathit{ans}_M \le \mathit{ans}_R$ and $\mathit{ans}_L \le \text{``good''} \le \mathit{ans}_R$}
                    \algitem $\mathrm{push}(stk, (L_{i-1}, R_{i-1}))$ \algcomment (Record current node)
                    \algpushcomment{(Middle is too high)}
                    \begin{If*}{$\mathit{ans}_M = \text{``high''}$}
                        \algitem Let $(L_i, R_i) \gets (M_{i-1}+1, R_{i-1})$. \algcomment (Move to right child)
                    \end{If*}
                    \begin{Else*}
                        \algitem Let $(L_i, R_i) \gets (L_{i-1}, M_{i-1})$. \algcomment (Move to left child)
                    \end{Else*}
                \end{If}
                \algpushcomment{(Inconsistent answers)}
                \begin{Else}
                    \begin{If*}{$(L_{i-1},R_{i-1}) = (1,n)$}
                        \algitem Let $(L_i, R_i) \gets (L_{i-1}, R_{i-1})$. \algcomment (Stay at root)
                    \end{If*}
                    \begin{Else*}
                        \algitem Let $(L_i, R_i) \gets \mathrm{pop}(stk)$. \algcomment (Move to parent)
                    \end{Else*}
                \end{Else}
            \end{Else}
        \end{For}
        \algitem Return $L_{20\log_2 n + 1}$.
    \end{code}
\end{algo}

To prove our formulation, we consider the edge-distance of every node from the set of good leaves (where a distance from a set is the minimum distance from a leaf in this set). Even though the sequence of distances is not memoryless, its guarantees suffice to apply the following lemma.

\begin{lemma}[Linear random walk] \label{lemma:linear-random-walk}
    Consider the following random walk with parameter $k$: we start with $X_1 = k$. In every step, we choose $B_i = 1$ with probability at least $1 - \frac{3}{100}$, and otherwise $B_i=0$. We allow $\Pr[B_i = 1]$ to depend on the history (the choices of $B_1,\ldots,B_{i-1}$), but the lower bound of $1 - \frac{3}{100}$ holds for every condition on individual histories. After choosing $B_i$, if $B_i = 1$ then $X_{i+1} = \max\{0, X_i - 1\}$, and if $B_i = 0$, then $0 \le X_{i+1} \le X_i + 1$ (but it must be an integer). In this setting, if $n \ge 20k+1$, then $\Pr[X_n = 0] \ge \frac{2}{3}$.
\end{lemma}
\begin{proof}
    For $1 \le i \le n$, let $G_i$ be the event ``$X_i = 0$''. Also, for $1 \le i < j \le n$, let $G_{i,j}$ be the event ``$\sum_{t=i}^{j-1} B_t \ge \frac{1}{2}(j - i)$''.

    A key observation is that $\bigvee_{i=1}^{n-1} G_i \wedge \bigwedge_{i=1}^{n-1} G_{i,n}$ implies $G_n$. If we assume the contrary, then there exists $1 \le i \le n-1$ for which $G_i \wedge G_{i,n}$ holds. Consider the maximal such $i$, and for every $i \le j \le n$ let $Y_j = X_i + \sum_{t=i}^{j-1} (-1)^{B_t}$. From the assertions on $X_1,\ldots,X_n$ and the maximality of $i$, $Y_j \ge X_j$ for every $i+1 \le j \le n$, a contradiction since $0 < X_n \le Y_n = X_i + \sum_{t=i}^{n-1} (-1)^{B_t} = X_i + (n-i) - 2 \abs{\{ i \le t \le n-1 : B_j=1 \}} \le X_i + (n-i) - 2 \cdot \frac{1}{2}(n-i) = X_i = 0$.

    It remains to show that $\Pr\left[\bigvee_{i=1}^{n-1} G_i \wedge \bigwedge_{i=1}^{n-1} G_{i,n}\right] \ge \frac{2}{3}$. Consider the negation of each part:
    \begin{eqnarray*}
        \Pr\left[\neg \bigvee_{i=1}^{n-1} G_i\right]
        = \Pr\left[\bigwedge_{i=1}^{n-1} \neg G_i\right]
        \le \Pr\left[\neg G_{n-1}\right]
        &=& \Pr\left[\Bin\left(n-1, \frac{3}{100}\right) \ge \frac{1}{2}(n-1 - k)\right] \\
        &\le& e^{-2\left(\left(\frac{1}{2} - \frac{3}{100}\right)(n-1) - \frac{1}{2}k\right)^2 / (n-1)} \\
        \text{[$n \ge 20k+1$]} &\le& e^{-2\left(\left(\frac{1}{2} - \frac{3}{100} - \frac{1}{40}\right)(n-1)\right)^2 / (n-1)} \\
        &=& e^{-0.39605 (n-1)}
        \le e^{-0.39605 \cdot 20}
        < \frac{1}{1000}
    \end{eqnarray*}

    For $1 \le i \le n-9$, we can use Chernoff bound to obtain that:
    \[
        \Pr[\neg G_{i,n}]
        \le \Pr\left[\Bin\left(n - i, \frac{3}{100}\right) \ge \frac{1}{2}(n - i) \right]
        \le e^{-2\left(\frac{1}{2} - \frac{3}{100}\right)^2 (n - i)}
        = e^{-0.4418(n-i)}
    \]

    By the union bound,
    \[  \Pr\left[\bigvee_{i=1}^{n-9} \neg G_{i,n} \right]
        \le \sum_{i=1}^{n-9} \Pr\left[\neg G_{i,n} \right]
        \le \sum_{i=1}^{n-9} e^{-0.4418(n-i)}
        \le \sum_{i=9}^{\infty} e^{-0.4418i}
        = \frac{e^{-0.4418 \cdot 9}}{1 - e^{-0.4418}}
        < \frac{1}{19}
    \]
    
    For $n-8 \le i \le n-1$ we use a collective bound:
    \[  \Pr\left[\bigvee_{i=n-8}^{n-1} \neg G_{i,n}\right]
        \le \Pr\left[\bigvee_{i=n-8}^{n-1} \left(B_i \ne 1\right) \right]
        \le \sum_{i=n-8}^{n-1} \Pr\left[B_i \ne 1 \right]
        \le 8 \cdot \frac{3}{100} = \frac{6}{25}
    \]

    Combined,
    \[  \Pr\left[\neg \bigwedge_{i=1}^{n-1} G_{i,n}\right]
        = \Pr\left[\bigvee_{i=1}^{n-1} \neg G_{i,n}\right]
        \le \Pr\left[\bigvee_{i=1}^{n-9} \neg G_{i,n}\right] + \Pr\left[\bigvee_{i=n-8}^{n-1} \neg G_{i,n}\right]
        \le \frac{1}{19} + \frac{6}{25}
        < \frac{3}{10}
    \]

    If $n \ge 20k + 1$ then:
    \[ \Pr[\neg G_n] \le \Pr\left[\neg \left(\bigvee_{i=1}^{n-1} G_i \wedge \bigwedge_{i=1}^{n-1} G_{i,n}\right)\right]
    \le \frac{1}{1000} + \frac{3}{10} < \frac{1}{3} \qedhere \]
\end{proof}

At this point we prove Lemma \ref{lemma:few-lies-binary-search-strict} (correctness of Algorithm \ref{fig:alg:strict-binary-search} as a binary search).

\begin{proof}[Proof of Lemma \ref{lemma:few-lies-binary-search-strict}]
    Let $a$ and $b$ be the endpoints of the goal range of $\mathcal A$ (which is promised to be a non-empty segment). Without loss of generality we can assume that $n$ is a power of $2$. Otherwise, we can extend the comparator to answer ``high'' with probability $1$ for every $n+1 \le i \le 2^{\ceil{\log_2 n}}$. We use a dyadic range tree with root range $\{1,\ldots,n\}$. Recall that all leaves have the same depth, $\log_2 n$, and let $L_\mathrm{good}$ be the set of leaves inside the goal range (those for which the comparator answers ``good'' with high probability).
    
    Algorithm \ref{fig:alg:strict-binary-search} defines a random walk on the dyadic range tree as follows: on a leaf $\{i\}$, we call the comparator to test whether $i$ is good or not. If it answers ``good'' then we stay on it, and otherwise we move to its parent. On an internal node $\{L,\ldots,R\}$, let $M = (L+R-1)/2$. We use the comparator to see whether the answers about $L$, $M$, $R$ make sense with respect to the predicates ``the assessments related to $L$, $M$ and $R$ indeed form a monotone sequence, and also imply that the range $[L,R]$ is not disjoint from the goal range''. If the answers make sense, then we move to one of the children based on the answer on $M$ (left child if $M$ is ``low'' or ``good'', right child if it is ``high''). Otherwise we move to its parent, unless we are already at the root, in which case we stay in place.

    Let $D_i$ be the edge-distance between the current node (which is not necessarily an ancestor of a ``good'' leaf) and the closest good leaf. Observe that at any point:
    \begin{itemize}
        \item If we are in the root or in a ``bad'' leaf $\{i\}$ (outside the goal range), then with probability at least $97/100$ we take the edge which moves us closer to the set of good leaves.
        \item If we are in an inner node that has a descendant ``good'' leaf, then with probability $97/100$ we take the edge to a child on the path to one of these leaves.
        \item If we are in an inner node that has only ``bad'' leaves in its subtree, then with probability $97/100$ we take the parent edge, moving closer to the set of good leaves.
        \item If we are in an ``good'' leaf, then we stay there with probability $99/100 > 97/100$.
    \end{itemize}
    These observations describe a memoryless random-walk on the dyadic tree. If we only consider the sequence $D_1,D_2,\ldots$ (which is not necessarily memoryless), then we satisfy the assertions of Lemma \ref{lemma:linear-random-walk}. Hence, with probability at least $2/3$, the $(20\log_2 n + 1)$st step of the algorithm is a good leaf, as desired.
\end{proof}

\section{Estimating $\mu(x)$ using $\alpha$}
\label{sec:est-mu-x-using-alpha}

In this section we prove Lemma \ref{lemma:main-alg-use-alpha}. For this, we define a function $h$ and show that $\E[h(\beta_{x,\alpha})]$ is a good approximation for $\frac{\alpha \mu(x)}{s_x}$. To estimate $\E[h(\beta_{x,\alpha})]$, we have to draw individual values of $\beta_{x,\alpha}$ and estimate them.

Recall that $\beta_{x,\alpha} = \Pr_\mu\left[\neg x \cond V_{x,\alpha} \cup \{x\} \right]$. In particular, it is fully determined by the choice of $A_\alpha$ and $V_x$. Since $\frac{\beta_{x,\alpha}}{1 - \beta_{x,\alpha}} = \frac{\mu(V_{x,\alpha})}{\mu(x)}$, we obtain that $\E\left[\frac{\beta_{x,\alpha}}{1 - \beta_{x,\alpha}}\right] = \frac{\alpha s_x}{\mu(x)}$. Alternatively, we can use $\mu(x) = \alpha \cdot s_x/\E\left[\frac{\beta_{x,\alpha}}{1 - \beta_{x,\alpha}}\right]$, where $\alpha$ is known and $s_x$ is already estimated within a $(1 \pm \eps/3)$-factor. To estimate $\mu(x)$ within a $(1 \pm \eps)$-factor as desired, it suffices to estimate $\E\left[\frac{\beta_{x,\alpha}}{1 - \beta_{x,\alpha}}\right]$ within a $(1 \pm \eps/2)$-factor, since $(1 \pm \eps/3)(1 \pm \eps/2) = 1 \pm \eps$.

Since $\frac{\beta_{x,\alpha}}{1-\beta_{x,\alpha}}$ is only bounded by $\frac{1}{\mu(x)}$, which is too large to effectively approximate, we truncate it at $T = 8 \ln \eps^{-1} + 100$ using the function $h(\beta) = \min\left\{T, \frac{\beta}{1-\beta}\right\}$, and use a separate argument to bound the difference that this truncation introduces to the expectation.

At this point we introduce Algorithm \ref{fig:alg:est-h-beta} and prove its correctness, thereby proving Lemma \ref{lemma:main-alg-use-alpha}.

\begin{algo}[ht!]
    \label{fig:alg:est-h-beta}
    \procname {$\procnameZestZexpectedZhZbeta(\mu,c,\eps;x,\alpha)$}
    \algoutput{$\hat{b} \in (1 \pm \eps/2) \alpha s_x / \mu(x)$}
    \begin{code}
        \algitem Let $T \gets 8 \ln \eps^{-1} + 100$.
        \algitem Let $M_1 \gets \ceil{9600 / \eps^2}$.
        \algitem Let $M_2 \gets \ceil{30 \ln M_1}$. \algcomment For median amplification (Observation \ref{obs:median-amplification}\ref{median-amplification:log-to-24c}).
        \algitem Let $\delta \gets \frac{\eps}{168 \ln \eps^{-1} + 2163}$.
        \algitem Let $q \gets \floor{25 M_2 \cdot \frac{\ln(6/\delta)}{\delta^3}}$.
        \begin{For}{$i$ from $1$ to $M_1$ \label{fig:alg:est-h-beta:step:outer-for}}
            \algitem Draw $A^{(i)}_\alpha$, according to its definition.
            \algitem $V^{(i)}_x \gets \procnameZinitializeZnewZVxHREF(c,\eps; x, q)$.
            \begin{For}{$j$ from $1$ to $M_2$}
                \algitem Let $\hat{\beta}_{i,j} \gets \procnameZestZsingleZbetaHREF(\mu,c,\eps; x,\alpha,\delta, A^{(i)}_\alpha, V^{(i)}_x)$.
            \end{For}
            \algitem Let $\hat{\beta}_i \gets \median(\hat{\beta}_{i,1},\ldots,\hat{\beta}_{i,M_2})$. \algcomment SP: $1 - \frac{1}{24M_1}$
            \algitem Define (without computing) $\beta_i \gets \beta(A^{(i)}_\alpha, V^{(i)}_x)$. \algcomment (For analysis only)
            \algitem Let $\hat{b}_i \gets \min\left\{\frac{\hat{\beta}_i}{1 - \hat{\beta}_i}, T\right\}$. \algcomment $= h(\hat{\beta}_i)$
        \end{For} 
        \algitem Let $\hat{b} \gets \frac{1}{M}\sum_{i=1}^M \hat{b}_i$.
        \algitem Return $\hat{b}$.
    \end{code}
\end{algo}

Before diving into the algorithmic logic we state a few arithmetic lemmas.

\begin{restatable}{observation}{obsZetaZMoneZq} \label{obs:eta-M1-q}
    $\eta_{c,\eps} M_1 q < \frac{1}{12}$.
\end{restatable}
\begin{proof}
    We use the following bounds for $\eps < \frac{1}{10}<e^{-2}$:
    \begin{eqnarray*}
        M_1 = & \ceil{\frac{9600}{\eps^2}} & \le \frac{9601}{\eps^2} \\
        M_2 = & \ceil{30 \ln M_1} \le 30 \ln \frac{9601}{\eps^2} + 1 = 30 \ln \frac{9601e^{1/30}}{\eps^2} & \le 200 \ln \eps^{-1} \\
        \delta = & \frac{\eps}{168 \ln \eps^{-1} + 2163} & \le \frac{\eps}{2331 \ln \eps^{-1}} \\
        \ln (6/\delta) \le & \ln 6 + \ln 2331 + \ln \ln \eps^{-1} + \ln \eps^{-1} & \le 12 \ln \eps^{-1}
    \end{eqnarray*}

    Therefore:
    \begin{eqnarray*}
        \eta_{c,\eps} M_1 q
        &\le& \eta_{c,\eps} \cdot M_1 \cdot 25 M_2 \frac{\ln (6/\delta)}{\delta^3} \\
        &=& \frac{1}{12} \cdot \eta_{c,\eps} \cdot 300 M_1 M_2 \frac{\ln (6/\delta)}{\delta^3} \\
        &\le& \frac{1}{12} \cdot \eta_{c,\eps} \cdot 300 \frac{9601}{\eps^2} \cdot 200 \ln \eps^{-1} \cdot \frac{12 \ln \eps^{-1}}{\eps^3 / (2331 \ln \eps^{-1})^3} \\
        &\le& \frac{1}{12} \cdot \paperZtargeterrZestsingle \cdot 9 \cdot 10^{19} \frac{(\ln \eps^{-1})^5}{\eps^5}
        < \frac{1}{12}
    \end{eqnarray*}
\end{proof}

\begin{restatable}{lemma}{lemmaZEZhZbetaZapproxZEZbetaZoverZoneZminusZbetaZnew} \label{lemma:E-h-beta-approx-E-beta-over-1-minus-beta-new}
    If $\gamma_x \le \alpha \le 50\gamma_x$ then $\E[h(\beta_{x,\alpha})] \in \left(1 \pm \frac{1}{10} \eps \right)\alpha s_x / \mu(x)$. In particular, $\E[h(\beta_{x,\alpha})] \ge \frac{9}{10}$ for $\eps < 1$.
\end{restatable}
We prove Lemma \ref{lemma:E-h-beta-approx-E-beta-over-1-minus-beta-new} in Appendix \ref{apx:tech-SHEET}.

\begin{lemma} \label{lemma:Var-hbeta}
    If $\gamma_x \le \alpha \le 50\gamma_x$ then $\Var\left[h(\beta_{x,\alpha})\right] \le 100$.
\end{lemma}
\begin{proof}
    For every $y \in \Omega$, let $\mathbf 1_{y \in V_{x,\alpha}}$ be an indicator for the event ``$y \in V_{x,\alpha}$''. Note that:
    \begin{eqnarray*}
        \Var\left[\mu(V_{x,\alpha})\right]
        &=& \sum_{y \in L_x \cup M_x} (\mu(y))^2 \Var\left[\mathbf 1_{y \in V_{x,\alpha}} \right] \\
        &=& \sum_{y \in L_x \cup M_x} (\mu(y))^2 \Pr[y \in V_{x,\alpha}](1 - \Pr[y \in V_{x,\alpha}]) \\
        &\le& \left(\max_{y \in L_x \cup M_x} \mu(y)\right) \cdot \sum_{y \in L_x \cup M_x} \mu(y) \Pr[y \in V_{x,\alpha}]
        \le 1.2\mu(x) \cdot \E[\mu(V_{x,\alpha})]
    \end{eqnarray*}

    Since $\frac{\beta_{x,\alpha}}{1 - \beta_{x,a}} = \frac{\mu(V_{x,\alpha})}{\mu(x)}$, we can now bound its variance.
    \[\Var\left[\frac{\beta_{x,\alpha}}{1 - \beta_{x,\alpha}}\right]
    = \Var\left[\frac{\mu(V_{x,\alpha})}{\mu(x)}\right]
    \le \frac{1.2\mu(x)\E[\mu(V_{x,\alpha})]}{(\mu(x))^2}
    = \frac{1.2\E[\mu(V_{x,\alpha})]}{\mu(x)}
    = \frac{1.2\alpha s_x}{\mu(x)}
    = 1.2\alpha / \gamma_x \le 100 \]

    Observe that $\Var[h(\beta_{x,\alpha})] \le \Var[\frac{\beta_{x,\alpha}}{1 - \beta_{x,\alpha}}]$, since $h$ is $1$-Lipschitz with respect to $\frac{\beta_{x,\alpha}}{1 - \beta_{x,\alpha}}$. Hence, $\Var[h(\beta_{x,\alpha})] \le 100$.
\end{proof}

\begin{restatable}{lemma}{lemmaZhZerrZscaling}\label{lemma:h-err-scaling}
    For $0 < \delta \le \frac{\eps}{21(T+3)}$ and $\hat{\beta} = \beta \pm \delta$, $h(\hat{\beta}) = h(\beta) \pm \max\{ 2\delta, \frac{1}{20}\eps h(\beta) \}$.
\end{restatable}
We prove Lemma \ref{lemma:h-err-scaling} in Appendix \ref{apx:tech-SHEET}.

\begin{lemma}[Generic bound] \label{lemma:affine-error-propagation}
    Let $r, r_1, r_2 > 0$. Let $X_1,\ldots,X_k$ be independent non-negative variables drawn from the same distribution, $\bar{X} = \frac{1}{k}\sum_{i=1}^k X_i$ be their average value, and $Y_1,\ldots,Y_k$ be another random sequence for which $\abs{Y_i} \le \max\{ r_1 \E[\bar{X}], r_2 X_i \}$ for every $1 \le i \le k$. If $r > r_1 + r_2$ then $\Pr\left[\frac{1}{k}\sum_{i=1}^k (X_i + Y_i) \ne (1 \pm r)\E[X] \right] \le \Pr\left[\bar{X} \ne (1 \pm r')\E[\bar{X}] \right]$ for $r' = \frac{r - r_1 - r_2}{1 + r_2}$.
\end{lemma}

\begin{proof}
    In the following we show that if $\bar{X} \in (1 \pm r')\E[\bar{X}]$ then $\frac{1}{k} \sum (X_i + Y_i) \in (1 \pm r) \E[\bar{X}]$. This implies that the event ``$\bar{X} \notin (1 \pm r')\E[\bar{X}]$'' contains the event ``$\frac{1}{k} \sum (X_i + Y_i) \notin (1 \pm r) \E[\bar{X}]$'', which implies that $\Pr\left[ \frac{1}{k} \sum (X_i + Y_i) \notin (1 \pm r) \E[\bar{X}] \right] \le \Pr\left[ \bar{X} \notin (1 \pm r') \E[\bar{X}] \right]$.
    
    Let $\bar{Y} = \frac{1}{k} \sum_{i=1}^k Y_i$. By the triangle inequality and the assumptions of the lemma,
    \begin{eqnarray*}
        |\bar{Y}|
        \le \frac{1}{k} \sum_{i=1}^k |Y_i|
        \le \frac{1}{k} \sum_{i=1}^k \max\{r_1 \E[\bar{X}], r_2 X_i\}
        \le \frac{1}{k} \sum_{i=1}^k (r_1 \E[\bar{X}] + r_2 X_i)
        = r_1 \E[\bar{X}] + r_2 \bar{X}
    \end{eqnarray*}

    Assume that $\bar{X} \in (1 \pm r') \E[\bar{X}]$. Combined with the previous bound we can obtain that:
    \begin{eqnarray*}
        \frac{1}{k}\sum_{i=1}^k (X_i + Y_i)
        = \bar{X} + \bar{Y}
        &\in& (1 \pm r')\E[\bar{X}] \pm r_1 \E[\bar{X}] \pm r_2 \bar{X} \\
        &\subseteq& (1 \pm r')\E[\bar{X}] \pm r_1 \E[\bar{X}] \pm (1 \pm r') r_2 \E[\bar{X}] \\
        &=& (1 \pm r' \pm r_1 \pm (1 \pm r') r_2)\E[\bar{X}] \\
        &=& (1 \pm (r' + r_1 + (1 + r') r_2))\E[\bar{X}] \\
        &=& (1 \pm ((1 + r_2) r' + r_1 + r_2))\E[\bar{X}] \\
        &=& (1 \pm r)\E[\bar{X}]
    \end{eqnarray*}
\end{proof}

\begin{lemma}\label{lemma:chebyshev-mean-Xi}
    Let $X_1, \ldots, X_{M_1}$ be a sequence of $M_1$ independent variables whose distribution is the same as $h(\beta_{x,\alpha})$. In this setting, $\Pr\left[\frac{1}{M_1}\sum_{i=1}^{M_1} X_i \ne \left(1 \pm \frac{1}{4}\eps\right)\E[h(\beta_{x,\alpha})] \right] \le \frac{50}{243}$.
\end{lemma}
\begin{proof}
    By Chebyshev's bound,
    \begin{eqnarray*}
        \Pr\left[\frac{1}{M_1}\sum_{i=1}^{M_1} X_i \ne \left(1 \pm \frac{1}{4}\eps\right) \E[h(\beta_{x,\alpha})]\right]
        &\le& \frac{\frac{1}{M_1} \Var[h(\beta_{x,\alpha})]}{\left(\frac{1}{4}\eps \E[h(\beta_{x,\alpha})]\right)^2} \\
        \text{[Lemma \ref{lemma:Var-hbeta}]} &\le& \frac{100/M_1}{\frac{1}{16} \eps^2 \left(\E[h(\beta_{x,\alpha})]\right)^2} \\
        \text{[Lemma \ref{lemma:E-h-beta-approx-E-beta-over-1-minus-beta-new}]} &\le& \frac{1600}{M_1 \eps^2 0.9^2}
        \le \frac{1600}{(9600 / \eps^2) \eps^2 \cdot 0.81}
        = \frac{1600}{9600 \cdot 0.81}
        = \frac{50}{243}
    \end{eqnarray*}
\end{proof}

\begin{lemma} \label{lemma:coupling-argument}
    The sequence $(V^{(1)}_x, \ldots, V^{(M_1)}_x)$ drawn by the algorithm is $\frac{1}{12}$-close to a sequence of $M_1$ independent drawings of $V_x$.
\end{lemma}
\begin{proof}
    Clearly, the $V^{(i)}_x$s are fully independent. For every $1 \le i \le M_1$, the distribution of $V^{(i)}_x$ is $\eta_{c,\eps}q$-close to the correct distribution of $V_x$. Hence, by a union bound, the distance between the two sequences is bounded by $\eta_{c,\eps}qM_1$. By Observation \ref{obs:eta-M1-q}, the latter expression is bounded by $\frac{1}{12}$.
\end{proof}

Recall Lemma \ref{lemma:main-alg-use-alpha} about the correctness of Algorithm \ref{fig:alg:est-h-beta}.
\lemmaZmainZalgZuseZalpha*

\begin{proof}[Proof of Lemma \ref{lemma:main-alg-use-alpha}]
    During Step \ref{fig:alg:est-h-beta:step:outer-for}, for every $i$, in the $i$th iteration we draw $A^{(i)}_\alpha$ and $V^{(i)}_x$. These correspond to some $\beta_i = \beta_{x,\alpha}(A^{(i)}_\alpha, V^{(i)}_x)$, which are not accessible to the algorithm, but are estimated by $\hat\beta_i$. By Lemma \ref{lemma:coupling-argument}, the sequence $(V^{(1)}_x, \ldots, V^{(M_1)}_x)$ is $\frac{1}{12}$-close to a sequence $(U^{(1)}_x, \ldots, U^{(M_1)}_x)$ of $M_1$ independent samples of $V_x$ drawn according to its correct distribution.
    
    For the analysis, we consider an optimal coupling of $(V^{(1)}_x, \ldots, V^{(M_1)}_x)$ with the above hypothetical sequence $(U^{(1)}_x, \ldots, U^{(M_1)}_x)$. The good event $G_\mathrm{eqv}$, whose probability is at least $\frac{11}{12}$, is defined as the event that $U^{(i)}=V^{(i)}$ for all $1\leq i\leq M_1$. This leads to a coupling of Algorithm \ref{fig:alg:est-h-beta} with a logically-equivalent algorithm that uses the $U^{(i)}_x$ sets instead of the $V^{(i)}_x$ sets, where the behaviors of Algorithm \ref{fig:alg:est-h-beta} and the hypothetical algorithm are identical when conditioned on $G_\mathrm{eqv}$.

    Recall that $\hat{\beta}_i$ is the median of $\ceil{30 \ln M_1}$ estimations of $\beta_i \pm \delta$, each of which is successful with probability at least $2/3$ (under the assumption that $G_{\mathrm{eqv}}$ happened). Since $M_1 > 150$, $\hat{\beta}_i = \beta_i \pm \delta$ with probability at least $1 - \frac{1}{24M_1}$ (Observation \ref{obs:median-amplification}\ref{median-amplification:log-to-24c}). Let $G_\mathrm{est}$ be the good event $\bigwedge_{i=1}^{M_1} (\hat{\beta}_i = \beta_i \pm \delta)$. Then by the union bound $\Pr\left[G_\mathrm{est} | G_\mathrm{eqv}\right] \ge \frac{23}{24}$.

    Assume that $G_\mathrm{eqv} \cap G_\mathrm{est}$ happens. Let $h(\hat\beta_i) = X_i + Y_i$, where $X_i = h(\beta_i)$ and $Y_i$ is the additive error, which according to Lemma \ref{lemma:h-err-scaling}, is bounded by $2\delta + \frac{1}{20}\eps h(\beta_i) = 2\delta + \frac{1}{20}\eps X_i$. Combined,
    \begin{eqnarray*}
        \Pr\left[\hat b \ne \left(1 \pm \frac{1}{3}\eps \right) \! \E\left[h(\beta_{x,\alpha})\right]\right]
        &\le& \Pr\left[\neg G_\mathrm{eqv}\right] + \Pr\left[\neg G_\mathrm{est} \cond G_\mathrm{eqv} \right] + \\
        & & +
        \Pr\left[\hat{b} \ne \left(1 \pm \frac{1}{3}\eps\right) \! \E\left[h(\beta_{x,\alpha})\right] \cond G_\mathrm{eqv} \wedge G_\mathrm{est} \right]\\
        (*) &\le& \frac{1}{12} + \frac{1}{24} + \Pr\left[\frac{1}{M_1} \sum_{i=1}^{M_1} (X_i + Y_i) \ne \left(1 \pm \frac{1}{3}\eps\right) \! \E\left[h(\beta_{x,\alpha})\right] \cond G_\mathrm{est} \right] \\
        (**) &\le& \frac{1}{8} + \Pr\left[\frac{1}{M_1} \sum_{i=1}^{M_1} X_i \ne \left(1 \pm \frac{1}{4}\eps\right) \E[h(\beta_{x,\alpha})] \cond G_\mathrm{est} \right] \\
        \text{[Lemma \ref{lemma:chebyshev-mean-Xi}]} &\le& \frac{1}{8} + \frac{50}{243}
        \le \frac{1}{3}
    \end{eqnarray*}
    $(*)$: since $X_i$ and $Y_i$ are independent of $G_\mathrm{eqv}$.\\
    $(**)$: we use Lemma \ref{lemma:affine-error-propagation} with the parameters  $r = \frac{1}{3}\eps$, $r_1 = 2\delta \le \frac{1}{300}\eps$, $r_2 = \frac{1}{20}\eps$. This results with $r'\geq\frac{1}{4}\eps$. An explicit bound:
    $   r'
        = \frac{r - r_1 - r_2}{1 + r_1}
        \ge \frac{\eps/3 - \eps/300 - \eps/20}{1 + \eps/300}
        = \frac{(7/25) \eps}{1 + \eps/300}
        \ge \frac{(7/25)\eps}{1 + 1/300}
        > \frac{1}{4}\eps
    $.

    Overall, with probability at least $2/3$, \[\hat{b} = (1 \pm \eps/3) \E[h(\beta_{x,\alpha})] \underset{\text{Lemma \ref{lemma:E-h-beta-approx-E-beta-over-1-minus-beta-new}}}= (1 \pm \eps/3)(1 \pm \eps/10)\alpha s_x/\mu(x) = (1 \pm \eps/2)\alpha s_x/\mu(x) \qedhere\]
\end{proof}

\section{Applications}
\label{sec:applications}

In this section we efficiently solve three tasks in the fully conditional model. In each of the tasks, we first construct an algorithm (or adapt an existing one) for an interim model in which one can obtain samples from the distribution along with informational queries about the distribution function itself, where the latter are received with some restrictions on availability and accuracy. Then we plug in our core estimator to provide these queries using conditional samples to complete each task. The lemmas providing this mechanism are Lemmas \ref{lemma:generic-application-single-distribution}, \ref{lemma:generic-application-k-distributions} and \ref{lemma:generic-application-k-distributions-testing}. Appendix \ref{apx:one-more-lemma} holds another another such lemma which might be useful in the future.

\subsection{Additional notations}

The following definitions relate to the restrictions that are imposed on our interim querying model.

\begin{definition}[$\eps$-approximation function] \label{def:eps-approximation-function}
    Let $\mu$ be a distribution over $\Omega$. A function $f : \Omega \to [0,1]$ is an \emph{$\eps$-approximation function} with respect to $\mu$ if $f(x) \in (1 \pm \eps)\mu(x)$ for every $x \in \Omega$.
\end{definition}

\begin{definition}[CDF $c$-truncation function] \label{def:c-truncation-function}
    Let $\mu$ be a distribution over $\Omega$. A function $f : \Omega \to [0,1]$ is a \emph{CDF $c$-truncation function} with respect to $\mu$ if:
    \begin{itemize}
        \item For every $x \in \Omega$ for which $\mathrm{CDF}_\mu(x) \ge c$, $f(x) = \mu(x)$.
        \item For every $x \in \Omega$ for which $\mathrm{CDF}_\mu(x) < c$, $f(x) \in \{0, \mu(x)\}$.
    \end{itemize}
\end{definition}

\begin{definition}[$(c,\eps)$-approximation function] \label{def:c-eps-approximation-function}
    Let $\mu$ be a distribution over $\Omega$. A function $f : \Omega \to [0,1]$ is a \emph{$(c,\eps)$-approximation function} with respect to $\mu$ if:
    \begin{itemize}
        \item $f(x) \in (1 \pm \eps)\mu(x)$ for every $x \in \Omega$ for which $\CDF_\mu(x) \ge c$.
        \item $f(x) \in (1 \pm \eps)\mu(x) \cup \{0\}$ for every $x \in \Omega$ for which $\CDF_\mu(x) < c$.
    \end{itemize}
\end{definition}

\begin{observation} \label{obs:c-eps-approximation-as-eps-approx-of-c-truncate}
    Let $\mu$ be a distribution over $\Omega$. A $(c,\eps)$-approximation function $h$ can be seen as a $(1 \pm \eps)$-multiplicative approximation of a $c$-truncated function $f$ (that is, $h(x) \in (1 \pm \eps)f(x)$ for every $x \in \Omega$).
\end{observation}

The following oracle definition is a restricted variation of the ``explicit sampler'' model of \cite{CFGM16}.

\begin{definition}[$r$-error $(c,\eps)$-explicit sampling oracle]
\label{def:r-lying-c-eps-explicit-sampling-oracle}
    Let $\mu$ be an input distribution over a set $\Omega$. The \emph{$r$-error $(c,\eps)$-explicit sampling oracle} for $\mu$ has no additional input, and outputs a pair $(x,p)$, where $x \in \Omega$ distributes like $\mu$ and with probability at least $1-r$:
    \begin{itemize}
        \item If $\CDF_\mu(x) \ge c$, then $p$ is in the range $(1 \pm \eps) \mu(x)$.
        \item If $\CDF_\mu(x) < c$, then $p$ is in the range $(1 \pm \eps) \mu(x) \cup \{0\}$.
    \end{itemize}

    The probability of error in the estimation of $\mu(x)$, as well as the distribution over the estimated value, is independent of these probabilities for other elements. The oracle guarantees \emph{consistency}, which means that if some element $y$ is drawn more than once, then all pairs of the form $(y,\cdot)$ have the same second entry.
\end{definition}

\begin{observation} \label{obs:r-lying-c-eps-explicit-decomposition}
    An $r$-error $(c,\eps)$-explicit sampling oracle can be seen as the following ensemble:
    \begin{itemize}
        \item A $(c,\eps)$-approximation function $g_\mathrm{truth} : \Omega \to [0,1]$.
        \item An arbitrary error function $g_\mathrm{err} : \Omega \to [0, 1]$.
        \item A random correctness vector $u \in \{0,1\}^\Omega$ whose entries are drawn independently, and for every $x \in \Omega$, $\Pr[u_x = 1] \ge 1 - r$.
        \item The estimation outcome of the drawn $x \sim \mu$ is $h(x) = g_\mathrm{truth}(x)$ if $u_x = 1$ and $h(x) = g_\mathrm{err}(x)$ if $u_x = 0$.
    \end{itemize}
    The function $g_\mathrm{truth}$ and $g_\mathrm{err}$ can be drawn from an arbitrary distribution over such functions.
\end{observation}

The following oracle definition is a restricted variation of the ``sample and query'' model, first defined in \cite{RS09} as the evaluation oracle.

\begin{definition}[$(c,\eps)$-peek oracle] \label{def:c-eps-peek-oracle}
    Let $\mu$ be an input distribution over a finite set $\Omega$. The \emph{$(c,\eps)$-peek oracle} for $\mu$ gets an element $x \in \Omega$ and returns:
    \begin{itemize}
        \item An arbitrary real number in the range $(1 \pm \eps)\mu(x)$, if $\CDF_\mu(x) \ge c$.
        \item An arbitrary real number in the range $\{0\} \cup (1 \pm \eps)\mu(x)$, if $\CDF_\mu(x) < c$.
    \end{itemize}
    This definition is stricter than the definition commonly used in other works, in that the set of $x\in\Omega$ for which ``$0$'' is an allowable answer is fully determined by $\mu$ itself, rather than depending on artifacts (and at times probabilistic events) of the algorithm that simulates it.
    
    The oracle guarantees \emph{consistency}, which means that  if the algorithm makes more than one query to an element $x$ then it receives the same answer to all of them.
\end{definition}

\begin{observation} \label{obs:c-eps-peek-as-querying-c-eps-approximation}
    The $(c,\eps)$-peek oracle for a distribution $\mu$ can be seen as the querying of an arbitrarily (and possibly probabilistically) predefined $(c,\eps)$-approximation function (Definition \ref{def:c-eps-approximation-function}).
\end{observation}

\begin{restatable}[Amplification of testing]{observation}{obsZmajorityZamplificationZadZhoc} \label{obs:majority-amplification-ad-hoc}
    Assume that we have a decision test whose answer is correct with probability at least $5/8$. Then the majority answer of $3$ indpendent trials is correct with probability at least $2/3$ and the majority answer of $45$ independent trials is correct with probability at least $3/4$.
\end{restatable}
We provide a proof for Observation \ref{obs:majority-amplification-ad-hoc} in Appendix \ref{apx:tldr}.

\subsection{Learning of histograms}
\label{subsec:learning-histograms::of-sec:applications}

We first prove a generic lemma about a reduction from the $(c,\eps)$-explicit sampling model to the fully conditional model.

\begin{lemma} \label{lemma:generic-application-single-distribution}
    Consider an algorithm $\mathcal A$ whose input is a distribution $\mu$ over $\Omega$ of size $N$ and its output is an element of a set $R$, and whose access to $\mu$ consists of making at most $q$ calls to the $r$-error $(c,\eps)$-explicit sampling oracle.

    Assume that for every input $\mu$ there exists a set $R_\mu \subseteq R$ for which $\Pr\left[\mathcal A(\mu) \in R_\mu\right] > \frac{2}{3}$, where the probability is over a draw of the outcome sequence resulting from the algorithm's calls to a valid $r$-error $(c,\eps)$-explicit sampling oracle (along with the algorithm's internal randomness).

    In this setting, there exists an algorithm $\mathcal A'$ in the fully conditional model whose sample complexity is $O(q \cdot (\log \log N + 1 / \eps^4) \cdot \poly(\log r^{-1}, \log c^{-1}, \log \eps^{-1}))$, such that for every $\mu$, $\Pr\left[\mathcal A'(\mu) \in R_\mu\right] > \frac{2}{3}$.
\end{lemma}
\begin{proof}
    Without loss of generality, we assume that the algorithm draws exactly $q$ samples $x_1,\ldots,x_q \sim \mu$ and receives $p_1,\ldots,p_q$ such that in expectation the fraction of errors is at most $r$.
    
    We run $\mathcal A$ and simulate the outcome estimation of the explicit sampling oracle while keeping ``history records''. In the $i$th call to the explicit sampling oracle we:
    \begin{itemize}
        \item Draw $x_i \sim \mu$, independent of past calls.
        \item Check whether $x_i = x_j$ for some $j \le i-1$. If such a $j$ exists, then we re-use the estimation for $\mu(x_j)$.
        \item If $x_i \notin \{x_1,\ldots,x_{i-1}\}$, then we use the median of $\ceil{30 \ln r^{-1}}$ independent calls to \procnameZmainZsingleHREF (Theorem \ref{th:estimation-task}) with parameters $(c, \eps)$.
    \end{itemize}

    Clearly, the sequence $x_1,\ldots,x_q$ is independently and identically distributed like $\mu$. By Observation \ref{obs:median-amplification}\ref{median-amplification:log-to-2c}, the probability to wrongly estimate an individual $\mu(x_i)$ (that is eligible for estimation) is bounded by $\frac{1}{2}r$, independently for every distinct $x_i$. Hence, we fully simulate the $r$-error oracle without any additional error.
    
    By Corollary \ref{cor:expected-complexity-of::th:estimation-task}, the expected complexity of a single estimation of $x_i \sim \mu$ is $O(\log \log N) + \poly(\log \eps^{-1}) \cdot O\left(\frac{\log^2 c^{-1}}{\eps^2} +  \frac{\log c^{-1}}{\eps^4} \right)$. We repeat this $\ceil{30 \ln r^{-1}}$ times for amplification for every $1 \le i \le q$. Overall, the expected sample complexity is bounded by \[O(q \cdot (\log \log N + 1/\eps^4) \cdot \poly(\log r^{-1}, \log c^{-1}, \log \eps^{-1})) \qedhere\]
\end{proof}

Most of this subsection is dedicated to an algorithm for $\eps$-learning the histogram of $\mu$ at the cost of $O(\log(N/\eps) / \eps^3)$ explicit samples. It works using the bucketing technique of \cite{batuFFKRW2001}: instead of considering the distribution itself, we consider a ``lower resolution picture'' that results from categorizing the possible values of $\mu(x)$ by powers of $(1-O(\eps))$ that are close to them.

We start by stating a folklore lemma for learning a distribution over a small domain.

\begin{lemma}[Folklore]\label{lemma:learn-simple-distribution}
    Let $\mu$ be a distribution over $\{1,\ldots,n\}$ for $n \ge 16$. Assume that we construct a distribution $\mu'$ over $\{1,\ldots,n\}$ as follows: we draw $q$ independent samples from $\mu$, for every $1 \le i \le n$ let $X_i$ be the random variable counting the number of occurrences of $i$ in our samples, and let $\mu'(i) = X_i / q$. If $q \ge n/\eps^2$, then with probability at least $8/9$, $\dtv(\mu,\mu') \le \eps$ and $\mu'(i) \le 2 \max\{\eps^2, \mu(i)\}$ for every $1 \le i \le n$.
\end{lemma}
\begin{proof}
    Let $E \subseteq \{1,\ldots,n\}$ be an arbitrary event. By the construction, $\mu'(E) = \frac{1}{q} \sum_{i \in E} X_i$. By Chernoff bound,
    \[  \Pr\left[\abs{\mu'(E) - \mu(E)} > \eps\right]
        = \Pr\left[\Bin(q,\mu(E)) \ne \mu(E)q \pm \eps q\right]
        \le 2e^{-2 \eps^2 q}
        \le 2e^{-2 n}
        < 2^{-n-5} \]
    
    By the union bound, the probability to deviate even once is bounded by $2^n \cdot 2^{-n-5} \le \frac{1}{32}$. Hence, with probability at least $31/32$, $\dtv(\mu,\mu') = \sup_E \abs{\mu(E) - \mu'(E)} \le \eps$.

    Additionally, consider $1 \le i \le n$ and let $p_i = \max\{\eps^2, \mu(i)\}$. By Chernoff bound,
    \[  \Pr\left[\mu'(i) > 2p_i\right]
        \le \Pr\left[\Bin(q,p_i) > 2p_i q\right]
        \le e^{-\frac{1}{3} p_i q}
        \le e^{-\frac{1}{3} \eps^2 (n/\eps^2)}
        = e^{-n/3} \]
    By the union bound, the probability that $\mu'(i) \le 2\max\{ \eps^2, \mu(i) \}$ for any $1 \le i \le n$ is bounded by $n \cdot e^{-n/3} \le \frac{1}{12.9}$.

    By the union bound, the probability of any ``bad event'' occurring is at most $\frac{1}{32} + \frac{1}{12.9} \le \frac{1}{9}$.
\end{proof}

Before we present the learning algorithm, we formally define the histogram buckets we wish to learn.

\begin{definition}[Bucket function] \label{def:bucket-function}
    Let $\eps > 0$, $t \ge 2 + \ln(N / \eps^2) / \eps$, and $\mu$ be a distribution over $\Omega$ of size $N$. A function $f : \Omega \to \{1,\ldots,t;\infty\}$ is a \emph{bucket function} if for every $x \in \Omega$:
    \begin{itemize}
        \item If $\mu(x) > e^{-\eps (t - 2)}$ and $\CDF_\mu(x) \ge \eps$, then $\mu(x) \in e^{-\eps (f(x) \pm 2)}$.
        \item Otherwise, $f(x) = \infty$ or $\mu(x) \in e^{-\eps (f(x) \pm 2)}$.
    \end{itemize}
\end{definition}

\begin{lemma} \label{lemma:eps-eps-approx-to-2eps-bucket-function}
    Let $h$ be an $(\eps,\eps)$-approximation function of $\mu$. The function $f(x) \!=\! \ceil{-\ln h(x) / (2\eps)}$, where values larger than $t$ are mapped to $\infty$, is a $2\eps$-bucket function of $\mu$.
\end{lemma}
\begin{proof}
    For every $x \in \Omega$ for which $\mu(x) > e^{-2\eps (t - 2)}$ and $\CDF(x) \ge \eps$ (noting that such a value is never mapped to the $\infty$-bucket):
    \begin{eqnarray*}
        \frac{\ln h(x)}{2\eps} - 1 \le &-f(x)& \le \frac{\ln h(x)}{2\eps} \\
        e^{-2\eps (f(x) + 2)} \le e^{\ln h(x) - 4\eps} = e^{-4\eps} h(x) \le &\mu(x)& \le e^{2\eps} h(x) = e^{\ln h(x) + 2\eps} \le e^{-2\eps (f(x) - 2)}
    \end{eqnarray*}
\end{proof}

\begin{definition}[Bucket distribution] \label{def:bucket-distribution}
    Let $\eps > 0$, $t \ge 2 + \ln (N / \eps^2) / \eps$, $\mu$ be a distribution over $\Omega$ and $f : \Omega \to \{1,\ldots,t;\infty\}$ be an $\eps$-bucket function. The \emph{bucket distribution} of $\mu$ with respect to $f$ is the distribution $\mu_f$ over $\{1,\ldots,t;\infty\}$ for which $\mu_f(i) = \Pr_{x \sim \mu}[f(x) = i]$.
\end{definition}

\begin{definition}[Bucket-transform $T_{\eps,\eps'}$]
    Let $\eps' \ge 2\eps$. The \emph{bucket transform from $\eps$ to $\eps'$}, $T_{\eps,\eps'} : \mathbb N \cup \{\infty\} \to \mathbb N \cup \{\infty\}$, maps $\infty$ to itself and every $i \in \mathbb N$ to $\ceil{\frac{\eps}{\eps'}i}$.
\end{definition}

\begin{lemma} \label{lemma:bucket-function-parameter-change}
    Let $f$ be an $\eps$-bucket function of a distribution $\mu$ with respect to some $t \ge 2 + \ln (N / \eps^2) / \eps$. For $\eps' \ge 2 \eps$ and $t' = 2 + \floor{\frac{\eps}{\eps'}(t - 2)}$, the function $g : \Omega \to \{1,\ldots,t';\infty\}$ defined as $g(x) = T_{\eps,\eps'}(f(x))$ is an $\eps'$-bucket function of $\mu$ with respect to $t'$.
\end{lemma}
\begin{proof}
    For validity, observe that $\ln \left((\eps'/\eps)^2\right) / \eps' \ge 1$ and hence we can obtain:
    \[
        t' = 2 + \floor{\frac{\eps}{\eps'}(t - 2)}
        \ge 1 + \frac{\eps}{\eps'} \cdot \frac{\ln (N/\eps^2)}{\eps}
        = 1 + \frac{\ln (N/(\eps')^2) + \ln \left((\eps'/\eps)^2\right) }{\eps'}
        \ge 2 + \frac{\ln (N/(\eps')^2)}{\eps'}
    \]

    Also, observe that $\eps'(t' - 2) \le \eps' \cdot \frac{\eps}{\eps'} (t - 2) = \eps (t - 2)$ and that $t' \ge \ceil{\frac{\eps}{\eps'}t}$.
    
    Consider $x$ for which $f(x) \ne \infty$. By definition of $f$, $f(x) \in \eps^{-1} \ln \frac{1}{\mu(x)} \pm 2$. By definition of $g$, $g(x) = \frac{\eps}{\eps'} \cdot \left(\eps^{-1} \ln \frac{1}{\mu(x)} \pm 2\right) \pm 1 = (\eps')^{-1} \ln \frac{1}{\mu(x)} \pm \left(2\frac{\eps}{\eps'} + 1\right) \subseteq (\eps')^{-1} \ln \frac{1}{\mu(x)} \pm 2$.

    Consider $x$ for which $f(x) = \infty$ (and hence $g(x) = \infty$ as well). If $\CDF_\mu(x) \le \eps$ then $\CDF_\mu(x) \le \eps'$ as well. Otherwise, $\mu(x) \le e^{-\eps(t-2)}$. In this case, by the constraint of $t'$, $\mu(x) \le e^{-\eps'(t'-2)}$ as well.
\end{proof}

\begin{lemma} \label{lemma:bucket-histogram-close}
    Let $\eps > 0$, $N \ge 1$, $t \ge 2 + \ln (N/\eps^2) / \eps$. Let $\mu$, $\tau$ be two distributions over $\Omega$ of size $N$ and $f_\mu, f_\tau : \Omega \to \{1,\ldots,t;\infty\}$ be $\eps$-bucket functions for $\mu$ and $\tau$ respectively. Let $\mu_{f_\mu}$ and $\tau_{f_\tau}$ the bucket distributions corresponding to $(\mu,f_\mu)$ and $(\tau,f_\tau)$ respectively. In this setting, $\dhist(\mu ; \tau) \le \eps'$ for $\eps' = \dtv(\mu_{f_\mu}, \tau_{f_\tau}) + 5\eps + 16\eps^2$.
\end{lemma}
\begin{proof}

    For every $i \in \{1,\ldots,t;\infty\}$, let $B^\mu_i = \{ x \in \Omega : f_\mu(x) = i \}$ and $B^\tau_i = \{ x \in \Omega : f_\tau(x) = i \}$. Also, let $(L^\mu_i, R^\mu_i)$ (resp. $(L^\tau_i, R^\tau_i)$) be a partition of $B^\mu_i$ (resp. $B^\tau_i$) for which $\abs{L^\mu_i} = \min\left\{\abs{B^\mu_i}, \abs{B^\tau_i}\right\}$ (resp. $\abs{L_i^\tau}=\min\left\{\abs{B^\mu_i}, \abs{B^\tau_i}\right\}$). Let $L^\mu = L^\mu_\infty \cup \bigcup_{i=1}^t L^\mu_i$, and analogously define $R^\mu$, $L^\tau$, $R^\tau$ as the corresponding unions.

    Let $\pi$ be a permutation over $\Omega$ such that for every $i \in \{1,\ldots,t;\infty\}$, $\pi$ maps $L^\mu_i$ onto $L^\tau_i$, and also maps $R^\mu$ onto $R^\tau$. Such a permutation exists since $\abs{L_i^\mu} = \abs{L_i^\tau}$ for every $i \in \{1,\ldots,t;\infty\}$ and $\abs{R^\mu} = N -\abs{L^\mu} = N - \abs{L^\tau} = \abs{R^\tau}$.

    Clearly, elements in $L^\mu \setminus L^\mu_\infty$ are mapped to elements in the same bucket, hence $\frac{\mu(x)}{\tau(\pi(x))} \in e^{\pm 4\eps} = 1 \pm (4\eps + 16\eps^2)$ for such elements. The mass of the other elements is bounded by:
    \[
        \mu(B^\mu_\infty) + \sum_{i=1}^t \mu(R^\mu_i)
        \le \mu(B^\mu_\infty) + \sum_{i=1}^t e^{-\eps (i-2)} \max\left\{0, \abs{B^\mu_i} - \abs{B^\tau_i}\right\}
    \]
    
    For the first part: every element with $\mu(x) \le e^{-\eps(t-2)}$ has $\CDF_\mu(x) \le N\cdot e^{-\eps(\ln (N/\eps^2) / \eps)} = \eps^2 < \eps$, hence for $B^\mu_\infty$ it suffices to only consider elements with $\CDF_\mu(x) \le \eps$. Their mass is bounded by $\eps$ due to the definition of $\CDF_\mu$. For the second part:
    \begin{eqnarray*}
        \sum_{i=1}^t e^{-\eps (i-2)} \max\left\{0, \abs{B^\mu_i} - \abs{B^\tau_i}\right\}
        &\le& \sum_{i=1}^t e^{-\eps (i-2)} \max\left\{0, e^{\eps (i+2)} \mu(B^\mu_i) - e^{\eps (i-2)} \tau(B^\tau_i)\right\} \\
        &=& \sum_{i=1}^t \max\left\{0, e^{4\eps} \mu(B^\mu_i) - \tau(B^\tau_i)\right\} \\
        &=& \sum_{i=1}^t \max\left\{0, \mu(B^\mu_i) - \tau(B^\tau_i) + \left(e^{4\eps} - 1\right) \mu(B^\mu_i)\right\} \\
        &\le& \sum_{i=1}^t \max\left\{0, \mu(B^\mu_i) - \tau(B^\tau_i)\right\} + \left(e^{4\eps} - 1\right) \sum_{i=1}^t \tau(B^\tau_i) \\
        &\le& \dtv(\mu_{f_\mu}, \tau_{f_\tau}) + (4\eps + 16\eps^2)
    \end{eqnarray*}

    Overall, $\Pr_\mu\left[\mu(x) \notin (1 \pm (4\eps + 16\eps^2))\tau(x)\right] \le \dtv(\mu_{f_\mu}, \tau_{f_\tau}) + 5\eps + 16\eps^2$.
\end{proof}

\begin{lemma} \label{lemma:shitty-integer-optimization}
    Let $N \ge 1$, $\eps > 0$ and $t \ge 2+ \ln(N/4\eps^2) / 2\eps$. Let $\mu$ be a distribution over $\Omega$ and let $f_\mu$ be a $2\eps$-bucket function with respect to $\mu$ and $t$. Let $\nu$ be a distribution over $\{1,\ldots,t;\infty\}$ for which $\dtv(\mu_{f_\mu}, \nu) \le 6 \eps$. In this setting, the following constraint problem is algorithmically solvable given full access to $t$, $\eps$, $N$ and $\nu$:
    \begin{itemize}
        \item $\sum_{i=1}^t N_i \le N$.
        \item $\sum_{i=1}^t N_i p_i \le 1$.
        \item $\sum_{i=1}^t \abs{N_i p_i - \nu(i)} \le 12 \eps$.
        \item For every $1 \le i \le t$: $N_i \ge 0$ is an integer.
        \item For every $1 \le i \le t$: $e^{-2\eps(i+2)} \le p_i \le e^{-2 \eps(i-2)}$.
    \end{itemize}
\end{lemma}

\begin{proof}
    Solvability: recall Definition \ref{def:bucket-distribution}, and observe that the following is a feasible solution: for every $1 \le i \le t$, $N_i = \abs{\{x : f_\mu(x) = i\}}$ and $p_i = \frac{\mu_{f_\mu}(i)}{N_i}$.
    
    Computability: by simple arithmetic $N_i \le (12 \eps + \nu(i))e^{2\eps (i+2)}$ for every $1 \le i \le t$, hence the search range for $(N_1,\ldots,N_t)$ is finite and can be exhausted algorithmically. The algorithm considers every assignment of $(N_1,\ldots,N_t)$ within their bounds, solves the corresponding LP problem (a sum of $t$ absolute values can be converted to to $2^t$ linear constraints) and tests the feasibility of the result assignment.
\end{proof}

Note that the time complexity relating to the above lemma is large. If we allow the algorithm to find a solution relaxing the third condition to $\sum_{i=1}^t \abs{N_i p_i - \nu(i)} \le 24\eps$ for $\nu$ satisfying the assertions of the lemma, then we can greatly reduce the time complexity, and this is still sufficient for the histogram learning task. We do not prove this here.

Algorithm \ref{fig:alg:peekaboo-learn-histogram-buckets} solves the histogram learning task by drawing $O(t/\hat{\eps}^2)$ sample elements, and estimating their mass to associate them with their buckets, possibly with a $\pm 2$-additive shift.

\begin{algo}
    \label{fig:alg:peekaboo-learn-histogram-buckets}
    \procname{$\procnameZlearnZhistogramZbuckets(\hat{\eps}; \mu)$}
    \algoracle{The $\hat{\eps}$-error $(\hat{\eps}, \hat{\eps})$-explicit sampling oracle, $\hat{\eps} < 1/27$}
    \algoutput{A distribution that is $6\hat{\eps}$-close to some $2\hat{\eps}$-bucket distribution of $\mu$}
    \algsuccpr{$2/3$}
    \begin{code}
        \algitem Let $t \gets \ceil{\ln (N/\hat{\eps}^2) / 2 \hat{\eps}} + 2$.
        \algitem Let $q \gets \ceil{(t+1)/\hat{\eps}^2}$.
        \algitem Set $X_1,\ldots,X_t;X_\infty \gets 0$.
        \begin{For}{$q$ times}
            \algitem Explicitly draw $y \sim \mu$ and obtain $\hat{p}_y$.
            \algitem Set $f(y) \gets \ceil{-\frac{\ln \hat{p}_y}{ 2 \hat{\eps}}}$.
            \begin{If}{$f(y) = 0$}
                \algitem Set $f(y) \gets 1$.
            \end{If}
            \begin{If}{$f(y) > t$}
                \algitem Set $f(y) \gets \infty$.
            \end{If}
            \algitem Set $X_{f(y)} \gets X_{f(y)} + 1$.
        \end{For}
        \algitem Let $\tau_B$ the distribution defined as $\tau_B(i) = X_i / q$ for every $i \in \{1,\ldots,t;\infty\}$.
        \algitem Return $\tau_B$.
    \end{code}
\end{algo}

\begin{lemma}[\procnameZlearnZhistogramZbuckets] \label{lemma:learning-histogram-buckets-correct}
    Algorithm \ref{fig:alg:peekaboo-learn-histogram-buckets} makes $O(\log (N/\hat{\eps}) / \hat{\eps}^3)$ calls to the $\hat{\eps}$-error $(\hat{\eps},\hat{\eps})$-explicit sampling oracle, and with probability at least $2/3$ returns a distribution $\tau_B$ that is $6 \hat{\eps}$-close to a bucket distribution of the form $\mu_{f_\mu}$ for some $2\hat{\eps}$-bucket function $f_\mu$ of $\mu$.
\end{lemma}
\begin{proof}
    By Observation \ref{obs:r-lying-c-eps-explicit-decomposition}, we can see the $\hat{\eps}$-error $(\hat{\eps},\hat{\eps})$-explicit sampling oracle as a propagation of the following:
    \begin{itemize}
        \item An $(\hat{\eps},\hat{\eps})$-approximation function $g_\mathrm{truth}$ of $\mu$.
        \item An arbitrary error function $g_\mathrm{err} : \Omega \to [0,1]$.
        \item A correctness vector $u \in \{0,1\}^\Omega$ whose entries are drawn independently, where $\Pr[u_x = 1] \ge 1 - r$ for every $x \in \Omega$.
        \item The estimation outcome of the oracle for $x$ is $h(x) = g_\mathrm{truth}(x)$ if $u_x = 1$ and $h(x) = g_\mathrm{err}(x)$ if $u_x = 0$.
    \end{itemize}
    This way, the analysis can use $h(y)$ instead of $\hat{p}_y$ for the samples, noting that $h(y)$ is also defined for non-sampled $y\in\Omega$.

    Let $f$ and $f'$ be the functions that map every $x \in \Omega$ to its $2\hat{\eps}$-bucket according to $g_\mathrm{truth}(x)$ and $h(x)$ respectively. More precisely, the bucket associated with the mass $p$ (which can be obtained from $g_\mathrm{truth}$ or from $h$) is $\ceil{-\ln p / (2\hat{\eps})}$ (or $\infty$ if larger than $t$). By Lemma \ref{lemma:eps-eps-approx-to-2eps-bucket-function}, $f$ is indeed a bucket function of $\mu$.

    Since $\mu_f$ and $\mu_{f'}$ are mappings of $\mu$ with respect to $f$ and $f'$, $\dtv(\mu_f, \mu_{f'}) \le \sum_{x \in \Omega : f(x) \ne f'(x)} \mu(x)$. For every $x$, the probability that $g_\mathrm{truth}(x) \ne h(x)$ is bounded by $\hat{\eps}$, and hence $\E_u[\dtv(\mu_f, \mu_{f'})] \le \hat{\eps}$. By Markov's inequality, with probability at least $4/5$ over the choice of the correctness vector $u$, $\dtv(\mu_f, \mu_{f'}) \le 5\hat{\eps}$.

    By Lemma \ref{lemma:learn-simple-distribution}, with probability at least $8/9$, the distribution $\tau_B$ constructed by the algorithm is $\hat{\eps}$-close to $\mu_{f'}$. By the triangle inequality and the union bound, with probability at least $1 - 1/9 - 1/5 > 2/3$, $\dtv(\tau_B, \mu_f) \le \dtv(\tau_B, \mu_f) + \dtv(\mu_{f'}, \mu_f) \le \hat{\eps} + 5\hat{\eps} = 6\hat{\eps}$.
    
    The sample complexity is trivial.
\end{proof}

The histogram learning algorithm works by converting the (approximated) distribution over buckets back to a distribution over $\Omega$.

\begin{algo}
    \label{fig:alg:peekaboo-learn-histogram}
    \procname{$\textsf{Learn-histogram}(\eps; \mu)$}
    \algoracle{The $\hat{\eps}$-error $(\hat{\eps}, \hat{\eps})$-explicit sampling oracle for $\hat{\eps} = \frac{1}{300}\eps$}
    \begin{code}
        \algitem Let $\tau_B \gets \procnameZlearnZhistogramZbucketsHREF(\hat{\eps}, \mu)$.
        \algitem Solve the following constraints problem: \algcomment \label{fig:alg:peekaboo-learn-histogram::step:solve-it} (Solvable by Lemma \ref{lemma:shitty-integer-optimization})
        \begin{Codeblock*}
            \algitem $\sum_{i=1}^t N_i \le N$.
            \algitem $\sum_{i=1}^t N_i p_i \le 1$.
            \algitem $\sum_{i=1}^t \abs{N_i p_i - \tau_B(i)} \le 12 \hat{\eps}$.
            \algitem For every $1 \le i \le t$: $N_i \ge 0$ is an integer.
            \algitem For every $1 \le i \le p$: $e^{-2\hat{\eps}(i + 2)} \le p_i \le e^{-2 \hat{\eps}(i - 2)}$.
        \end{Codeblock*}
        \algitem Let $s \gets \sum_{i=1}^t N_i p_i$.
        \algitem Construct a function $f' : \Omega \to \{1,\ldots,t;\infty\}$ such that for every $1 \le i \le t$ there are exactly $N_i$ elements for which $f(x) = i$. The other elements are mapped to $\infty$.
        \algitem Construct a distribution $\tau'$ over $\Omega$ where for every $1 \le i \le t$ there are exactly $N_i$ elements whose probability mass is exactly $p_i/s$. The other $N - \sum_{i=1}^t N_i$ elements have zero mass.
        \algitem Return $\tau'$.
    \end{code}
\end{algo}

\begin{lemma} \label{lemma:peekaboo-histogram}
    Let $\mu$ be a distribution over $\Omega = \{1,\ldots,N\}$. Algorithm \ref{fig:alg:peekaboo-learn-histogram} solves the $\eps$-histogram learning task at the cost of $O(\log (N/\eps) / \eps^3)$ calls to the $\frac{1}{300}\eps$-explicit sampling oracle.
\end{lemma}
\begin{proof}
    By Lemma \ref{lemma:learning-histogram-buckets-correct}, with probability $2/3$, the call to \procnameZlearnZhistogramZbuckets returns a distribution $\tau_B$ that is $6\hat{\eps}$-close to some $2\hat{\eps}$-bucket distribution of $\mu$. Let $f_\mu$ be a $2 \hat{\eps}$-bucket function of $\mu$ for which $\tau_B$ is $6\hat{\eps}$-close to $\mu_{f_\mu}$. Lemma \ref{lemma:shitty-integer-optimization} implies that the constraint problem defined by the algorithm in Step \ref{fig:alg:peekaboo-learn-histogram::step:solve-it} is solvable.
    
    Observe that $\tau_B(\infty) \le \mu_{f_\mu}(\infty) + 6\hat{\eps} \le \frac{4\hat{\eps}^2}{N} \cdot N + 6\hat{\eps} = 6\hat{\eps} + 4\hat{\eps}^2$.
    
    Observe that $\tau'(\infty) = 0$, and note that $s = \sum_{i=1}^t N_i p_i \ge \sum_{i=1}^t \tau_B(i) - 12\hat{\eps} = (1 - \tau_B(\infty)) - 12\hat{\eps} \ge 1 - (18\hat{\eps} + 4\hat{\eps}^2)$. Also, $s \le 1$ by the constraints of the construction.
    
    Since $s = 1 \pm (18 \hat{\eps} + 4\hat{\eps}^2)$, $s^{-1} = 1 \pm 20\hat{\eps}$ and hence:
    \begin{eqnarray*}
        \sum_{i=1}^t \abs{N_i p_i / s - \tau_B(i)}
        &=& \sum_{i=1}^t \abs{(1 \pm 20\hat{\eps}) N_i p_i - \tau_B(i)} \\
        &\le& 20\hat{\eps} \sum_{i=1}^t N_i p_i + \sum_{i=1}^t \abs{N_i p_i - \tau_B(i)} \\
        &\le& 20\hat{\eps} \cdot 1 + 12\hat{\eps}
        = 32\hat{\eps}
    \end{eqnarray*}

    Considering Definition \ref{def:bucket-distribution} and Lemma \ref{lemma:bucket-function-parameter-change} with respect to $\hat{\eps}$ and $\eps' = 16 \cdot 2 \hat{\eps}$, let:
        \begin{itemize}
        \item $t' = 2 + \floor{(t - 2) / 16}$.
        \item $\nu_B$ be the map of $\tau_B$ according to $T_{2\hat{\eps}, 32\hat{\eps}}$, that is, $\nu_B(j) = \sum_{i : T_{2 \hat{\eps}, 32 \hat{\eps}}(i) = j} \tau_B(i)$.
        \item $f'$ be the $32\hat{\eps}$-bucket function of $\tau'$ with respect to $t'$, defined such that every element with mass $p_i/s$ is mapped to the $T_{2\hat{\eps},32\hat{\eps}}(i)$th bucket.
        \item $f'_\mu$ be the $32\hat{\eps}$-bucket function of $\mu$ with respect to $t'$ that is constructed from $f_\mu$ by Lemma \ref{lemma:bucket-function-parameter-change}.
    \end{itemize}

    Observe that $\nu_B$ is $6\hat{\eps}$-close to $\mu_{f'_\mu}$ since $\tau_B$ is $6\hat{\eps}$-close to $\mu_{f_\mu}$, and that $\mu_{f'_\mu}$ is a $32\hat{\eps}$-bucket distribution of $\mu$.

    By the construction, $\tau'_{f'}(\infty) = 0$ and $\tau'_{f'}(i) = N_i p_i / s$ for $1 \le i \le t$. We can use the bound for $\sum_{i=1}^t \abs{N_i p_i / s - \tau_B(i)}$ to obtain:
    
    \begin{eqnarray*}
        \dtv(\tau'_{f'},\nu_B)
        &\le& \dtv(\tau'_f,\tau_B) \\
        &=& \frac{1}{2} \tau_B(\infty) + \frac{1}{2} \sum_{i=1}^t \abs{N_i p_i/s - \tau_B(i)}
        \le \frac{1}{2} \cdot (6\hat{\eps} + 4\hat{\eps}^2) + \frac{1}{2} \cdot 32\hat{\eps} = 19\hat{\eps} + 2\hat{\eps}^2
    \end{eqnarray*}

    By the triangle inequality, $\dtv(\mu_{f'_\mu}, \tau'_{f'}) \le \dtv(\mu_{f'_\mu}, \nu_B) + \dtv(\nu_B, \tau'_{f'}) \le 6\hat{\eps} + (19\hat{\eps} + 2\hat{\eps}^2) \le 25\hat{\eps} + 4\hat{\eps}^2$.
    
    By Lemma \ref{lemma:bucket-histogram-close}, $\dhist(\mu ; \tau') \le \eps'$ for
    \[  \eps'
        = \dtv(\mu_{f'_\mu}, \tau'_{f'}) + (5(32\hat{\eps}) + 16(32\hat{\eps})^2)
        \le (25\hat{\eps} + 4\hat{\eps}^2) + (160\hat{\eps} + 16384\hat{\eps}^2)
        = 185\hat{\eps} + 16388\hat{\eps}^2
        \le 300\hat{\eps} = \eps \qedhere \]
\end{proof}

At this point we recall Theorem \ref{th:ubnd-histogram} and prove it.
\begin{theorem}[Learning histograms] \label{th:ubnd-histogram}
    We can use $O\left(\frac{\log N \log \log N}{\eps^7} \cdot \poly(\log \eps^{-1})\right)$ conditional samples to solve the $\eps$-histogram learning.
\end{theorem}
\begin{proof}
    This is an application of Lemma \ref{lemma:generic-application-single-distribution} over the $O(\log (N/\eps) / \eps^3)$-sample algorithm for learning histograms stated in Lemma \ref{lemma:peekaboo-histogram}.
\end{proof}

Since label-invariant properties are determined by histograms, we can obtain a universal tester for label-invariant properties.

\begin{corollary} \label{cor:universal-tester-label-invariant}
    There exists a universal tester for $\eps$-testing every label-invariant property $\mathcal P$ using $O(\log N / \eps^7 \cdot \poly(\log \eps^{-1}))$ conditional samples.
\end{corollary}
\begin{proof}
    We learn a distribution $\tau$ for which $\dhist(\mu ; \tau) \le \frac{1}{4}\eps$ and accept if there exists any distribution $\mu' \in \mathcal P$ for which $\dhist(\mu' ; \tau) \le \frac{1}{4}\eps$. By two applications of Lemma \ref{lemma:dtv-permutation-le-twice-dhist}, if we have accepted due to some $\tau$ then there are two permutations $\pi$ and $\pi'$ such that $d_{\mathrm{TV}}(\mu,\pi\tau)\leq\frac{1}{2}\eps$ and $d_{\mathrm{TV}}(\mu',\pi'\tau)\leq\frac{1}{2}\eps$. By the triangle inequality (and invariance under permutations) we obtain $d_{\mathrm{TV}}(\mu,(\pi \circ (\pi')^{-1})\mu')\leq\eps$ as required.
\end{proof}

\subsection{Total-variation distance estimation}
\label{subsec:dtv-est::of-sec:applications}

We first prove a generic lemma about a reduction from the $(c,\eps)$-peek model to the fully conditional model for inputs consisting of multiple distributions. The analysis here has a penalty of $O(1/\eps)$ in comparison to Lemma \ref{lemma:generic-application-single-distribution}, since we no longer assume that a queried element $x$ is drawn from the distribution it is queried from, which requires the use of Corollary \ref{cor:worst-case-complexity-of::th:estimation-task} (worst case cost) rather than Corollary \ref{cor:expected-complexity-of::th:estimation-task} (expected cost).

\begin{lemma} \label{lemma:generic-application-k-distributions}
    Consider an algorithm $\mathcal A$ whose input is a $k$-tuple $\vec{\mu} = (\mu_1,\ldots,\mu_k)$ of distributions over $\Omega_1,\ldots,\Omega_k$ (respectively), and its output is an element of a discrete set $R$. Assume that $\mathcal A$ draws at most $q$ samples and makes at most $q$ calls to the $(c,\eps)$-peek oracle. Let $N = \max\left\{ \abs{\Omega_1}, \ldots, \abs{\Omega_k}\right\}$.

    Assume that for every input $\vec{\mu}$ there exists a set $R_{\vec{\mu}} \subseteq R$ for which $\Pr\left[\mathcal A(\vec{\mu}) \in R_{\vec{\mu}}\right] > \frac{2}{3}$ for every possible valid outcome sequence of the $(c,\eps)$-oracle (i.e., one that comes from an $(\eps,c)$-approximation function corresponding to the oracle).

    In this setting, there exists an algorithm $\mathcal A'$ in the fully conditional model whose sample complexity is $O(q \cdot (\log \log N + 1 / \eps^2 c + 1 / \eps^5) \cdot \poly(\log q, \log c^{-1}, \log \eps^{-1}))$, such that for every input $\vec{\mu}$, $\Pr\left[\mathcal A'(\vec{\mu}) \in R_{\vec{\mu}}\right] > \frac{5}{8}$.
\end{lemma}
\begin{proof}
    We run $\mathcal A$ and simulate the outcome of the $\eps$-peek oracle. In each call to the $(c,\eps)$-peek oracle with $x \in \Omega_i$ and $\mu_i$, we call \procnameZmainZsingleHREF with parameters $(\mu_i,c,\eps)$ on $x_i$ (Theorem \ref{th:estimation-task}). We amplify the success probability to $1 - \frac{1}{24q}$ using the median of $\ceil{30 \ln (12q)}$ such calls (Observation \ref{obs:median-amplification}\ref{median-amplification:log-to-2c}). Each time we estimate the probability mass of an element, we record it in a ``history''. If the same element is queried again later, we use the history record rather than calling the \procnameZmainZsingleHREF procedure again. This guarantees the consistency of the oracle (required by Definition \ref{def:c-eps-peek-oracle}).

    The probability to have a wrong estimation is bounded by $q \cdot \frac{1}{24q} = \frac{1}{24}$. Hence, the probability to correctly simulate the $(c, \eps)$-peek oracle is at least $23/24$. If the simulation is correct, then the output of the simulated $\mathcal A$ belongs to $R_{\vec{\mu}}$ with probability at least $2/3$. Overall, the probability of the simulation to output an element in $R_{\vec{\mu}}$ is at least $2/3 - 1/24 = 5/8$.

    By Corollary \ref{cor:worst-case-complexity-of::th:estimation-task}, the worst-case complexity of a single estimation of $x$ is $O(\poly(\log c^{-1}, \log \eps^{-1})) \cdot O(\log \log N + \frac{1}{\eps^2 c} + \frac{\log^6 \eps^{-1}}{\eps^5})$. We repeat this $O(\log q)$ times for amplification for $q$ requests. Overall, the expected sample complexity is at most $O(q \cdot (\log \log N + 1 / \eps^2 c + 1/\eps^5) \cdot \poly(\log q, \log c^{-1}, \log \eps^{-1}))$.
\end{proof}

\begin{lemma} \label{lemma:generic-application-k-distributions-testing}
    Consider a testing algorithm $\mathcal A$ for some property $\mathcal P$ with success probability $2/3$ whose input is a $k$-tuple $\vec{\mu} = (\mu_1,\ldots,\mu_k)$ of distributions over $\Omega_1,\ldots,\Omega_k$ (respectively). Assume that $\mathcal A$ draws at most $q$ samples and makes at most $q$ calls to the $(c,\eps)$-peek oracle. Let $N = \max\left\{ \abs{\Omega_1}, \ldots, \abs{\Omega_k}\right\}$.
    
    In this setting, there exists a testing algorithm $\mathcal A'$ for $\mathcal P$ with success probability $2/3$ in the fully conditional model whose sample complexity is $O(q \cdot (\log \log N + 1 / \eps^2 c + 1 / \eps^5) \cdot \poly(\log q, \log c^{-1}, \log \eps^{-1}))$.
\end{lemma}
\begin{proof}
    Let $R = \{\accept, \reject\}$. By Lemma \ref{lemma:generic-application-k-distributions}, there exists a testing algorithm $\mathcal A''$ in the conditional model with the guaranteed sample complexity and success probability at least $5/8$. We define $\mathcal A'$ as a majority-of-$3$ amplification of $\mathcal A''$. The success probability of $\mathcal A'$ is at least $2/3$ by Observation \ref{obs:majority-amplification-ad-hoc}.
\end{proof}

Most of this subsection is dedicated to an algorithm for estimating $\dtv(\mu, \tau)$ within $\pm \eps$-additive error at the cost of $O(1/\eps^2)$ samples and $O(1/\eps^2)$ calls to the $\frac{1}{6}\eps$-peek oracle.

\begin{restatable}{lemma}{lemmaZcZtruncationZadditiveZerrZtwoc} \label{lemma:c-truncation-additive-err-2c}
    For every pair of $c$-truncated functions $f_\mu, f_\tau : \Omega \to [0,1]$ with respect to $\mu$ and $\tau$,
    \[\dtv(\mu, \tau) = \frac{1}{2}\left(\E_{x\sim \mu}\left[\max\left\{0, 1 - \frac{f_\tau(x)}{\mu(x)}\right\}\right] + \E_{x\sim \tau}\left[\max\left\{0, 1 - \frac{f_\mu(x)}{\tau(x)}\right\}\right] \right) \pm 2c \]
\end{restatable}
We defer the proof of Lemma \ref{lemma:c-truncation-additive-err-2c} to Appendix \ref{apx:tldr}.

\begin{algo}
    \procname{$\procnameZestimateZboundedZratio(\mu, \eps, f, g; \hat{f}, \hat{g})$}
    \label{fig:alg:est-bounded-ratio}
    \alginput{$f$ and $g$, inaccessible to the algorithm}
    \alginput{Oracle access to $\hat{f}(x) \in (1 \pm \eps)f(x)$ for every $x$}
    \alginput{Oracle access to $\hat{g}(x) \in (1 \pm \eps)g(x)$ for every $x$}
    \algoutput{$Y = \E_{x \sim \mu}\left[\max\left\{0, 1 - \frac{f(x)}{g(x)}\right\}\right] \pm 4\eps$}
    \algsuccpr{$5/6$}
    \begin{code}
        \algitem $M \gets \ceil{6/\eps^2}$.
        \begin{For}{$i$ from $1$ to $M$}
            \algitem $x_i \sim \mu$.
            \algitem $\hat{p}_i \gets \hat{f}(x_i)$.
            \algitem $\hat{q}_i \gets \hat{g}(x_i)$.
            \algitem $X_i \gets \min\left\{1, \frac{\hat{p}_i}{\hat{q}_i}\right\}$.
        \end{For}
        \algitem Let $\bar{X} = \frac{1}{M} \sum_{i=1}^M X_i$.
        \algitem Return $Y = 1 - \bar{X}$.
    \end{code}
\end{algo}

\begin{lemma}[\procnameZestimateZboundedZratio] \label{lemma:est-dtv-internal-random-var}
    Let $\mu$ be a distribution over $\Omega$ and let $f, g : \Omega \to [0,1]$ be two inaccessible functions. Assume that we have oracle access to functions $\hat{f}, \hat{g} : \Omega \to [0,1]$ such that for every $x \in \Omega$, $\hat{f}(x) \in (1 \pm \eps)f(x)$ and $\hat{g}(x) \in (1 \pm \eps)g(x)$. Algorithm \ref{fig:alg:est-bounded-ratio} estimates $\E_\mu\left[\max\left\{0, 1 - \frac{f(x)}{g(x)}\right\}\right]$ within $\pm 4\eps$ and success probability $5/6$, at the cost of $O(1 / \eps^2)$ oracle calls.
\end{lemma}
\begin{proof}
    It suffices to show that $\bar{X}$ estimates $\E_\mu\left[\min\left\{1, \frac{f(x)}{g(x)}\right\}\right]$ within $\pm 4\eps$-error with probability at least $5/6$.

    We explicitly bound the additive error in a single trial. If $\frac{f(x)}{g(x)} \ge \frac{1 + \eps}{1 - \eps}$, then $\frac{\hat{f}(x)}{\hat{g}(x)} \ge 1$, and hence $\min\left\{1, \frac{\hat{f}(x)}{\hat{g}(x)}\right\} = \min\left\{1, \frac{f(x)}{g(x)}\right\} = 1$. If $\frac{f(x)}{g(x)} \le \frac{1 + \eps}{1 - \eps}$, then the error $(\frac{1 \pm \eps}{1 \pm \eps} - 1) \frac{f(x)}{g(x)}$ is bounded by $\pm 3\eps$.

    For $q = \ceil{6/\eps^2}$, let $X_1,\ldots,X_q$ be independent samples of $\min\left\{1, \frac{\hat{f}(x)}{\hat{g}(x)}\right\}$, each costing two oracle calls. Let $\bar{X} = \frac{1}{q}\sum_{i=1}^q X_i$. Clearly, all $X_i$s are bounded between $0$ and $1$, hence their variance is bounded by $1$ as well. An average over $\ceil{6/\eps^2}$ trials has variance $\Var[\bar{X}] \le \frac{1}{6}\eps^2$, and hence by Chebyshev inequality, the probability to deviate by more than $\eps$ is bounded by $1/6$.

    Overall, with probability at least $5/6$,
    \[  \bar{X}
        = \E_\mu\left[\min\left\{1, \frac{\hat{f}(x)}{\hat{g}(x)}\right\}\right] \pm \eps
        = \left(\E_\mu\left[\min\left\{1, \frac{f(x)}{g(x)}\right\}\right] \pm 3\eps \right) \pm \eps
        = \E_\mu\left[\min\left\{1, \frac{f(x)}{g(x)}\right\}\right] \pm 4 \eps
        \qedhere \]
\end{proof}

\begin{algo}
    \procname{$\textsf{Estimate-$\dtv$}(\eps; \mu,\tau)$}
    \label{fig:alg:estimate-dtv}
    \algoracle{The $(\hat{\eps},\hat{\eps})$-peek oracles of $\mu$ and $\tau$ for $\hat{\eps} = \frac{1}{6}\eps$}
    \begin{code}
        \algitem Let $f_\mu$ be a non-accessible, arbitrary $\hat{\eps}$-truncated function of $\mu$, implicitly defined by the output of the peek oracle.
        \algitem Let $f_\tau$ be a non-accessible, arbitrary $\hat{\eps}$-truncated function of $\tau$, implicitly defined by the output of the peek oracle.
        \algitem Consider the following functions:
        \begin{Codeblock*}
            \item $f(x)$: $f_\mu$ (not accessible).
            \item $g(x)$: $f_\tau$ (not accessible).
            \item $\hat{f}(x)$ is the oracle call to the $(\hat{\eps},\hat{\eps})$-peek oracle in $\mu$.
            \item $\hat{g}(x)$ is the oracle call to the $(\hat{\eps},\hat{\eps})$-peek oracle in $\tau$.
        \end{Codeblock*}
        \algitem Let $X^\mu \gets \procnameZestimateZboundedZratioHREF(\mu, \hat{\eps}, f, g; \hat{f}, \hat{g})$.
        \algitem Let $X^\tau \gets \procnameZestimateZboundedZratioHREF(\tau, \hat{\eps}, g, f; \hat{g}, \hat{f})$.
        \algitem Return $\frac{1}{2}X^\mu + \frac{1}{2}X^\tau$.
    \end{code}
\end{algo}

\begin{lemma} \label{lemma:peekaboo-est-dtv}
    Let $\mu$ and $\tau$ be two distributions over $\Omega = \{1,\ldots,N\}$. Algorithm \ref{fig:alg:estimate-dtv} estimates $\dtv(\mu, \tau)$ within $\pm \eps$-additive error at the cost of $O(1/\eps^2)$ samples and $O(1/\eps^2)$ calls to the $\frac{1}{6}\eps$-peek oracle.
\end{lemma}
\begin{proof}[Proof]
    By Observation \ref{obs:c-eps-peek-as-querying-c-eps-approximation}, the $\hat{\eps}$-peek oracles (for $\mu$ and $\tau$) can be seen as query oracles of $(\hat{\eps},\hat{\eps})$-approximation functions $h_\mu$ and $h_\tau$. By Observation \ref{obs:c-eps-approximation-as-eps-approx-of-c-truncate}, these $h_\mu$ and $h_\tau$ can be seen as $(1 \pm \hat{\eps})$-approximations of $\hat{\eps}$-truncated functions $f_\mu$, $f_\tau$

    By Lemma \ref{lemma:est-dtv-internal-random-var}, with probability at least $2/3$ (a union bound over two $5/6$-success events):
    \begin{eqnarray*}
        X^\mu
        &=& \E_\mu\left[ \max\left\{0, \frac{f_\tau(x)}{f_\mu(x)} \right\} \right] \pm 4\hat{\eps} \\
        X^\tau
        &=& \E_\tau\left[ \max\left\{0, \frac{f_\mu(x)}{f_\tau(x)} \right\} \right] \pm 4\hat{\eps}
    \end{eqnarray*}

    If this happens, then:
    \begin{eqnarray*}
        \frac{1}{2}X^\mu + \frac{1}{2}X^\tau
        &=& \frac{1}{2}\E[X^\mu] + \frac{1}{2}\E[X^\tau] \pm 4\hat{\eps} \\
        \text{[Lemma \ref{lemma:c-truncation-additive-err-2c}]} &=& (\dtv(\mu,\tau) \pm 2\hat{\eps}) \pm 4\hat{\eps}
        = \dtv(\mu,\tau) \pm 6\hat{\eps}
        = \dtv(\mu,\tau) \pm \eps
    \end{eqnarray*}
\end{proof}

We now recall Theorem \ref{th:ubnd-est-dtv} and prove it.
\thZubndZestZdtv*
\begin{proof}
    This is an application of Lemma \ref{lemma:generic-application-single-distribution} over the $O(1/\eps^2)$-sample algorithm for estimating $\dtv(\mu,\tau) \pm \eps$ stated in Lemma \ref{lemma:peekaboo-est-dtv}.
\end{proof}

\subsection{Non-tolerant testing of equivalence}
\label{subsec:eps-test-equivalence::of-sec:applications}

In this subsection we present a non-tolerant $\eps$-test for equivalence of two distributions $\mu$ and $\tau$. Our result reduces the polynomial degree of $1/\eps$ in comparison to the trivial reduction from estimating the distance between $\mu$ and $\tau$ within $\pm \frac{1}{2}\eps$-additive error.

\begin{algo}
    \label{fig:alg:peekaboo-non-tolerant-core}
    \procname{$\textsf{Equivalence-Test-core}(\eps; \mu,\tau)$}
    \algoracle{The $(\frac{1}{16}\eps,\frac{1}{16}\eps)$-peek oracles for $\mu$ and $\tau$}
    
    \begin{code}
        \begin{For}{$\ceil{3/\eps}$ times}
            \algitem Draw $x \sim \mu$.
            \algitem Call the $(\frac{1}{16}\eps,\frac{1}{16}\eps)$-peek oracle for $\mu$, $x$ to obtain $\hat{p}$.
            \algitem Call the $(\frac{1}{16}\eps,\frac{1}{16}\eps)$-peek oracle for $\tau$, $x$ to obtain $\hat{q}$.
            \begin{If}{$\abs{\hat{q} / \hat{p} - 1} > \eps/4$}
                \algitem Return \reject.
            \end{If}
        \end{For}
        \algitem Return \accept.
    \end{code}
\end{algo}

\begin{lemma}[Based on \cite{RS09}] \label{lemma:peekaboo-non-tolerant-new}
    Let $\mu$ and $\tau$ be two distributions over $\Omega = \{1,\ldots,N\}$. Algorithm \ref{fig:alg:peekaboo-non-tolerant-core} distinguishes between $\tau=\mu$ and $\dtv(\tau,\mu) > \eps$ using $O(1/\eps)$ independent samples from $\mu$ and $O(1/\eps)$ calls to the $(\frac{1}{16}\eps, \frac{1}{16}\eps)$-peek oracle.
\end{lemma}
\begin{proof}
    Let $\hat{\eps} = \frac{1}{16}\eps$. By Observation \ref{obs:c-eps-peek-as-querying-c-eps-approximation}, the $(\hat{\eps},\hat{\eps})$-peek oracles for $\mu$ and $\tau$ can be seen as query oracles to $(\hat{\eps}, \hat{\eps})$-approximation functions $f_\mu$ and $f_\tau$ (respectively).
    
    If $\mu = \tau$ and $\CDF_\mu(x) \ge \hat{\eps}$, then we expect that $\abs{\frac{f_\tau(x)}{f_\mu(x)} - 1} = \abs{\frac{1 \pm \hat{\eps}}{1 \pm \hat{\eps}} - 1} \le 3\hat{\eps} < \frac{1}{4}\eps$.

    If $\tau(x) < (1 - \eps / 2)\mu(x)$ and $\CDF_\mu(x) \ge \hat{\eps}$, then we expect that $1 - \frac{f_\tau(x)}{f_\mu(x)} \ge 1 - \frac{1 \pm \hat{\eps}}{1 \pm \hat{\eps}} \cdot \left(1 - \frac{1}{2}\eps\right) > \frac{1}{4}\eps$.

    If $\mu = \tau$, then in every iteration, the probability to draw $x$ for which $\CDF_\mu(x) \ge \hat{\eps}$ is at least $1 - \hat{\eps}$. By the union bound, the probability to reject is at most $\ceil{3/\eps} \cdot \hat{\eps} \le \frac{4}{\eps} \cdot \frac{1}{16}\eps = \frac{1}{4}$.

    For $\dtv(\mu, \tau) > \eps$, let $A = \{ x : \tau(x) < (1 - \eps/2)\mu(x) \}$. By definition of the total variation distance,
    \begin{eqnarray*}
        \dtv(\tau,\mu)
        = \sum_{\mu(x) > \tau(x)} (\mu(x) - \tau(x))
        &=& \sum_{\mu(x) > \tau(x)} \mu(x)\left(1 - \frac{\tau(x)}{\mu(x)}\right) \\
        &=& \E_{x\sim \mu}\left[\max\left\{0, 1 - \frac{\tau(x)}{\mu(x)}\right\}\right] \\
        &\le& \mu(A) \cdot 1 + \mu(\neg A) \cdot \frac{1}{2}\eps
        \le \mu(A) + \frac{1}{2}\eps
    \end{eqnarray*}
    Since $\dtv(\tau,\mu) > \eps$, we obtain that $\mu(A) > \frac{1}{2}\eps$. The probability to draw an $A$-sample from $\mu$ which has $\CDF_\mu(x) \ge \hat{\eps}$, and hence reject, is at least $1 - \left(1 - (\frac{1}{2}\eps - \hat{\eps})\right)^{\ceil{3/\eps}} = 1 - \left(1 - \frac{7}{16}\eps\right)^{\ceil{3/\eps}} \ge 1 - e^{-21/16} > \frac{2}{3}$.
\end{proof}

\begin{lemma} \label{lemma:almost-ubnd-nontol-equivalence}
    Let $\mu$, $\tau$ be two distributions over $\Omega = \{1,\ldots,N\}$. There exists an algorithm for distinguishing, with probability at least $5/8$, between the case where $\mu=\tau$ and the case where $\dtv(\mu,\tau) > \eps$ using:
    \begin{itemize}
        \item $O((\log \log N / \eps + 1/\eps^6) \cdot \poly(\log \eps^{-1}))$ conditional samples at worst-case.
        \item $O((\log \log N / \eps + 1/\eps^5) \cdot \poly(\log \eps^{-1}))$ conditional samples in expectation if $\mu=\tau$.
    \end{itemize}
\end{lemma}
\begin{proof}
    We simulate every peek call of Algorithm \ref{fig:alg:peekaboo-non-tolerant-core} using \procnameZmainZsingleHREF and amplify its confidence by taking the median of $O(\log \eps^{-1})$ independent calls. Lemma \ref{lemma:generic-application-k-distributions} implies the correctness and the worst-case complexity of this reduction. Since the success probability of the core algorithm is at least $2/3$ and the reduction error is at most $1/24$, the success probability of the result algorithm is at least $5/8$.
    
    Observe that if $\mu=\tau$ then 
    every call to \procnameZmainZsingleHREF (for answering a peek call) is performed on a value $x$ that was sampled from a distribution that is identical to the one for which it is queried. This allows the use of Corollary \ref{cor:expected-complexity-of::th:estimation-task} to obtain the expected-case complexity of the reduction when $\mu=\tau$.
\end{proof}

We recall Theorem \ref{th:ubnd-nontol-equivalence} and prove it.
\thZubndZnontolZequivalence*

\begin{proof}
    Let $Q = O((\log \log N / \eps + 1/\eps^5) \cdot \poly(\log \eps^{-1}))$ be the expected number of samples of the algorithm that is guaranteed by Lemma \ref{lemma:almost-ubnd-nontol-equivalence} in the case where $\mu=\tau$ (for any choice of $\mu$). Consider the following algorithm: we run $45$ independent instances of the algorithm of Lemma \ref{lemma:almost-ubnd-nontol-equivalence} and take the majority answer, with the exception that if we make our $(540Q+1)$st sample, then we immediately reject and terminate the run.

    If $\mu = \tau$, then by Markov's inequality, the probability to draw the $(540Q+1)$st sample is smaller than $1/12$. Hence, with probability $11/12$ the core algorithm runs successfully, and by Observation \ref{obs:majority-amplification-ad-hoc}, it accepts with probability at least $3/4$. By the union bound, we accept with probability at least $2/3$.
    
    If $\dtv(\mu, \tau) > \eps$, then either we terminate after $540Q+1$ queries and reject, or the core algorithm runs successfully and rejects with probability at least $3/4$. The probability to reject is at least $3/4 > 2/3$.
\end{proof}

\section{Lower bounds}
\label{sec:lbnds}

\subsection{Tight lower bound for the $(c,\eps)$-estimation task}
\label{subsec:tight-lbnd-c-eps-estimation::of-sec:lbnds}

In this subsection we show a tight lower bound of $\Omega(\log \log N)$ for estimating the probability mass of individual elements using conditional samples. Since the demonstrating distributions are uniform over their support, the expected case and the worst case are identical.

For some integer $1 \le k \le \log N$, we define $D_k$ as the following distribution over inputs (in themselves distributions over $\{1,\ldots,N\}$): we draw a set $K \subseteq \{1,\ldots,N\}$ such that every element belongs to $K$ with probability $2^{-k}$ independently, and then return the uniform distribution over $K$.
\begin{lemma}
    Let $1 \le k \le \log N - \log \log N - 8$. With probability $1 - o(1/N)$ over the drawing of $\mu$ from $D_k$, every element in the support of $\mu$ has mass in the range $\left(1 \pm \frac{1}{9}\right) \cdot \frac{2^k}{N}$.
\end{lemma}
\begin{proof}
    By Chernoff's bound, $\Pr\left[\abs{K} \in \left(1 \pm \frac{1}{10}\right) 2^{-k} N\right] \ge 1 - 2e^{-\frac{1}{300} \cdot 2^{-k} N} = 1 - o(1/N)$. If this happens, then every element in the support of $\mu$ has probability mass $\frac{1}{\abs{K}} = \frac{1}{1 \pm 1/10} \cdot \frac{2^k}{N} = \left(1 \pm \frac{1}{9}\right) \frac{2^k}{N}$.
\end{proof}

Let $k_\mathrm{min} = \floor{\frac{1}{3}\log N}$ and $k_\mathrm{max} = \ceil{\frac{2}{3}\log N}$. We use $k_\mathrm{min}$ and $k_\mathrm{max}$ to define the ``composed'' distribution over inputs: $D$ draws $k$ uniformly in the range $\{k_\mathrm{min},\ldots,k_\mathrm{max}\}$ and then returns the pair $(k,\mu)$ where $\mu$ is an input distribution drawn from $D_k$.

\begin{observation} \label{obs:estimation-to-calc-k}
    Let $(k,\mu) \sim D$. A conditional-sampling algorithm that draws $x \sim \mu$ and estimates $\mu(x)$ within $1 \pm \frac{1}{9}$-factor with probability at least $p$ can correctly obtain $k$ with probability at least $p - o(1)$ for sufficiently large $N$.
\end{observation}
\begin{proof}
    If $N$ is sufficiently large, then $k_\mathrm{max} \le \log N - \log \log N - 8$. With probability $1-o(1/N) = 1-o(1)$, the mass of individual elements in $\mu$ is in the range $\left(1 \pm \frac{1}{9}\right)\frac{2^k}{N}$. Hence, with probability $1-p-o(1)$, the algorithm obtains an estimation $\hat{p} \in \left(1 \pm \frac{1}{9}\right)\left(1 \pm \frac{1}{9}\right)\frac{2^k}{N} = \left(1 \pm \frac{1}{4} \right)\frac{2^k}{N}$.
    
    In this case, the algorithm can retrieve $k$ using $\hat{k} =\mathrm{round}(\log (N/\hat{p}))$ since $\log (N/\hat{p}) = \log 2^k + \log \left(1 \pm \frac{1}{4}\right) = k \pm 0.42$. Hence, the rounding of $\log (N/\hat{p})$ to the nearest integer results in $k$.
\end{proof}

By Yao's principle \cite{yao77}, every probabilistic algorithm can be seen as a distribution over deterministic algorithms, and a lower bound against all deterministic algorithms using a single distribution over inputs translates to a lower bound against all probabilistic algorithms. A deterministic querying algorithm can be characterized as a decision tree, where every internal node (including the root) holds a query, and every edge corresponds to a possible outcome.

\subsubsection*{Our interim models and additional notations}

For our lower bound we investigate the relationship of three models. The first is not related to distributions at all, and is just a model for the plain binary search task for a value $k$ that is drawn uniformly from the set $I=\{k_\mathrm{min},\ldots,k_\mathrm{max}\}$. The second model, a uniform ``conditional'' sampling model, uses the responses to the comparison queries with $k$ to provide additional simulated responses to a conditional sampling oracle, although at this point no actual distribution is used.

The third model, a ``leaking'' conditional sampling model, draws a distribution $\mu$ over $\Omega=\{1,\ldots,N\}$ (whose size is $2^{\Theta(|I|)}$) using $D_k$, and complements the comparison queries with actual conditional samples. In particular, the expressiveness of algorithms under this last model is at least as strong as the expressiveness of algorithms that only take conditional samples from $\mu$. By Observation \ref{obs:estimation-to-calc-k}, an estimation of an element drawn from $\mu$ with high probability reveals the value of $k$. To finalize, we show that the behavior of an algorithm under the leaking model is very close to its behavior under the uniform model (which is fully simulated from just the comparison queries), and hence a working estimation algorithm provides an algorithm that with high probability solves the binary search problem for $k$. This implies the lower bound of $\Omega(\log |I|)=\Omega(\log\log N)$.

\begin{definition}[The $n$-range binary search model]
    For a parameter $n$ and a fixed well-ordered set $I$ of size $n$, the input of the algorithm is some $k \in I$, which is inaccessible. In every step, the algorithm chooses some $s$ and queries the predicate ``$s \le k$''. In the end, the algorithm chooses $k' \in I$. The algorithm \emph{succeeds} if $k' = k$.
\end{definition}

The following observation is well known, and easy to prove by considering the possible number of leaves of a bounded depth binary tree.

\begin{observation} \label{obs:binary-search-lbnd}
    Every algorithm in the $n$-range binary search model whose success probability is strictly greater than $1/2$, over a uniformly random choice of $k\in I$, must make $q > \log n - 1$ queries.
\end{observation}

We define a common framework for the two conditional sampling models that we define shortly: the uniform conditional model and the leaking conditional model.

\begin{definition}[Common framework for conditional sampling models]
    For a parameter $N$, the input of the algorithm is $k \in \{k_\mathrm{min},\ldots,k_\mathrm{max}\}$, which is inaccessible, and a distribution $\mu$ over $\{1,\ldots,N\}$. In every step, the algorithm chooses a non-empty subset $B \subseteq \{1,\ldots,N\}$ and receives a pair $(b,y)$ where $b \in \{ \text{``$\le$''}, \text{``$>$''} \}$ and $y \in B \cup \{\text{``err''}\}$. The behavior of $(b,y)$ given $k$, $\mu$, $B$ and the execution path so-far is determined by the specific model.
\end{definition}
Note that the upper-bound algorithm for the estimation task in this paper interfaces with a model that has no $b$ component in the answers to its queries. The leaking model that we define below provides conditional query access to a specific drawn distribution along with some additional information given through the additional component. Also, the leaking model does not require a logical guarantee that $B$ has strictly positive probability mass in the input distribution $\mu$ (a guarantee that our upper-bound algorithm satisfies). The option for an $\text{``err''}$ answer for the $y$ component is used by the leaking model to also handle zero probability condition sets.

Every algorithm in the common framework defined above can be described as a decision tree whose internal nodes (including the root) hold the condition set $B$ and whose edges are labeled with the possible outcomes $(b,y)$.

\begin{definition}[Characterization of a decision node]
    A decision node $u$ of a decision tree $A$ is characterized by:
    \begin{itemize}
        \item $\ell_u$ (short form: $\ell$), the node-distance of $u$ from the root. ($\ell=1$ for the root).
        \item A sequence $(b_{u,1},y_{u,1}),\ldots,(b_{u,\ell_u-1},y_{u,\ell_u-1})$ (short form: $(b_1,y_1),\ldots,(b_{\ell-1},y_{\ell-1})$) describing the path from the root to $u$.
        \item A non-empty condition set $B_u$ (short form: $B$).
    \end{itemize}
\end{definition}

\begin{definition}[Set of already-seen elements, $Y_\mathrm{old}$]
    Let $u$ be a decision node characterized by $(\ell_u, (b_{u,i},y_{u,i})_{1\le i \le \ell_u-1},B_u)$. The \emph{set of already-seen elements} is $Y_\mathrm{old}(u) \overset{\mathrm{def}}= \{y_{u,1},\ldots,y_{u,\ell_u-1}\} \setminus \{\text{``err''}\}$ (short form: $Y_\mathrm{old}$).
\end{definition}

\begin{definition}[Set of ruled-out elements, $Y_\mathrm{out}$]
    Let $u$ be a decision node characterized by $(\ell_u, (b_{u,i},y_{u,i})_{1\le i \le \ell_u-1},B_u)$. The \emph{set of ruled-out elements} is $Y_\mathrm{out}(u) \overset{\mathrm{def}}= \bigcup_{1 \le j \le \ell_u-1 : y_{u,j} = \text{``err''}} B_{u_j}$ (short form: $Y_\mathrm{out}$).
\end{definition}

The following defines the set of elements for which a node query can provide new information.

\begin{definition}[Net condition set, net condition size]
    Let $u$ be a decision node characterized by $(\ell_u, (b_{u,i},y_{u,i})_{1\le i \le \ell_u-1},B_u)$. The \emph{net condition set} of $u$ is $B'_u \overset{\mathrm{def}}= B \setminus (Y_\mathrm{old}(u) \cup Y_\mathrm{out}(u))$ (short form: $B'$). The \emph{net condition size} of $u$ is $s_u = \abs{B'_u}$ (short form: $s$).
\end{definition}

Based on the above notations we define the uniform conditional model and the leaking conditional sampling model.

\begin{definition}[The uniform conditional model]
    This model is based on the framework for conditional sampling models. Let $u$ be a decision node characterized by $(\ell, (b_i,y_i)_{1\le i \le \ell-1}, B)$. The behavior of $(b,y)$, which is the outcome of the query to be made by $u$, is defined as follows:
    \begin{itemize}
        \item If $s_u \le 2^k$, then the outcome of the query is $(\text{``$\le$''}, y)$ for $y$ uniformly drawn from $B_u \cap Y_\mathrm{old}(u)$ if it is not empty, and otherwise it is $(\text{``$\le$''},\text{``err''})$.
        \item If $s_u > 2^k$, then the outcome of the query is $(\text{``$>$''}, y)$ for $y$ uniformly drawn from $B_u \setminus Y_\mathrm{out}(u)$ if it is not empty, and otherwise it is $(\text{``$>$''},\text{``err''})$.
    \end{itemize}
    In the end, the algorithm chooses $k' \in \{1,\ldots,n\}$. The algorithm \emph{succeeds} if $k' = k$.
\end{definition}

Note that this is essentially a simulation model, as it gives its query answers without taking any samples from $\mu$. The following lemma indeed connects it to the ``pure binary search'' model.

\begin{lemma} \label{lemma:uniform-condition-to-binary-search}
    Every $q$-query algorithm in the uniform conditional model is behaviorally identical to a $q$-query probabilistic algorithm in the $n$-range binary search model, where $n = k_\mathrm{max} - k_\mathrm{min} + 1$. Specifically, such an algorithm is equivalent to a distribution over (deterministic) binary decision trees that only use queries on whether $s\leq k$ for some $s$ (i.e., use only the $b$ components of the answers provided by the uniform conditional model).
\end{lemma}
\begin{proof}
    Consider a decision tree in the common conditional framework, in which every edge is labeled by a pair $(b,y)$ for some $b\in\{ \text{``$\le$''}, \text{``$>$''} \}$ and $y\in B_u$. For every node $u$ in the decision tree, consider the possible distributions over its children under the uniform conditional model conditioned on the value of $b$ and on the algorithm reaching this node.

    If $b=\text{``$\le$''}$, this means in particular that $s_u\leq 2^k$. If $B_u \cap Y_\mathrm{old}(u)=\emptyset$ then the edge labeled by $(\text{``$\le$''},\text{``err''})$ is taken with probability $1$, and otherwise the outgoing edge is chosen uniformly from the set of edges whose labels are in the set $\{(\text{``$\le$''},y):y\in B_u \cap Y_\mathrm{old}(u)\}$.

    If $b=\text{``$>$''}$, this means in particular that $s_u>2^k$. If $B_u \setminus Y_\mathrm{out}(u)=\emptyset$ then the edge labeled by $(\text{``$>$''},\text{``err''})$ is taken with probability $1$, and otherwise the outgoing edge is chosen uniformly from the set of edges whose labels are in the set $\{(\text{``$>$''},y):y\in B_u \setminus Y_\mathrm{out}(u)\}$.

    The common theme here is that the identity of $u$ and the value of $b$ by themselves determine a set of outgoing edges, from which one is uniformly picked, without any dependency on the other parameters of the input. This means that a run of this $q$-query decision tree can be alternatively described by the following process:
    \begin{itemize}
        \item For every node $u$ of the tree, one edge is picked uniformly from the set of the relevant outgoing edges with $b=\text{``$\le$''}$, and one edge is picked uniformly from the set of the relevant outgoing edges with $b=\text{``$>$''}$.
        \item Then, all edges in the tree that were not picked in the previous step are removed, after which all nodes with no remaining path to the root are removed as well. In the remaining tree, the edge labels are trimmed to include only the $b$ component, which refers to a comparison of some $k'=\lceil\log s\rceil$ with $k$. The result of this process is a deterministic binary decision tree that can be run under the binary search model.
        \item Finally, the resulting tree is run with respect to the input $k$. \qedhere
    \end{itemize}
\end{proof}

\begin{definition}[The leaking conditional sampling model]
    This model is based on the framework for conditional sampling models. Let $u$ be a decision node characterized by $(\ell, (b_i,y_i)_{1\le i \le \ell-1}, B)$. The behavior of $(b,y)$, which is the outcome of the query to be made by $u$, is defined as follows:
    \begin{itemize}
        \item $b$ is ``$\le$'' if $s_u \le 2^k$ and ``$>$'' if $s_u > 2^k$.
        \item If $\mu(B) > 0$ then $y$ is drawn from $\mu$ when conditioned on $y \in B$, and otherwise $y = \text{``err''}$.
    \end{itemize}
\end{definition}

Clearly, the leaking conditional sampling model is not weaker than any reasonable variant of the classic conditional sampling model, and hence it is suitable for lower bound statements. All such models behave the same when $\mu(B) > 0$, but the fallback behavior when $\mu(B) = 0$ is explicitly defined by every model. The return of an error message in the case where $\mu(B) = 0$ provides the most information among common fallback behaviors (uniform choice, minimum, etc.), which makes it the best choice for lower bound statements.

In the following, $A$ is a decision tree of height $q$ corresponding to a deterministic algorithm in the common framework of conditional sampling models. Our random variables are:
\begin{itemize}
    \item $u_1,\ldots,u_{q+1}$ -- the nodes on the execution path. 
    \item $u$ -- an alias for the leaf $u_{q+1}$.
    \item $(b_1,y_1),\ldots,(b_q,y_q)$ -- the outcomes of the queries. In other words, for every $1 \le i \le q$, $(b_i,y_i)$ is the label on the edge from $u_i$ to $u_{i+1}$. Note that $y_1,\ldots,y_q$ are generally random even that the analyzed algorithm is deterministic.
    \item $K$ -- the support of the input distribution that is drawn according to $D_k$ (following the random choice of $k$). It plays a role only in the analysis of the leaking model.
\end{itemize}

In the following, we refer to the set $\Lambda$ that refers to the combination of the choice of $k$, the support $K$ (relevant for the leaking model), and the outcome (the leaf reached) of a run of the given deterministic algorithm. The two distributions that we analyze over $\Lambda$ are $\mathcal U$, the one resulting from the uniform model, and $\mathcal L$, the one resulting from the leaking model.

In particular, note the following well-known common bound.

\begin{lemma} \label{lemma:adaptive-yao-template-bound}
    Let $\mathcal U$ and $\mathcal L$ be two distributions over $\Lambda$. If there exists an event $E \subseteq \Lambda$ for which $\mathcal L(x) > (1 - \eps)\mathcal U(x)$ for every $x \in E$, then $\dtv(\mathcal U,\mathcal L) \le \eps + \Pr_{\mathcal U}[\neg E]$.
\end{lemma}

We also use some shorthand. In particular, a set of leaves (or more generally, nodes) of the analyzed algorithm (given as a decision tree) is identified with the event of reaching a node from this set. Also, the notation $\Pr_{\mathcal U}[E|k]$ refers to the probability of an event $E$ (usually given by a set of leaves) when conditioned on the event of the specific $k$ being drawn from the range $\{k_\mathrm{min},\ldots,k_\mathrm{max}\}$.

\subsubsection*{Analysis for the uniform model}

\begin{definition}[The set of improbable elements, $A_\mathrm{small}$]
    Let $u$ be a decision node characterized by $(\ell,(b_i,y_i)_{1\le i \le \ell-1}, B)$. Let $u_1,\ldots,u_\ell$ be its path from the root (where $u_1$ is the root and $u_\ell = u$). The \emph{set of small elements} with respect to $u$ and some $k$ is $A_\mathrm{small}(u,k) = \bigcup_{1 \le i \le \ell : s_{u_i} \le 2^k / \branchZlbndZexpZdiv} B'_{u_i}$ (short form: $A_\mathrm{small}$).
\end{definition}

\begin{definition}[The good events, $\mathcal G_k$, $\mathcal G^{(1)}_k$, $\mathcal G^{(2)}_k$, $\mathcal G^{(3)}_k$]
    Let $u$ be a leaf. Let $u_1,\ldots,u_q,u_{q+1}$ be its path from the root (where $u_1$ is the root and $u_{q+1} = u$). We define the following good events about $u$:
    \begin{itemize}
        \item $\mathcal G^{(1)}_k$: for every $1 \le i \le q$, $s_{u_i} \notin \left(\frac{1}{\branchZlbndZexpZdiv}, \branchZubndZexp\right) \cdot 2^k$.
        \item $\mathcal G^{(2)}_k$: for every $1 \le i \le q$, $y_i \notin A_\mathrm{small}(u_i,k)$.
        \item $\mathcal G^{(3)}_k$: for every $1 \le i \le q$, if $s_{u_i} \ge 2^k \cdot \branchZubndZexp$, then $y_i \in B_{u_i} \setminus Y_\mathrm{old}(u_i)$.
        \item $\mathcal G_k$: the intersection $\mathcal G^{(1)}_k \wedge \mathcal G^{(2)}_k \wedge \mathcal G^{(3)}_k$.
    \end{itemize}
\end{definition}

\begin{lemma}
\label{lemma:tons-of-locally-good-k-values}
    Let $A$ be a decision tree representing a deterministic algorithm in the common framework for conditional sampling algorithms that makes $q \le \branchZubndZq$ queries. There exists a set $G \subseteq \{k_\mathrm{min}, \ldots, k_\mathrm{max}\}$ of size at least $\left(1 - \branchZfracZbadZkZchoices\right)n$, for $n = k_\mathrm{max} - k_\mathrm{min} + 1$, such that for every $k \in G$, considering the (random) leaf $u$ that the execution path reaches in the uniform model, $\Pr_{\mathcal U}\left[u \in \mathcal G^{(1)}_k \cond k \right] \ge 1 - \frac{\log \log \log N}{15 \log \log N}$.
\end{lemma}
\begin{proof}
    As observed in Lemma \ref{lemma:uniform-condition-to-binary-search}, a decision tree in the uniform model behaves as a distribution over deterministic binary search trees, where every node $u$ in such a tree corresponds to a comparison of $k$ with some $s_u$, receiving an answer $b\in \{ \text{``$\le$''}, \text{``$>$''} \}$.

    A binary decision tree of edge-height $q \le \branchZubndZq$ has exactly $2^q - 1 < \frac{\log N}{(\log \log N)^2}$ decision nodes. For every decision node $u_i$, there are at most $\ceil{\log \branchZlbndZexpZdiv + \log \branchZubndZexp} \le 5 \log \log N$ ``bad'' choices of $k$ for which $s_{u_i} \in \left(\frac{1}{\branchZlbndZexpZdiv}, \branchZubndZexp\right) \cdot 2^k$. Considering the whole tree, there are at most $\frac{5 \log N}{\log \log N}$ such bad choices.

    For every $k$, let $p_k$ be the probability to choose a binary tree that $k$ is bad with respect to it. By linearity of expectation, $\sum_{k = k_\mathrm{min}}^{k_\mathrm{max}} p_k$ is the expected number of $k$s that are bad with respect to the drawn binary tree, which is bounded by $\frac{5 \log N}{\log \log N}$.
    
    For a uniform drawing of $k$ between $k_\mathrm{min}$ and $k_\mathrm{max}$, $\E_k\left[p_k\right] \le \frac{5 \log N / \log \log N}{n} \le \frac{15 \log N / \log \log N}{\log N} = \frac{1}{15 \log \log N}$. The last transition is correct since $n = k_\mathrm{max} - k_\mathrm{min} + 1 \ge \frac{1}{3}\log N$.

    By Markov's inequality, there are at most $\frac{n}{\log \log \log N}$ choices of $k$ for which $p_k \ge \frac{\log \log \log N}{15 \log \log N}$.
\end{proof}

\begin{lemma} \label{lemma:tons-of-ks-with-hp-good-leaf}
    Let $A$ be a decision tree representing a deterministic algorithm in the common framework for conditional sampling algorithms that makes $q \le \branchZubndZq$ queries. There exists a set $G \subseteq \{k_\mathrm{min},\ldots,k_\mathrm{max}\}$ of size at least $\left(1 - \branchZfracZbadZkZchoices \right)n$, for $n = k_\mathrm{max} - k_\mathrm{min} + 1$, such that for every $k \in G$, considering the (random) leaf $u$ that the execution path reaches in the uniform model, $\Pr_{\mathcal U}\left[u \in \mathcal G_k \cond k\right] \ge 1 - \branchZfracZbadZleafZforZgoodZk$.
\end{lemma}
\begin{proof}
    We use the set $G$ provided by Lemma \ref{lemma:tons-of-locally-good-k-values}. For $k\in G$, with probability at least $1 - \frac{\log \log \log N}{15 \log \log N}$, we reach a leaf $u$ (on the $q+1$st level) for which, for every $1 \le i \le q$, $s_{u_i} \notin \left(\frac{1}{\branchZlbndZexpZdiv}, \branchZubndZexp\right) \cdot 2^k$. Also, for every $1 \le i \le q$, the probability that $y_i \in A_\mathrm{small}(u_i,k)$ is bounded by:
    \begin{itemize}
        \item Zero if $s_{u_i} \le 2^k / \branchZlbndZexpZdiv$, by definition of the model.
        \item $\frac{\abs{A_\mathrm{small}(u_i,k)}}{\abs{B_{u_i} \cap Y_\mathrm{old}(i)} + s_{u_i}} \le \frac{q}{8 \log^4 N}$ if $s_{u_i} \ge 2^k \cdot \branchZubndZexp$.
    \end{itemize}

    Also, if $s_{u_i} \ge 2^k \cdot \branchZubndZexp$, then the probability to obtain an already-seen $y_i$ is bounded by $\frac{\abs{B_{u_i} \cap Y_\mathrm{old}(i)}}{\abs{B_{u_i} \cap Y_\mathrm{old}(i)} + s_{u_i}} \le \frac{q}{q + 2^k \cdot \branchZubndZexp} \le \frac{1}{\log N}$.

    By a union bound, the probability of $u$ to be good is at least $1 - \left(\frac{\log \log \log N}{15 \log \log N} + \frac{q^2}{8 \log^4 N} + \frac{q}{\log N}\right) \ge 1 - \branchZfracZbadZleafZforZgoodZk$ for $N$ large enough.
\end{proof}

\subsubsection*{Analysis for the leaking model}

\begin{lemma} \label{lemma:A-small-disjoint-olds}
    Let $k$ be fixed and let $u \in \mathcal G_k$ be a leaf (node of depth $q+1$) whose path from the root is $u_1,\ldots,u_q,u_{q+1}$. For every $1 \le i \le q$, $A_\mathrm{small}(u_i,k)$ is disjoint from $Y_\mathrm{old}(u)$.
\end{lemma}
\begin{proof}
    Let $(b_1,y_1),\ldots,(b_q,y_q)$ be the outcome sequence. By definition, $Y_\mathrm{old} = \{y_1,\ldots,y_q\} \setminus \{\text{``err''}\}$.
    
    For $i \le j$, the definition of $\mathcal G_k$ eliminates the possibility that $y_i \in A_\mathrm{small}(u_j,k)$.

    For $i > j$, if $y_i \in A_\mathrm{small}(u_j,k)$, then due to monotonicity, $y_i \in A_\mathrm{small}(u_i,k)$ as well, a contradiction to the definition of $\mathcal G_k$.
\end{proof}

\begin{observation} \label{obs:pr-good-u-given-k-in-the-uniform-model}
    For a given $k$ and a leaf $u \in \mathcal G_k$ whose path from the root is $u_1,\ldots,u_q,u_{q+1}$,
    \[\Pr_{\mathcal U}\left[u | k\right] = \prod_{i=1}^q \begin{cases}
        \begin{array}{>{\displaystyle}c>{\displaystyle}l>{\displaystyle}l}
            1 & s_{u_i} \le 2^k / \branchZlbndZexpZdiv, & B_{u_i} \cap Y_\mathrm{old} = \emptyset \\
            \frac{1}{\abs{B_{u_i} \cap Y_\mathrm{old}}} & s_{u_i} \le 2^k / \branchZlbndZexpZdiv, & B_{u_i} \cap Y_\mathrm{old} \ne \emptyset \\
            \frac{1}{\abs{B_{u_i} \setminus Y_\mathrm{out}}} & s_{u_i} \ge 2^k \cdot \branchZubndZexp
        \end{array}
    \end{cases}\]
\end{observation}
\begin{proof}
    If $s_{u_i} \le 2^k / \branchZlbndZexpZdiv$, then by definition, $y_i$ is uniformly drawn from  $B_{u_i} \cap Y_\mathrm{old}$, unless this intersection is empty, and in this case $y_i$ is ``$\mathrm{err}$'' with probability $1$.

    If $s_{u_i} \ge 2^k \cdot \branchZubndZexp$, then since $Y_\mathrm{out}(u_i) \subseteq A_\mathrm{small}(u_i,k)$, $\abs{B_{u_i} \setminus Y_\mathrm{out}} \ge \left(1 - \frac{1}{8 \log N}\right) \abs{B_{u_i}} > 0$. Hence, by definition $y_i$ is uniformly drawn from $B_{u_i} \setminus Y_\mathrm{out}\neq\emptyset$.
\end{proof}

\begin{lemma} \label{lemma:good-u-is-multiplicative-large}
    For every $k\in\{k_\mathrm{min},\ldots,k_\mathrm{max}\}$ and leaf $u \in \mathcal G_k$ (which is good with respect to $k$), $\Pr_{\mathcal L}\left[u \cond k\right] \ge \left(1 - \frac{3q}{\log N} \right)\Pr_{\mathcal U}\left[u \cond k\right]$.
\end{lemma}
\begin{proof}
    Let $u_1,\ldots,u_q,u_{q+1}$ be the path from the root $u_1$ to $u_{q+1}=u$, and let $(b_1,y_1),\ldots,(b_q,y_q)$ be the sequence of answers in this path. Let $t$ be the number of indexes for which $s_{u_i} \ge 2^k \cdot \branchZubndZexp$. Recall that since $u \in \mathcal G_k$, $t$ is also the number of unique non-error values in $y_1,\ldots,y_q$.
    
    We define the following good events corresponding to an input distribution $\mu \sim D_k$ (which is fully determined by the set $K$):
    \begin{itemize}
        \item $G_1$: $\{y_1,\ldots,y_q\}\setminus \{\text{``err''}\} \subseteq K$.
        \item $G_2$: $A_\mathrm{small}(u_q,k)$ is disjoint from $K$ (in the rest of this proof, non-indexed instances of $A_\mathrm{small}$ refer to this set).
        \item $G_3$: for every $1 \le i \le q$, if $s_{u_i} \ge 2^k \cdot \branchZubndZexp$, then $\abs{B_{u_i} \cap K} \le \left(1 + \frac{1}{\log N}\right)2^{-k}\abs{B_{u_i}  \setminus Y_\mathrm{out}}$.
    \end{itemize}

    By the chain rule,
    \[  \Pr_{\mathcal L}\left[u \cond k \right]
        \ge \Pr_{\mathcal L}\left[u \wedge G_1 \wedge G_2 \wedge G_3 \cond k \right]
        = \Pr_{\mathcal L}\left[G_1 \cond k \right] \cdot \Pr_{\mathcal L}\left[G_2 \wedge G_3 \cond k, G_1 \right] \cdot \Pr_{\mathcal L}\left[u \cond k, G_1, G_2, G_3 \right]
    \]

    Clearly, $\Pr_{\mathcal L}[G_1|k] = (2^{-k})^t$.

    For bounding the probability for $G_2$, we note that by Markov's inequality, the probability that $A_\mathrm{small}$ is disjoint from $K$ is at least $1 - 2^{-k} \cdot \abs{A_\mathrm{small}} \ge 1 - \frac{q}{\log N}$. This holds also when conditioned on $G_1$ since $A_\mathrm{small} \cap \{y_1,\ldots,y_q\} = \emptyset$ (Lemma \ref{lemma:A-small-disjoint-olds}).

    For nodes with $s_{u_i} \ge 2^k \cdot \branchZubndZexp$, since $Y_\mathrm{out}(u_i) \subseteq A_\mathrm{small}(u_i,k)$, we obtain that $\abs{B_{u_i} \setminus Y_\mathrm{out}} \ge \left(1 - \frac{1}{8 \log N}\right) \abs{B_{u_i}}$ (and by definition of $\mathcal G_k$ we cannot get $y_i=\text{``err''}$ for such nodes).
    
    By Chernoff's inequality and a union bound,
    \begin{eqnarray*}
        \Pr_{\mathcal L}\left[\neg G_3 \cond k,G_1 \right]
        &\le& q \cdot \Pr\left[\Bin\left(\abs{B_u} - \left(t + \abs{A_\mathrm{small}}\right), 2^{-k}\right) + t > \left(1 + \frac{1}{\log N}\right) 2^{-k}\abs{B_u \setminus Y_\mathrm{out}} \right] \\
        (*) &\le& q \cdot \Pr\left[\Bin\left(\abs{B_u}, 2^{-k}\right) + t > \left(1 + \frac{1}{\log N}\right)2^{-k}\abs{B_u \setminus Y_\mathrm{out}} \right] \\
        (**) &\le& q \cdot \Pr\left[\Bin\left(\abs{B_u}, 2^{-k}\right) + t > \left(1 + \frac{3}{4\log N}\right)2^{-k}\abs{B_u} \right] \\
        (*{*}*) &\le& q \cdot \Pr\left[\Bin\left(\abs{B_u}, 2^{-k}\right) > \left(1 + \frac{1}{2 \log N}\right)2^{-k}\abs{B_u} \right] \\
        &\le& q \cdot e^{-\frac{1}{12 \log^2 N} \cdot \branchZubndZexp}
        < \frac{q}{\log N}
    \end{eqnarray*}
    $(*)$: since the random variable $\Bin\left(\abs{B_u}, 2^{-k}\right)$ has ``more opportunities'' to be bigger than a given bound $\Bin\left(\abs{B_u} - (t + \abs{A_\mathrm{small}}), 2^{-k}\right)$.\\
    $(**)$: since $\abs{B_u \setminus Y_\mathrm{out}} \ge \left(1 - \frac{1}{8\log N}\right)\abs{B_u}$ and $\left(1 + \frac{1}{\log N}\right)\left(1 - \frac{1}{8\log N}\right) \ge \left(1 + \frac{3}{4\log N}\right)$ for large enough $N$.\\
    $(*{*}*)$: since $t \le \log \log N \le \frac{1}{8\log N} \cdot 2^{-k}\abs{B_u}$ and $\left(1 + \frac{3}{4\log N}\right)\left(1- \frac{1}{3\log N}\right) \ge \left(1 + \frac{1}{2\log N}\right)$ for large enough $N$.

    By the union bound, $\Pr_{\mathcal L}\left[G_2 \wedge G_3 \cond k, G_1\right] \ge 1 - \frac{2q}{\log N}$.
    
    Finally,
    \begin{eqnarray*}
        \Pr_{\mathcal L}\left[u \cond k, G_1, G_2, G_3 \right]
        \!&=&\! \prod_{i=1}^q \Pr_{\mathcal L}\left[y_i \cond k, G_1, G_2, G_3 \right] \\
        \!&\ge&\! \prod_{i=1}^q \begin{cases}
            \begin{array}{>{\displaystyle}c>{\displaystyle}l>{\displaystyle}l}
                1 & s_{u_i} \le 2^k / \branchZlbndZexpZdiv, & B_{u_i} \cap Y_\mathrm{old} = \emptyset \\
                \frac{1}{\abs{B_{u_i} \cap Y_\mathrm{old}}} & s_{u_i} \le 2^k / \branchZlbndZexpZdiv, & B_{u_i} \cap Y_\mathrm{old} \ne \emptyset \\
                \frac{1}{\left(1 + \frac{1}{\log N}\right)2^{-k} \abs{B_{u_i} \setminus Y_\mathrm{out}}} & s_{u_i} \ge 2^k \cdot \branchZubndZexp
            \end{array}
        \end{cases} \\
        \text{[Observation \ref{obs:pr-good-u-given-k-in-the-uniform-model}]} \!&=&\! \frac{2^{kt}}{\left(1 + \frac{1}{\log N}\right)^t} \Pr_{\mathcal U}[u|k]
        \ge 2^{kt} \left(1 - \frac{t}{\log N}\right) \Pr_{\mathcal U}[u|k]
        \ge 2^{kt} \left(1 - \frac{q}{\log N}\right) \Pr_{\mathcal U}[u|k]
    \end{eqnarray*}

    Combined,
    \[  \Pr_{\mathcal L}\left[u \cond k\right] \ge 2^{-kt} \cdot \left(1 - \frac{2q}{\log N}\right) \cdot 2^{kt} \left(1 - \frac{q}{\log N}\right)\Pr_{\mathcal U}[u|k]
        \ge \left(1 - \frac{3q}{\log N}\right) \Pr_{\mathcal U}[u|k] \qedhere \]
\end{proof}

\begin{lemma} \label{lemma:leaking-close-to-uniform}
    Consider a deterministic algorithm making $q\leq \branchZubndZq$ queries in the common framework for conditional sampling algorithms. For at least $\left(1 - \branchZfracZbadZkZchoices\right)n$ choices of $k$, where $n = k_\mathrm{max} - k_\mathrm{min} + 1$, the distance between the distributions over execution paths of the algorithm, when executed on either the leaking model or on the uniform model, is bounded by $\frac{1}{\log \log N}$ when considering $\mu \sim D_k$.
\end{lemma}
\begin{proof}
    By Lemma \ref{lemma:tons-of-ks-with-hp-good-leaf}, $\Pr_{\mathcal U}\left[\mathcal G_k | k\right] \ge 1 - \branchZfracZbadZleafZforZgoodZk$ for $\left(1 - \branchZfracZbadZkZchoices\right)n$ choices of $k$. By Lemma \ref{lemma:good-u-is-multiplicative-large}, if $u \in \mathcal G_k$, then $\Pr_{\mathcal L}[u|k] \ge \left(1 - \frac{3q}{\log N}\right)\Pr_{\mathcal U}[u|k]$. Hence, by Lemma \ref{lemma:adaptive-yao-template-bound}, the total variation distance between the distribution of the respective runs is bounded by $\frac{3q}{\log N} + \branchZfracZbadZleafZforZgoodZk \le \frac{3\log \log N}{\log N} + \branchZfracZbadZleafZforZgoodZk \le \frac{1}{\log \log N}$ for these choices of $k$.
\end{proof}

We are now ready to prove our lower bound. Note that in particular it applies to algorithms solving the $(\frac19,\frac19)$-estimation task.

\begin{theorem} \label{th:lbnd-tight-c-eps-task}
    Every conditional sampling algorithm that, with probability at least $p$ for a fixed $p>\frac12$, can estimate an element drawn from $\mu$ within a factor of $1\pm\frac19$, must draw $\Omega(\log \log N)$ conditional samples.
\end{theorem}
\begin{proof}
    By Observation \ref{obs:estimation-to-calc-k}, such an algorithm can compute $k$ with probability at least $p-o(1)$ when its input $(k,\mu)$ is drawn from $D$ (that is, $k$ is uniformly drawn from $\{k_\mathrm{min},\ldots,k_\mathrm{max}\}$ and then $\mu$ is drawn from $D_k$). By Lemma \ref{lemma:leaking-close-to-uniform}, unless $q > \branchZubndZq$, the chosen $k$ is with probability $1-o(1)$ such that the same algorithm, when executed in the uniform conditional sampling model, has its distribution over runs $o(1)$-close to the one produced by the leaking model. Hence the algorithm can compute $k$ with probability $p-o(1) > \frac12$ under the uniform conditional model as well. By Lemma \ref{lemma:uniform-condition-to-binary-search} and Observation \ref{obs:binary-search-lbnd}, $\branchZubndZq$ queries do not suffice for computing $k$ in this model with success probability greater than $\frac12$, and hence the algorithm must make strictly more than $\branchZubndZq = \Omega(\log \log N)$ queries.
\end{proof}

\subsection{Lower bound estimation task under weaker models}
\label{subsec:lbnd-c-eps-est::of-sec:lbnds}

Recall Theorem \ref{th:ubnd-nontol-equivalence}:
\thZubndZnontolZequivalence*

The proof of Theorem \ref{th:ubnd-nontol-equivalence} assumes that the $(c,\eps)$-peek oracle can be simulated using $T$ conditional samples in expectation where $c=\eps$, and obtains the upper bound of $T \cdot \tilde{O}(1 / \eps)$ conditional samples for an $\eps$-test of equivalence.

We now review two well-investigated distribution testing models that are more restrictive than the full one in which our estimator operates. For each of them we use a known lower bound on the equivalence testing task along with the above observation to provide a corresponding lower bound for the $(c,\eps)$-estimation task.

\begin{definition}[The subcube conditional oracle]
    A set $A \subseteq \{0,1\}^n$ is a \emph{subcube} if there exist $A_1,\ldots,A_n \subseteq \{0,1\}$ for which $A = A_1 \times \cdots \times A_n$. The \emph{subcube conditional oracle} is the restriction of the conditional oracle to answer only subcube condition sets.
\end{definition}


For product distributions, \cite[Theorem 43]{baynets17} shows a lower bound of $\tilde{\Omega}(n)$ samples on $(\eps/2,\eps)$-tolerant equivalence testing of product distributions over $\{0,1\}^n$ (the size of the sample set is $N=2^n$). \cite{jhw18} improves the $\eps$-dependency of the lower bound. As observed in \cite{adar2024improved}, subcube conditional sampling has no additional power over unconditional sampling when the input distributions are guaranteed to be product distributions. This implies the following bound.

\begin{corollary}
    Every algorithm that solves the $(c,\eps)$-estimation task using subcube sampling must make at least $\tilde{\Omega}(\log N)$ subcube queries in expectation for every sufficiently small $\eps > 0$ and $c>0$.
\end{corollary}

\begin{definition}[The interval conditional oracle]
    A set $A \subseteq \{1,\ldots,N\}$ is an \emph{interval} if there exist $1 \le a \le b \le N$ for which $A = \{ i : a \le i \le b \}$. The \emph{interval conditional oracle} is the restriction of the conditional oracle to answer only interval condition sets.
\end{definition}

For interval conditions, \cite{crs15} show a lower bound of $\tilde{\Omega}(\log N)$ interval queries for uniformity testing, which is a special case of equivalence testing. This implies the following bound.

\begin{corollary}
    Every algorithm that solves the $(c,\eps)$-estimation task using interval conditions must make at least $\tilde{\Omega}(\log N)$ subcube queries in expectation for every sufficiently small $\eps > 0$ and $c>0$.
\end{corollary}

Note that the polylogarithmic algorithm from \cite{CFGM16} in particular applies to both the subcube conditional model and the interval conditional model. The above corollaries in particular imply a limit on the possibility for its improvement.

\subsection{Lower bound for testing label-invariant properties}
\label{subsec:lbnd-lbl-inv::of-sec:lbnds}

We show that there exist a label-invariant property that has an $\Omega(\log N / \eps)$ lower bound for $\eps$-testing using the conditional model for every sufficiently small $\eps > 0$. We show that some $k$-bit string property is linearly hard to test in an ad-hoc testing model, and encode string instances related to this property in the histogram of distributions over a domain of size $N = 2^{\Omega(\eps k)}$.

\begin{definition}[Notations]~

    \begin{itemize}
        \item Let $X$ be a set. We use $2^X$ to denote the set of all subsets of $X$.
        \item Let $I$ be a set of integers. For an integer $k$, we use $k - I$ to denote the set $\{ k-i : i \in I \}$.
        \item Let $I$ be a set of integers. For an integer $k$, we use $\neg_k I$ to denote the set $\{1,\ldots,k\} \setminus I$.
    \end{itemize}
\end{definition}

\begin{definition}[$q$-uniform family]
    A family $\mathcal I \subseteq 2^{\{1,\ldots,k\}}$ is \emph{$q$-uniform} if, for every subset $J \subseteq \{1,\ldots,k\}$ of size $q$, the intersection of $J$ with a uniformly drawn set $I \sim \mathcal I$ is uniformly distributed over $2^J$.
\end{definition}

\begin{definition}[$k$-paired set]
    A set $I \subseteq \{1,\ldots,k\}$ is \emph{$k$-paired} if $(k+1)-I = \neg_k I$.
\end{definition}

\begin{definition}[paired $q$-uniform family] \label{def:paired-q-uniform}
    For an even $k$, a family $\mathcal I \subseteq 2^{\{1,\ldots,k\}}$ is \emph{paired $q$-uniform} if every $I \in \mathcal I$ is a $k$-paired set, and for every subset $J \subseteq \{1,\ldots,k\}$ of size $q$ which is disjoint from $(k+1) - J$, the intersection of $J$ with a uniformly drawn set $I \sim \mathcal I$ is uniformly distributed over $2^J$.
\end{definition}

\begin{observation} \label{obs:q-uniform-pairing}
    Let $\mathcal I \subseteq \{1,\ldots,k\}$ be a $q$-uniform family. The family $\mathcal I' = \{ I \cup ((2k+1) - (\neg_k I)) : I \in \mathcal I\} \subseteq 2^{\{1,\ldots,2k\}}$ is a paired $q$-uniform family.
\end{observation}

\begin{observation} \label{obs:q-uniform-pairing-all-J}
    For an even $k$, let $\mathcal I \subseteq 2^{\{1,\ldots,k\}}$ be a paired $q$-uniform family. For every $J \subseteq \{1,\ldots,k\}$ of size less than $q$, $I'\subseteq J$ for which $\Pr_{I\sim\mathcal I}[I\cap J=I']>0$ and $j$ for which $\{j,k+1-j\}\cap J=\emptyset$, if we uniformly draw $I \sim \mathcal I$, then $\Pr_{I\sim\mathcal I}[j \in I | J\cap I = I'] = \frac{1}{2}$.
\end{observation}
\begin{proof}
    Set $J'=J\setminus(\{1,\ldots,k/2\}\cap((k+1)-J))$. In words, $J'$ is the result of taking $J$ and removing every $j\leq k/2$ for which $\{j,k+1-j\}\subseteq J$. Note that for a random choice over family of paired sets, the events $I\cap J=I'$ and $I\cap J'=I'\cap J'$ are identical. Also note that $J'\cup\{j\}$ and $(k+1)-(J'\cup \{j\})$ are disjoint by the assertion on $j$. Hence, 
    \begin{eqnarray*}
        \Pr_{I\sim\mathcal I}[j \in I | J\cap I = I']
        &=& \Pr_{I\sim\mathcal I}[j \in I | J'\cap I = J'\cap I'] \\
        \text{[Chain rule]}
        &=& \frac{\Pr_{I\sim\mathcal I}[(j\in I)\wedge (J'\cap I = J'\cap I')]}{\Pr_{I\sim\mathcal I}[J'\cap I = J'\cap I']}\\
        &=& \frac{\Pr_{I\sim\mathcal I}[(J'\cup\{j\})\cap (I\cup\{j\}) = (J'\cap I')\cup\{j\}]}{\Pr_{I\sim\mathcal I}[J'\cap I = J'\cap I']}\\
        \text{[Definition \ref{def:paired-q-uniform}]} &=& \frac{2^{-|J'\cup\{j\}|}}{2^{-|J'|}} = \frac12
    \end{eqnarray*}
\end{proof}

\begin{definition}[$\eps$-pairwise far families]
    Two families $\mathcal I_1, \mathcal I_2 \subseteq 2^{\{1,\ldots,k\}}$ are \emph{$\eps$-pairwise far} if for every $I_1 \in \mathcal I_1$ and $I_2 \in \mathcal I_2$, $\abs{I_1 \Delta I_2} > \eps k$.
\end{definition}

\begin{lemma} \label{lemma:union-with-mirrored-negation-keeps-distance}
    Let $I_1, I_2 \subseteq \{1,\ldots,k\}$ be $\eps$-far subsets. Let $J_1 = I_1 \cup ((2k+1) - \neg_k I_1)$ and $J_2 = I_2 \cup ((2k+1) - \neg_k I_2)$. In this setting, $J_1$ and $J_2$ are $\eps$-far as well.
\end{lemma}
\begin{proof}
    \begin{eqnarray*}
        \abs{J_1 \Delta J_2}
        &=& \abs{\left(I_1 \cup ((2k+1) - (\neg_k I_1))\right) \Delta \left(I_2 \cup ((2k+1) - (\neg_k I_2))\right)} \\
        &=& \underbrace{\abs{I_1 \Delta I_2}}_{\text{in $\{1,\ldots,k\}$}} + \underbrace{\abs{((2k+1) - (\neg_k I_1)) \Delta ((2k+1) - (\neg_k I_2))}}_{\text{in $\{k+1,\ldots,2k\}$}} \\
        &=& \abs{I_1 \Delta I_2} + \abs{(\neg_k I_1) \Delta (\neg_k I_2)} \\
        &=& 2 \abs{I_1 \Delta I_2}
        > 2 \cdot \eps k
        = \eps \cdot (2k)
    \end{eqnarray*}
\end{proof}

\begin{definition}[Weighted sampling oracle]
    Let $I \subseteq \{1,\ldots,k\}$ be a subset. The \emph{weighted sampling} oracle for $I$ gets a weight function $w : \{1,\ldots,k\} \to [0,1]$ as its input, and its output is an index $i \in \{1,\ldots,k\}$ and a bit $b \in \{0,1\}$ distributed as follows:
    \begin{itemize}
        \item If $\sum_{i=1}^k w(i) > 0$, then the probability to draw the index $i$ is $\frac{w(i)}{\sum_{i=1}^k w(i)}$. The oracle returns $(i, 1)$ if $i \in I$ and $(i,0)$ if $i \notin I$.
        \item If $\sum_{i=1}^k w(i) = 0$, then the oracle indicates an error.
    \end{itemize}
\end{definition}

\begin{lemma} \label{lemma:applying-paired-q-uniformness-to-adaptivity}
    Let $\mathcal I$ be a paired $q$-uniform family of subsets of $\{1,\ldots,k\}$. A sequence of $q$ weighted sampling oracle calls with inputs $w_1,\ldots,w_q$ to a uniformly chosen $I \sim \mathcal I$ results in a sequence of pairs $(j_1,b_1), \ldots, (j_q,b_q)$ where the $j_i$s are indexes and the $b_i$s are bits. In this setting, for every $1 \le i \le q$ for which $\{j_i, (k+1)-j_i\} \cap \{j_1,\ldots,j_{i-1}\}=\emptyset$, the bit $b_i$ is uniformly distributed, even when conditioned on the values of $(b_1,\ldots,b_{i-1})$. Additionally, for every other $i$, the bit $b_i$ is a function of $j_1,\ldots,j_{i-1}$ and $b_1,\ldots,b_{i-1}$ (which does not depend on $\mathcal I$ or $I$ at all). This holds even if the input $w_i$ can be chosen based on the result of the previous $i-1$ calls.
\end{lemma}
\begin{proof}
    Note that $\mathcal I$ is non-empty since any paired $q$-uniform family must consist of at least $2^q$ sets.
    
    Consider the $i$th call ($1 \le i \le q$) to the weighting sampling oracle. It uses internal randomness and adaptivity to choose an index $j_i$ and query its belonging to the input set $I$. Let $I' = \{ j_1,\ldots,j_{i-1}\} \cap I$ be the knowledge about past queried indexes. If $j_i, (k+1)-j_i \notin \{j_1,\ldots,j_{i-1}\}$ then, due to $\mathcal I$ being paired $q$-uniform and Observation \ref{obs:q-uniform-pairing-all-J}, $\Pr_{I \sim \mathcal I}\left[j_i \in I \cond I \cap \{j_1,\ldots,j_{i-1}\} = I'\right] = \frac{1}{2}$.

    On the other hand, if $j_i=j_{i'}$ for some $i'<i$ then $b_i=b_{i'}$ deterministically, and if $k+1-j_i=j_{i'}$ for some $i'<i$ then $b_i=1-b_{i'}$ deterministically, irrespective of $I$ or $\mathcal I$.
\end{proof}

\begin{lemma} \label{lemma:parameterized-lbnd-paired-secret-sharing}
    Let $k \ge 2$, $q \ge 2$. Let $\mathcal I \subseteq 2^{\{1,\ldots,k\}}$ be a paired $q$-uniform family. If there exists another paired $q$-uniform family $\mathcal I' \subseteq 2^{\{1,\ldots,k\}}$ which is $\eps$-pairwise far from $\mathcal I$, then every $\eps$-testing algorithm for distinguishing between $\mathcal I$ and being $\eps$-far from $\mathcal I$ must make more than $q$ calls to the weighted sampling oracle.
\end{lemma}
\begin{proof}
    Let $\mathcal U = \frac{1}{2} \cdot (\mathrm{uni}(\mathcal I) \times \{1\}) + \frac{1}{2} \cdot (\mathrm{uni}(\mathcal I') \times \{0\})$ be the distribution that uniformly chooses $b \in \{0,1\}$, and then uniformly draws a set $I$ from $\mathcal I$ if $b=1$ and from $\mathcal I'$ if $b=0$. If $\mathcal A$ is an $\eps$-test for the property $\mathcal I$, then $\Pr_{(b,I) \sim \mathcal U}\left[(\mathcal A(I) = \accept) \leftrightarrow b \right] > \frac{1}{2}$, since it should accept every $I \in \mathcal I$ with probability strictly greater than $\frac{1}{2}$ and reject every $I \in \mathcal I'$ with probability strictly greater than $\frac{1}{2}$.

    If $\mathcal A$ makes at most $q$ calls to the weighted sampling oracle, then by Lemma \ref{lemma:applying-paired-q-uniformness-to-adaptivity}, it receives an identical distribution of outputs regardless of whether $I$ is drawn from $\mathcal I$ or from $\mathcal I'$. This implies that $b$ and $\mathcal A(I)$ are independent, and thus $\Pr_{(b,I) \sim \mathcal U}\left[(\mathcal A(I) = \accept) \leftrightarrow b \right] = \frac{1}{2}$. This is a contradiction, and hence $\mathcal A$ must make more than $q$ oracle calls.
\end{proof}

\begin{lemma}[Lemma 22 in \cite{pcpp19}] \label{lemma:cite:pcpp-code-ensembles-exist}
    A set $\{v_1\ldots,v_{3r}\}$ of random vectors in $\{0,1\}^{4r}$ satisfies with probability $1-o(1)$ the following two conditions: $\mathrm{Span}\{v_1,\ldots,v_{3r}\}$ is a $\frac{1}{30}$-distance code, and $\mathrm{Span}\{v_{r+1},\ldots,v_{3r}\}$ is a $\frac{1}{10}$-dual distance code.
\end{lemma}

\begin{lemma}[Direct application of Lemma \ref{lemma:cite:pcpp-code-ensembles-exist}] \label{lemma:hard-to-distinguish-string-prop-pair-exists}
    For every sufficiently large $r$, there exist two families $\mathcal J_1$ and $\mathcal J_2$ of subsets of $\{1,\ldots,4r\}$, each of them having size $2^{2r}$, such that both of them are $q$-uniform for $q=\ceil{2r/5}$, which are $\frac{1}{30}$-pairwise far from each other. Additionally, the two families contain no members with fewer than $\ceil {2r/15}$ elements.
\end{lemma}

\begin{lemma} \label{lemma:hard-to-distinguish-string-prop-pair-exists-paired}
    For every sufficiently large $r$, there exist two families $\mathcal I_1$ and $\mathcal I_2$ of subsets of $\{1,\ldots,8r\}$, each of them having size $2^{2r}$, such that both of them are paired $q$-uniform for $q=\ceil{2r/5}$, which are $\frac{1}{30}$-pairwise far from each other. Additionally, the two families contain only members with exactly $4r$ elements.
\end{lemma}
\begin{proof}
    Let $\mathcal J_1$ and $\mathcal J_2$ be two $q$-uniform families of subsets of $\{1,\ldots,4r\}$ that are $\frac{1}{30}$-pairwise far, whose existence is guaranteed by Lemma \ref{lemma:hard-to-distinguish-string-prop-pair-exists}.

    Let $\mathcal I_1 = \{ J \cup ((8r+1) - (\neg_{4r} J)) : J \in \mathcal J_1\}$ and $\mathcal I_2 = \{ J \cup ((8r+1) - (\neg_{4r} J)) : J \in \mathcal J_2 \}$. By Observation \ref{obs:q-uniform-pairing}, $\mathcal I_1$ and $\mathcal I_2$ are paired $q$-uniform families. Note that they are $\frac{1}{30}$-far by Lemma \ref{lemma:union-with-mirrored-negation-keeps-distance}.

    Every $I \in \mathcal I_1 \cup \mathcal I_2$ has size exactly $4r$ since there exists some $J \in \mathcal J_1 \cup \mathcal J_2$ for which $I = J \cup ((8r+1) - (\neg_{4r} J))$ and hence $\abs{I} = \abs{J} + \abs{\neg_{4r} J} = \abs{J} + (4r - \abs{J}) = 4r$.
\end{proof}

We now define the distributions whose histograms can encode subsets of $\{1,\ldots,8r\}$ as above, and for which we can perform a reduction from the conditional testing model.

\begin{definition}[Non-empty $k$-partition]
    A $k$-tuple $\mathcal S = (S_1,\ldots,S_k)$ is a \emph{non-empty $k$-partition} if the sets $S_1,\ldots,S_k$ are non-empty and mutually disjoint.

    If $\Omega = \bigcup_{i=1}^k S_i$, then we say that $\mathcal S$ is a non-empty $k$-partition \emph{of $\Omega$}.
\end{definition}

\begin{definition}[Chunk distribution]
    Let $\mathcal S = (S_1,\ldots,S_k)$ be a non-empty $k$-partition. Let $I \subseteq \{1,\ldots,k\}$ be a non-empty set of indexes. The \emph{chunk distribution} $\mu_{\mathcal S, I}$ over $\bigcup_{i=1}^k S_i$ is defined such that for every $i \in I$, $\mu_{\mathcal S, I}(S_i) = \frac{1}{\abs{I}}$ and the restriction of $\mu_{\mathcal S, I}$ to $S_i$ is the uniform distribution over $S_i$. More precisely, for every $i \in I$ and $j \in S_i$, $\mu_{\mathcal S,I}(j) = \frac{1}{|I| \cdot |S_i|}$, and $\mu_{\mathcal S,I}(j) = 0$ for every $j \notin \bigcup_{i \in I} S_i$.
\end{definition}

\begin{definition}[Set of chunk distributions] \label{def:set-of-chunk-distributions}
    Let $\mathcal S$ be a non-empty $k$-partition. Let $\mathcal I$ be a family of non-empty subsets of $\{1,\ldots,k\}$. The \emph{set of chunk distributions} with respect to $\mathcal S$ and $\mathcal I$ is the set $\mathcal H_{\mathcal S, \mathcal I} = \{ \mu_{\mathcal S, I} : I \in \mathcal I\}$, where $\mu_{\mathcal S, I}$ is the chunk distribution corresponding to $\mathcal S$ and $I$.
\end{definition}

\begin{definition}[The histogram property $\mathcal P_{\mathcal S, \mathcal I}$]
    Let $\mathcal S$ be a non-empty $k$-partition of a subset of a domain set $\Omega$ and $\mathcal I$ be a family of non-empty subsets of $\{1,\ldots,k\}$. The {histogram property with parameters $\mathcal S$ and $\mathcal I$} is the property $\mathcal P_{\mathcal S,\mathcal I}$ of all distributions $\mu$ over $\Omega$ that are a permutation of a distribution in $\mathcal H_{\mathcal S, \mathcal I}$.
\end{definition}

\begin{observation}
    $\mathcal P_{\mathcal S, \mathcal I}$ is label-invariant.
\end{observation}

\begin{definition}[$\rho$-increasing partition]
    A non-empty $k$-partition $\mathcal S = (S_1,\ldots,S_k)$ is $\rho$-increasing if for every $2 \le i \le k$, $\abs{S_i} \ge \rho \abs{S_{i-1}}$.
\end{definition}

\begin{lemma} \label{lemma:unisize-string-to-hist-distances}
    Let $\mathcal S$ be a $(1+\eps)$-increasing non-empty $k$-partition and let $I_1, I_2 \subseteq \{1,\ldots,k\}$ be two $\frac{1}{30}$-pairwise far subsets of size exactly $\frac{1}{2}k$. In this setting, $\mu_{\mathcal S, I_1}$ is $\frac{1}{120}\eps$-far from any permutation of $\mu_{\mathcal S, I_2}$.
\end{lemma}
\begin{proof}
    Let $\mathcal S = (S_1,\ldots,S_k)$, $\Omega = \bigcup_{i=1}^k S_i$, some $j \in I_1 \setminus I_2$, $x \in A_j$, $y \in \bigcup_{i \in I_2} S_i$ and $j' \ne j$ for which $y \in A_{j'}$. Note that:
    \begin{eqnarray*}
        \mu_{\mathcal S, I_2}(y)
        = \frac{1}{\abs{I_2} \abs{S_{j'}}}
        = \frac{\abs{S_j}}{\abs{S_{j'}}} \cdot \frac{1}{\abs{I_1} \abs{S_j}}
        = \frac{\abs{S_j}}{\abs{S_{j'}}} \cdot \mu_{\mathcal S, I_1}(x)
    \end{eqnarray*}

    We have two cases with respect to the order of $j'$ and $j$.
    \begin{align*}
        & j' > j:
        &~& \frac{\abs{S_j}}{\abs{S_{j'}}} \le (1 + \eps)^{j-j'} \le (1 + \eps)^{-1} \le 1 - \frac{1}{2}\eps,
        &~& \mu_{\mathcal S, I_2}(y) \le \left(1 - \frac{1}{2}\eps\right)\mu_{\mathcal S, I_1}(x)
        \\
        & j' < j:
        &~& \frac{\abs{S_j}}{\abs{S_{j'}}} \ge (1 + \eps)^{j-j'} \ge (1 + \eps)^{+1} \ge 1 + \frac{1}{2}\eps,
        &~& \mu_{\mathcal S, I_2}(y) \ge \left(1 + \frac{1}{2}\eps\right)\mu_{\mathcal S, I_1}(x)
    \end{align*}

    In both cases, $\abs{\mu_{\mathcal S, I_1}(x) - \mu_{\mathcal S, I_2}(y)} \ge \frac{1}{2}\eps \mu_{\mathcal S, I_1}(x)$. This bound holds for every $x$ in the support of $\mu_{\mathcal S, I_1}$ and hence, for every permutation $\pi$ over $\Omega$,
    \begin{eqnarray*}
        \dtv(\mu_{\mathcal S, I_1}, \pi \mu_{\mathcal S, I_2})
        &=& \frac{1}{2} \sum_{i=1}^k \sum_{x \in S_i} \abs{\mu_{\mathcal S, I_1}(x) - \mu_{\mathcal S, I_2}(\pi(x))} \\
        &\ge& \frac{1}{2} \sum_{i \in I_1 \setminus I_2}  \sum_{x \in S_i} \abs{\mu_{\mathcal S, I_1}(x) - \mu_{\mathcal S, I_2}(\pi(x))} \\
        &\ge& \frac{1}{2} \sum_{i \in I_1 \setminus I_2} \sum_{x \in S_i} \frac{1}{2}\eps \mu_{\mathcal S, I_1}(x)
        = \frac{1}{4}\eps \sum_{i \in I_1 \setminus I_2} \abs{S_i} \cdot \frac{1}{\abs{I_1} \abs{S_i}}
        = \frac{1}{2k}\eps \cdot \abs{I_1 \setminus I_2}
    \end{eqnarray*}

    By a symmetric analysis we can obtain that $\dtv(\mu_{\mathcal S, I_2}, \pi \mu_{\mathcal S, I_1}) \ge \frac{1}{2k}\eps \cdot \abs{I_2 \setminus I_1}$.

    Since $\dtv$ is invariant under both-sides permutation ($\dtv(\pi \mu_1, \pi \mu_2) = \dtv(\mu_1, \mu_2)$) we obtain that:
    \begin{eqnarray*}
        \min_\pi \dtv(\mu_{\mathcal S, I_1}, \pi \mu_{\mathcal S, I_2})
        &=& \min_\pi \dtv(\mu_{\mathcal S, I_2}, \pi \mu_{\mathcal S, I_1}) \\
        &\ge& \frac{1}{2k}\eps \max\left\{ \abs{I_1 \setminus I_2}, \abs{I_2 \setminus I_1} \right\} \\
        &\ge& \frac{1}{2k}\eps \cdot \frac{1}{2}\abs{I_1 \Delta I_2}
        > \frac{1}{4k}\eps \cdot \frac{1}{30}k
        = \frac{1}{120}\eps
    \end{eqnarray*}
\end{proof}

Next, we show that the conditional oracle for chunk distributions can be simulated using the weighted oracle (Algorithm \ref{fig:alg:init-cond-simulator} to initialize, \ref{fig:alg:simulate-cond-by-weighted-sampling} to simulate).

\begin{algo}
    \procname{$\procnameZinitializeZcondZsimulator(k, \mathcal S, I)$}
    \label{fig:alg:init-cond-simulator}
    \alginput{$I \subseteq \{1,\ldots,k\}$, accessible only through the weighted sampling oracle}
    \begin{code}
        \algitem Let $I' \gets \emptyset$ be the (initially empty) partial knowledge about elements in $I$.
        \algitem Let $J' \gets \emptyset$ be  the (initially empty) partial knowledge about elements outside $I$.
         \algitem Return $(k, \mathcal S, I, I', J')$.
    \end{code}
\end{algo}

\begin{algo}
    \procname{$\procnameZsimulateZcond(\mathit{obj}, C)$}
    \label{fig:alg:simulate-cond-by-weighted-sampling}
    \alginput{An object $\mathit{obj}$ created by $\procnameZinitializeZcondZsimulator$}
    \alginput{A condition set $C$}
    \algcontract{Side effects}{The algorithm may alter the $I'$, $J'$ components of $\mathit{obj}$}
    \algoutput{A sample $x \sim \mu_{\mathcal S, I}$ conditioned on $C$, or $x \sim C$ uniformly if $\mu_{\mathcal S, I}(C) = 0$}
    \begin{code}
        \algitem Let $k, \mathcal S, I, I', J'$ be the components of $\mathit{obj}$ as a $5$-tuple.
        \algitem Let $S_1,\ldots,S_k$ be the components of $\mathcal S$ as a $k$-tuple.
        \begin{While}{not explicitly terminated}
            \algitem Let $\hat{C} \gets \{ 1 \le i \le k : (S_i \cap C \ne \emptyset) \wedge (i \notin J') \wedge ((k+1) - i \notin I') \}$.
            \algpushcomment{$\mu_{\mathcal S,I}(C)=0$}
            \begin{If}{$\hat{C} = \emptyset$}
                \algitem Draw $x \sim C$ uniformly.
                \algitem Return $x$.
            \end{If}
            \begin{Else}
                \algitem Let $w$ be the weight function for which:
                \begin{Codeblock*}
                    \algitem If $i \in \hat{C}$, then $w(i) = \frac{\abs{S_i \cap C}}{\abs{S_i}}$.
                    \algitem If $i \notin \hat{C}$, then $w(i) = 0$.
                \end{Codeblock*}
                \algitem Call the weight sampling oracle for $I$ with $w$ to obtain $(i, b)$.
                \algpushcomment{($i \in I$)}
                \begin{If}{$b = 1$}
                    \algitem Add $i$ to $I'$.
                    \algitem Draw $x \sim S_i \cap C$ uniformly.
                    \algitem Return $x$.
                \end{If}
                \algpushcomment{($i \notin I$)}
                \begin{Else}
                    \algitem Add $i$ to $J'$.
                \end{Else}
            \end{Else}
        \end{While}
    \end{code}
\end{algo}

\begin{lemma} \label{lemma:cond-simulator-correct}
    For every $I \subseteq \{1,\ldots,k\}$ that is a member of a paired family, the distribution of output of a sequence starting with a single call to \procnameZinitializeZcondZsimulator (Algorithm \ref{fig:alg:init-cond-simulator}) with $k$, $\mathcal S$, $I$ followed by $q$ calls to \procnameZsimulateZcond (Algorithm \ref{fig:alg:simulate-cond-by-weighted-sampling}) over the produced object with the condition sets $C_1, \ldots, C_q$ is identical to the distribution of output of a sequence that, for each $1 \le i \le q$, draws $x_i$ from $\mu_{\mathcal S, I}$ conditioned on $C_i$ with the fallback of uniformly drawing $x_i$ from $C_i$ if $\mu_{\mathcal S, I}(C_i) = 0$. This bound holds also for an adaptive choice of each $C_i$ based on $x_1,\ldots,x_{i-1}$.
\end{lemma}
\begin{proof}
    Observe that, if $\mu_{\mathcal S, I}(C) > 0$, then using the sets $I'$ and $J'$ only affects the query complexity, and not the distribution of the output, since we ignore zeroes. Hence, Algorithm \ref{fig:alg:simulate-cond-by-weighted-sampling} is identical to a rejection sampler of the distribution over $\{1,\ldots,k\}$ defined by $w$ with respect to the event $I$.
    
    If $\mu_{\mathcal S, I}(C) = 0$, then in every iteration of the while loop in which $\hat{C} \ne \emptyset$, the call to the weighted oracle results in an output of the form $(i,0)$ for some $i \in \hat{C}$. Since the algorithm keeps history, every such $i$ is then excluded from further iterations, and hence after at most $k$ steps the set $\hat{C}$ becomes empty. When this happens, the algorithm exits the loop and uniformly draws $x \sim C$.
\end{proof}

\begin{lemma} \label{lemma:numeric-bound-for-simulating-cond-by-weighted-paired}
    Let $q \ge 5$ and $\mathcal I$ be a paired $4q$-uniform family of subsets of $\{1,\ldots,k\}$. If $I$ is drawn uniformly from $\mathcal I$, then with probability at least $\frac{9}{10}$, the simulator uses at most $4q$ weighted sampling oracle calls to simulate a sequence of $q$ conditional samples from $\mu_{\mathcal S, I}$ (according to the scheme of Lemma \ref{lemma:cond-simulator-correct}).
\end{lemma}
\begin{proof}
    Consider the $i$th sampling oracle call for $1 \le i\le 4q$ (inside the $j$th call of Algorithm \ref{fig:alg:simulate-cond-by-weighted-sampling} for some $1\le j \le q$). The probability to terminate is at least $\frac{1}{2}$: if we query an already-queried index (or its paired index) then we always terminate (because we take care to never query an index that is already known to be zero), and if we query a new bit, then the probability that its value is $1$ is exactly $\frac{1}{2}$, even if conditioned on past queries, due to the $4q$-uniformness and Observation \ref{obs:q-uniform-pairing-all-J}.

    The probability to make $4q$ oracle calls before the $q$th termination is bounded by (Chernoff):
    \[  \Pr\left[\Bin(4q, 1/2) < q\right]
        \le e^{-2 (2q - q)^2 / (4q)}
        = e^{-q/2} < \frac{1}{10} \qedhere \]
\end{proof}

\begin{restatable}{lemma}{lemmaZtechnicalZsizesZofZAi} \label{lemma:technical-sizes-of-Ai}
    Consider the sequence where $N_1 = 1$ and for every $i \ge 2$, $N_i = \ceil{(1 + 120\eps)N_{i-1}}$. For every $N \ge 1$, $\eps < \frac{1}{120}$ and $k \le \ln N / (240 \eps)$, $\sum_{i=1}^k N_i < \frac{\sqrt{N} \log_2 N}{\eps^2}$.
\end{restatable}
We prove Lemma \ref{lemma:technical-sizes-of-Ai} in Appendix \ref{apx:tldr}.

\begin{theorem}[Lower bound for label-invariant testing] \label{th:lbnd-label-invariant-testing-new}
    For every sufficiently small $\eps > 0$ and every sufficiently large $N$, there exists a label-invariant property of distributions over $\{1,\ldots,N\}$ for which every $\eps$-test must draw $\Omega(\log N / \eps)$ conditional samples.
\end{theorem}
\begin{proof}
    If $N \le 1/\eps^5$, then we can use the trivial lower-bound $\Omega(1 / \eps^2) = \omega(\log N / \eps)$ of distinguishing between the uniform distribution over $\{1,2\}$ and the distribution that draws $1$ with probability $\frac{1}{2} + \frac{1}{2}\eps$ and $2$ with probability $\frac{1}{2} - \frac{1}{2}\eps$. In the following we assume that $N > 1/\eps^5$. For sufficiently small $\eps$, this implies that $\sqrt{N} \ln N / \eps^2 \le N$.

    Let $\Omega = \{1,\ldots,N\}$, $q = 2 \floor{ \frac{1}{2}\ln N / (4800 \cdot \eps)}$, $r = \frac{5}{2} q$, $k' = 4r = 10q$, $k = 2k' = 20q$. If $q = 0$ then the lower bound of one query is trivial. Hence, in the following we assume that $N$ is sufficiently large to have $q \ge 1$.

    Observe that $k \le \ln N / (240\eps)$. Let $\mathcal S = (S_1,\ldots,S_k)$ be the following $1 + 120\eps$-increasing non-empty $k$-partition: let $N_1 = 1$ and $N_i = \ceil{(1 + 120\eps)N_{i-1}}$ for every $2 \le i \le k$. The size of $S_i$ is $N_i$ for $1 \le i \le k-1$ and at least $N_k$ for $i=k$. Such a partition exists since Lemma \ref{lemma:technical-sizes-of-Ai} guarantees that $\sum_{i=1}^k N_i \le \frac{\sqrt{N} \ln N}{\eps^2} \le N$.
    
    By Lemma \ref{lemma:hard-to-distinguish-string-prop-pair-exists-paired}, there exists two paired $q$-uniform ($q=\frac{2}{5}r$) properties $\mathcal I_1$ and $\mathcal I_2$ of subsets of $\{1,\ldots,k\}$ ($k=8r$) that are $\frac{1}{30}$-pairwise far and that only consist of subsets of size $\frac{1}{2}k = k'$.
    
    By Lemma \ref{lemma:parameterized-lbnd-paired-secret-sharing}, every algorithm that distinguishes between $\mathcal I_1$ and $\mathcal I_2$ with success probability greater than $1/2$ must make at least $q$ calls to the weighted sampling oracle.
    
    By Lemma \ref{lemma:unisize-string-to-hist-distances}, for every $\mu_1 \in \mathcal P_{\mathcal S, \mathcal I_1}$ and $\mu_2 \in \mathcal P_{\mathcal S, \mathcal I_2}$, $\mu_2$ is $\frac{1}{120}\eps$-far from every relabeling of $\mu_1$.

    Consider a $\frac{1}{120}\eps$-testing algorithm $\mathcal A$ for $\mathcal P_{\mathcal S, \mathcal I_1}$. In particular, $\mathcal A$ distinguishes between $\mathcal P_{\mathcal S, \mathcal I_1}$ and $\mathcal P_{\mathcal S, \mathcal I_2}$ with success probability at least $2/3$.
    
    By Lemma \ref{lemma:numeric-bound-for-simulating-cond-by-weighted-paired}, we can construct an algorithm $\mathcal A'$ that distinguishes between $\mathcal I_1$ and $\mathcal I_2$ by simulating the conditional sampling calls of $\mathcal A$ using weighted sampling. With probability at least $9/10$, the number of weighted samples used by the simulation is at most four times the number of conditional samples drawn by $\mathcal A$. If we terminate the simulation after reaching this bound, it can still distinguish between $\mathcal I_1$ and $\mathcal I_2$ with probability at least $2/3 - 1/10 > 1/2$.
    
    Since $\mathcal I_1$ and $\mathcal I_2$ are indistinguishable using $q$ or fewer weighted samples with any success probability greater than $1/2$, $\mathcal A$ must draw strictly more than $q/4 = \Omega(\log N / \eps)$ conditional samples.
\end{proof}

\begin{corollary} \label{cor:lbnd-histogram}
    For every sufficiently small $\eps > 0$ and every sufficiently large $N$, every algorithm that solves that $\eps$-histogram learning task must draw $\Omega(\log N / \eps)$ conditional samples.
\end{corollary}
\begin{proof}
    This holds since we can $\eps$-test any label-invariant property by $\eps/4$-histogram learning the input distribution (see Corollary \ref{cor:universal-tester-label-invariant}).
\end{proof}

\newpage
\appendix

\section{Summary of paper notations}
\label{apx:notation-table}

The following table summarizes the specific notation (mostly set in Subsection \ref{sec:prelims:subsec:paper-specific-notations}) that is used throughout this paper.

\begin{center}
    \begin{tabularx}{\textwidth}{ |l|l|X|l| }
        \hline
        \textbf{Notation} & \textbf{Name}  & \textbf{Short description} & \textbf{Definition} \\
        \hline
        $L_x$ & The $x$-light set & Set of $\mu(y) \le \mu(x)$ & \ref{def:LMH-x} \\
        \hline
        $M_x$ & The $x$-medium set & Set of $\mu(x) < \mu(y) < 1.2 \mu(x)$ & \ref{def:LMH-x} \\
        \hline
        $H_x$ & The $x$-heavy set & Set of $\mu(y) \ge 1.2\mu(x)$ & \ref{def:LMH-x} \\
        \hline
        $\eta_{c,\eps}$ & The target error & $\min\left\{\paperZtargeterrZpreamble, \paperZtargeterrZevbeta, \paperZtargeterrZestsingle \right\}$ & \ref{def:target-error} \\
        \hline 
        $f_x$, $f_{x,c,e}$ & The target function & The acceptance probability of a canonical procedure (\procnameZtargetZtestHREF) that distinguishes between $L_x$ and $H_x$ & \ref{def:target-function} \\
        \hline
        $V_x$, $V_{x,c,\eps}$ & The target set & Contains all of $L_x$, a random subset of $M_x$ (using $f_x$), and disjoint from $H_x$ & \ref{def:V-x} \\
        \hline
        $s_x$, $s_{x,c,\eps}$ & The scale mass & $\E[\mu(V_x)]$ & \ref{def:s-x} \\
        \hline
        $w_x$, $w_{x,c,\eps}$ & The weight of $x$ & $\mu(x) + s_x$, & \ref{def:w-x} \\
        & & ``the CDF of $x$ and possibly a bit more'' & \\
        \hline
        $\alpha$ & The filter density & A parameter for the filter set & \ref{def:A-alpha} \\
        \hline
        $A_\alpha$ & The filter set & Every element except $x$ belongs to $A_\alpha$ with probability $\alpha$, iid & \ref{def:A-alpha} \\
        \hline
        $V_{x,\alpha}$, $V_{x,\alpha,c,\eps}$ & The filtered target set & $V_x \cap A_\alpha$ & \ref{def:V-x-alpha} \\
        \hline
        $\beta_{x,\alpha}$, $\beta_{x,\alpha,c,\eps}$ & The filtered density & $\Pr_\mu\left[\neg x \cond V_{x,\alpha} \cup \{x\}\right]$, & \ref{def:beta-x-alpha} \\
        &  & equals $\mu\left(V_{x,\alpha}\right)/\mu\left(V_{x,\alpha}\cup\{x\}\right)$ &  \\
        \hline
        $\gamma_x$ & The goal magnitude & $\mu(x) / s_x$, a good $\alpha$ is $\Theta(\gamma_x)$ & \ref{def:gamma-x} \\
        \hline
        $\kappa$ & Used in \procnameZtargetZtestZexplicitHREF & Hard-coded to \paperZvalofZkappa & Section \ref{sec:estimation-procs} \\
        \hline
        $h(\beta)$ & \multirow{2}{\widthof{The filtered target set}}{Truncated assessment function} & $h(\beta) = \min\{ \beta / (1 - \beta), T\}$, & Section \ref{sec:est-mu-x-using-alpha} \\
        & & $T = 8 \ln \eps^{-1} + 100$ & \\
        \hline
        $\Ct[X | B]$ & Contribution of $X$ over $B$ & $\E[X | B] \cdot \Pr[B]$ & \ref{def:contribution} \\
        \hline
        $\dhist(\mu;\tau)$ & Histogram divergence & Minimum $\eps$ for which there exists a permutation $\pi$ over the domain such that $\Pr\left[\mu(x) \notin (1 \pm \eps)\tau(\pi(x)) \right] \le \eps$ & \ref{def:histogram-divergence} \\
        \hline
    \end{tabularx}
\end{center}

\section{Procedural dependency chart}\label{apx:call-chart}

The following diagram describes the calling dependencies of the various procedures defined in Sections \ref{sec:our-algorithm} through \ref{sec:est-mu-x-using-alpha}. These are grouped by function. Procedures in gray are not called from any procedure outside their group, but they may themselves call an outside procedure (following the outgoing arrows) at the behest of the procedure that called them from their own group.

\includegraphics{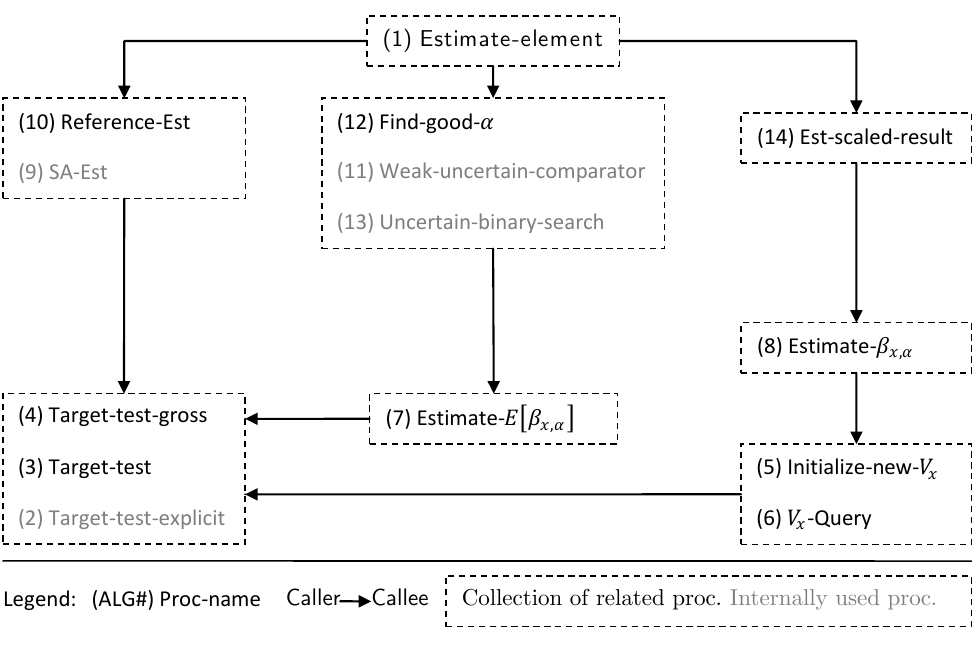}

\section{Reducing the extreme constants in the estimator at a cost}
\label{apx:target-test-galactic}

The target-test (Algorithm \ref{fig:alg:target-test-explicit}) presented in Section \ref{sec:estimation-procs} requires impractical constant factors, which carry over to the estimator. In this appendix we show that, at the cost of an additional $O(\log \frac{1}{\eps c})$-penalty applied to the $\log \log N$-factor of \procnameZmainZsingleHREF, we can significantly reduce them.

Most of the algorithmic parts require that, for every $y \in \Omega$:
\[ \abs{\Pr[y \in V_x] - \Pr[\procnameZtargetZtestHREF(x,y) = \accept]} \le \eta_{c,\eps} \]

Apart from that, Procedure $\procnameZestZexpectedZbeta$ requires that, for every $y \in \Omega$:
\[ \abs{\Pr[y \in V_x] - \Pr[\procnameZtargetZtestZgrossHREF(x,y) = \accept]} \le \frac{1}{10^8} \]

We observe that, if \procnameZtargetZtestZgross would be identical to \procnameZtargetZtest, then it would still satisfy the second constraint, since for $y \in L_x \cup H_x$ this follows from the first constraint (and from the bound $\eta_{c,\eps} \le 10^8$), and for $y \in M_x$, $\Pr[y \in V_x]$ and $\Pr[\procnameZtargetZtestHREF(x,y) = \accept]$ are identical by definition and hence their difference is zero. Hence, it suffices to implement \procnameZtargetZtest more efficiently and use it instead of \procnameZtargetZtestZgross. Note that this is where the $O(\log \frac{1}{c \eps})$-penalty comes from: the binary search makes $O(\log \log N)$ calls to \procnameZestZexpectedZbeta, and every such a call under this scheme now uses \procnameZtargetZtest, whose implementation requires $O(\log \frac{1}{c \eps})$ conditional samples rather than $O(1)$ conditional samples.

At this point we present \textsf{Target-test-lightweight}, which is a cheaper implementation of $\procnameZtargetZtest$ (but does not have a compatible ``gross'' version), and show that it satisfies the required constraints.

\begin{algo}
    \procname{$\textsf{Target-test-lightweight}(\mu, c, \eps; x, y)$}
    \label{fig:alg:target-test-classic}
    \alginput{$y \in \Omega$}
    \algoutput{\accept or \reject}
    \begin{code}
        \algpushcomment{Technical guarantee}
        \begin{If}{$y = x$}
            \algitem \reject
        \end{If}
        \algitem Let $\ell \gets \ceil{968 \ln \eta_{c,\eps}^{-1}}$.
        \algitem Draw $z_1,\ldots,z_\ell$ independent samples from $\mu$ conditioned on $\{x, y\}$.
        \algitem Let $Y = \abs{\{i : z_i = y \}}$.
        \begin{If}{$Y < \frac{23}{44} \ell$}
            \algitem \accept.
        \end{If}
        \begin{Else}
            \algitem \reject.
        \end{Else}
    \end{code}
\end{algo}

\begin{lemma}[\textsf{Target-test-lightweight}] \label{lemma:alg:target-function-classic-correct}
    Algorithm \ref{fig:alg:target-test-classic} uses $O(\log \frac{1}{\eps c})$ conditional samples, accepts with probability at least $1-\eta_{c,\eps}$ if the input belongs to $L_x$ and rejects with probability at least $1-\eta_{c,\eps}$ if the input belongs to $H_x$.
\end{lemma}
\begin{proof}
    Observe that $Y \sim \Bin\left(\ell, \frac{\mu(y)}{\mu(x) + \mu(y)}\right)$, for $\ell \ge 968 \ln (1/\eta_{c,\eps})$.

    If $y \in L_x$, then $\mu(y) \le \mu(x)$ and $\E[Y] \le \frac{1}{2}\ell$. The probability to reject is bounded by $e^{-2\left(\frac{23}{44} - \frac{1}{2}\right)^2 \ell} \le e^{-\frac{1}{968} \ell} \le e^{-\frac{968}{968} \ln (1/\eta_{c,\eps})} = \eta_{c,\eps}$.
    
    If $y \in H_x$, then $\mu(y) \ge 1.2\mu(x)$ and $\E[Y] \ge \frac{6}{11}\ell$. The probability to accept is bounded by $e^{-2\left(\frac{23}{44} - \frac{6}{11}\right)^2 \ell} \le e^{-\frac{1}{968} \ell} \le e^{-\frac{968}{968} \ln (1/\eta_{c,\eps})} = \eta_{c,\eps}$.

    Since $\eta_{c,\eps}$ is the minimum of three expressions that are all polynomial in $c$ and $\eps$, the complexity is $O\left(\log \eta_{c,\eps}^{-1}\right) = O\left(\log \frac{1}{\eps c}\right)$.
\end{proof}

\section{Another generic application lemma}
\label{apx:one-more-lemma}

We provide here a variant of Lemmas \ref{lemma:generic-application-single-distribution}, \ref{lemma:generic-application-k-distributions} and \ref{lemma:generic-application-k-distributions-testing}. While not used in the application examples that we provided Section \ref{sec:applications}, we believe that it has potential for future similar applications.

\begin{definition}[$\eps$-explicit sampling oracle]
\label{def:eps-explicit-sampling-oracle}
    Let $\mu$ be an input distribution over a set $\Omega$. The \emph{$\eps$-explicit sampling oracle} for $\mu$ has no additional input, and outputs a pair $(x,p)$, where $x \in \Omega$ distributes like $\mu$ and with probability $1$, $p \in (1 \pm \eps)\mu(x)$.
    
    The oracle guarantees \emph{consistency}, which means that if some element $y$ is drawn more than once, then all pairs of the form $(y,\cdot)$ have the same second entry.
\end{definition}
Note that the $\eps$-explicit sampling oracle is a restricted case of the $r$-lying $(c,\eps)$-explicit sampling oracle (Definition \ref{def:r-lying-c-eps-explicit-sampling-oracle}) when using $r=c=0$. In particular, following Observation \ref{obs:r-lying-c-eps-explicit-decomposition}, this oracle can be thought of as the result of sampling and receiving the corresponding values of an arbitrary (possibly probabilistic) $\eps$-approximation function $g_\mathrm{truth} : \Omega \to [0,1]$ along with the samples.

\begin{lemma} \label{lemma:generic-application-k-distributions-sampling-only}
    Consider an algorithm $\mathcal A$ whose input is a $k$-tuple $\vec{\mu} = (\mu_1,\ldots,\mu_k)$ of distributions over $\Omega_1,\ldots,\Omega_k$ (respectively), and its output is an element of a discrete set $R$. Assume that $\mathcal A$ makes at most $q$ calls to the $\eps$-explicit sampling oracle. Let $N = \max\left\{ \abs{\Omega_1}, \ldots, \abs{\Omega_k}\right\}$.

    Assume that for every input $\vec{\mu}$ there exists a set $R_{\vec{\mu}} \subseteq R$ for which $\Pr\left[\mathcal A(\vec{\mu}) \in R_{\vec{\mu}}\right] > \frac{2}{3}$ whenever the algorithm is supplied with $\eps$-explicit sampling oracles for the input distributions.

    In this setting, there exists an algorithm $\mathcal A'$ in the fully conditional model whose sample complexity is $O(q \log q \log \log N + \frac{q}{\eps^4} \poly(\log q, \log \eps^{-1}))$, such that for every input $\vec{\mu}$, $\Pr\left[\mathcal A'(\vec{\mu}) \in R_{\vec{\mu}}\right] > \frac{13}{24}$.
\end{lemma}

\begin{proof}
    We run $\mathcal A$ and simulate the outcome of the $\eps$-sampling oracle. In each call to the $\eps$-sampling oracle for $\mu_i$, we unconditionally draw $x_i \sim \mu_i$ and call \procnameZmainZsingleHREF with parameters $(\mu_i,1/24q,\eps)$ on $x_i$ (Theorem \ref{th:estimation-task}). We amplify the success probability to $1 - \frac{1}{24q}$ using the median of $\ceil{30 \ln (12q)}$ such calls (Observation \ref{obs:median-amplification}\ref{median-amplification:log-to-2c}). Each time we estimate the probability mass of an element, we record it in a ``history''. If the same element is sampled again later, we use the history record rather than calling the \procnameZmainZsingleHREF procedure again. This guarantees the consistency of the oracle (required by Definition \ref{def:eps-explicit-sampling-oracle}).

    The probability to draw an element $x_i$ for which $\CDF_{\mu_i}(x_i) < 1/24q$ is clearly bounded by $1/24q$. Therefore, by the union bound, the probability to draw such a rare element within $q$ samples is bounded by $1/24$.
    
    The probability to have a wrong estimation for any $x_i$, assuming that $\CDF_{\mu_i}(x_i) \ge 1/24q$ for all $1 \le i \le q$, is bounded by $q \cdot \frac{1}{24q} = \frac{1}{24}$. Hence, the probability to correctly simulate the $\eps$-explicit sampling oracle is at least $11/12$. This is by the union bound over two bad events: a $1/24$ bound for the event of drawing a hard-to-estimate element, and another $1/24$ bound for the event of an incorrect estimation for an estimable element. If the simulation is correct, then the output of the simulated $\mathcal A$ belongs to $R_{\vec{\mu}}$ with probability at least $2/3$. Overall, the probability of the simulation to output an element in $R_{\vec{\mu}}$ is at least $2/3 - 1/12 = 7/12$.

    By Corollary \ref{cor:expected-complexity-of::th:estimation-task} (using $c=\frac{1}{24q}$), the expected complexity of a single estimation of $x$ is $O(\log \log N) + O\left(\frac{\poly(\log q, \log \eps^{-1})}{\eps^4}\right)$. We repeat this $O(\log q)$ times for amplification of $q$ oracle calls. Therefore, the expected sample complexity is at most $O(q \log q \log \log N + \frac{q}{\eps^4} \poly(\log q, \log \eps^{-1}))$.

    By Markov's inequality, with probability at least $23/24$, the actual sample complexity is at most $24$ times the expected complexity, which is asymptotically equivalent. Overall, with probability at least $7/12-1/24 = 13/24$, the algorithm terminates after $O(q \log q \log \log N + \frac{q}{\eps^4} \poly(\log q, \log \eps^{-1}))$ samples and outputs an element in $R_{\vec{\mu}}$.
\end{proof}

\section{Technical analysis of the filtered target set}
\label{apx:tech-SHEET}

\subsection{Concentration inequalities for the filtered target set $V_{x,\alpha}$}

We prove here some Chernoff-like inequalities for the mass of $V_{x,\alpha}$, derived (unsurprisingly) by first proving a bound on the expectation of its Moment Generating Function.

\begin{lemma}[Moment Generating Function of $\mu(V_{x,\alpha})$] \label{apxlemma:V-x-alpha-MGF}
    For every $r \le 1$ (possibly negative) and $0 < \alpha \le 1$, $\E\left[e^{\frac{r}{1.2 \mu(x)} \left(V_{x,\alpha} - \E[\mu(V_{x,\alpha})]\right)}\right] \le e^{\frac{5r^2}{8 \mu(x)} \E[\mu(V_{x,\alpha})]}$.
\end{lemma}
\begin{proof}
    For every $y \in M_x$, let $X_y$ be a random variable that gets $\mu(y)$ with probability $f_x(y)\alpha$ and $0$ otherwise, and let $X = \sum_{y \in M_x} X_y$. Clearly, $\E[X] = \alpha \sum_{y \in M_x} f_x(y)\mu(y) = \E[\mu(V_{x,\alpha})] - \mu(L_x)$.

    Recall that $e^x \le 1 + x + \frac{3}{4}x^2$ for every $x \le 1$, and let $\lambda \le \frac{1}{1.2 \mu(x)}$ (possibly negative).
    \begin{eqnarray*}
        \E\left[e^{\lambda(X_y - \E[X_y])}\right]
        &=& e^{-\lambda \alpha f_x(y) \mu(y)} \cdot (\alpha f_x(y) e^{\lambda \mu(y)} + (1 - \alpha f_x(y))) \\
        &=& e^{-\lambda \alpha f_x(y)\mu(y)} \cdot (\alpha f_x(y) (e^{\lambda \mu(y)} - 1) + 1) \\
        \text{[Since $\lambda \mu(y) \le 1$]} &\le& e^{-\lambda \alpha f_x(y)\mu(y)} \cdot \left(\alpha f_x(y) (1 + (\lambda\mu(y)) + \frac{3}{4}(\lambda\mu(y))^2 - 1) + 1\right) \\
        &=& e^{-\lambda \alpha f_x(y)\mu(y)} \cdot \left(\alpha f_x(y) ((\lambda\mu(y)) + \frac{3}{4}(\lambda\mu(y))^2) + 1\right)
    \end{eqnarray*}
    
    We use the well-known bound $t+1 \le e^t$ to obtain:
    \begin{eqnarray*}
        \E\left[e^{\lambda(X_y - \E[X_y])}\right] &\le& e^{-\lambda \alpha f_x(y)\mu(y)} \cdot e^{\alpha f_x(y) ((\lambda\mu(y)) + \frac{3}{4}(\lambda\mu(y))^2)} \\
        &=& e^{\frac{3}{4}\alpha f_x(y)(\lambda\mu(y))^2} \\
        &=& e^{\frac{3}{4} \lambda^2 \mu(y) \cdot (\alpha f_x(y) \mu(y))}
        \le e^{\frac{3}{4} \cdot 1.2 \lambda^2 \cdot \mu(x) \cdot \E[X_y]}
        = e^{\frac{9}{10} \lambda^2 \cdot \mu(x) \cdot \E[X_y]}
    \end{eqnarray*}
    Since the $X_y$s are independent, this implies that $e^{\lambda (X-\E[X])} \le e^{\frac{9}{10} \lambda^2 \mu(x) \cdot \E[X]}$. We choose $\lambda = \frac{r}{1.2 \cdot \mu(x)}$ to obtain the desired bound.
\end{proof}

\begin{lemma} \label{apxlemma:V-alpha-chernoff-small-delta-up}
    For every $0 \le \delta \le 1$, $\Pr[\mu(V_{x,\alpha}) \ge (1 + \delta) \E[\mu(V_{x,\alpha})]] \le e^{-\frac{1}{4\mu(x)}\delta^2\E[\mu(V_{x,\alpha})]}$.
\end{lemma}
\begin{proof}
    We use Lemma \ref{apxlemma:V-x-alpha-MGF} using $r=\frac{3}{4} \delta$ to obtain $\E\left[e^{\frac{5\delta}{8 \mu(x)} \left(V_{x,\alpha} - \E[\mu(V_{x,\alpha})]\right)}\right] \le e^{\frac{45 \delta^2}{128 \mu(x)} \E[\mu(V_{x,\alpha})]}$.
    
    We now use Chernoff-Markov bound:
    \begin{eqnarray*}
        \Pr\left[\mu(V_{x,\alpha}) \ge (1 + \delta)\E[\mu(V_{x,\alpha})]\right]
        &=& \Pr\left[\mu(V_{x,\alpha}) - \E[\mu(V_{x,\alpha})] \ge \delta \E[\mu(V_{x,\alpha})]\right] \\
        &=& \Pr\left[e^{\frac{5\delta}{8 \mu(x)}\left(\mu(V_{x,\alpha}) - \E[\mu(V_{x,\alpha})]\right)} \ge e^{\frac{5\delta^2}{8 \mu(x)}\E[\mu(V_{x,\alpha})]}\right] \\
        &\le& e^{-\frac{5\delta^2}{8 \mu(x)}\E[\mu(V_{x,\alpha})]}\E\left[e^{\frac{5\delta}{8 \mu(x)}\left(\mu(V_{x,\alpha}) - \mu(\E[V_{x,\alpha})]\right)}\right] \\
        &\le& e^{-\frac{5\delta^2}{8 \mu(x)}\E[\mu(V_{x,\alpha})]}e^{\frac{45\delta^2}{128\mu(x)}\E[\mu(V_{x,\alpha})]} \\
        &=& e^{-\frac{35}{128\mu(x)}\delta^2\E[\mu(V_{x,\alpha})]}
        \le e^{-\frac{\delta^2}{4\mu(x)}\E[\mu(V_{x,\alpha})]}
    \end{eqnarray*}
\end{proof}

\begin{lemma} \label{apxlemma:V-alpha-chernoff-small-delta-down}
    For every $0 \le \delta \le 1$, $\Pr[\mu(V_{x,\alpha}) \le (1 - \delta) \E[\mu(V_{x,\alpha})]] \le e^{-\frac{1}{4\mu(x)}\delta^2\E[\mu(V_{x,\alpha})]}$.
\end{lemma}
\begin{proof}
    We use Lemma \ref{apxlemma:V-x-alpha-MGF} using $r=-\frac{3}{4}\delta$ to obtain $\E\left[e^{-\frac{5\delta}{8\mu(x)} \left(V_{x,\alpha} - \E[\mu(V_{x,\alpha})]\right)}\right] \le e^{\frac{45\delta^2}{128\mu(x)} \E[X_y]}$.
    
    We now use Chernoff-Markov bound:
    \begin{eqnarray*}
        \Pr\left[\mu(V_{x,\alpha}) \le (1 -\delta)\E[\mu(V_{x,\alpha})]\right]
        &=& \Pr\left[\mu(V_{x,\alpha}) - \E[\mu(V_{x,\alpha})] \le -\delta\E[\mu(V_{x,\alpha})]\right] \\
        &=& \Pr\left[e^{-\frac{5\delta}{8 \mu(x)}\left(\mu(V_{x,\alpha}) - \E[\mu(V_{x,\alpha})]\right)} \ge e^{\frac{5\delta^2}{8 \mu(x)}\E[\mu(V_{x,\alpha})]}\right] \\
        &\le& e^{-\frac{5\delta^2}{8 \mu(x)}\E[\mu(V_{x,\alpha})]}\E\left[e^{-\frac{5\delta}{8 \mu(x)}\left(\mu(V_{x,\alpha}) - \mu(\E[\mu(V_{x,\alpha})]\right)}\right] \\
        &\le& e^{-\frac{5\delta^2}{8 \mu(x)}\E[\mu(V_{x,\alpha})]}e^{\frac{45\delta^2}{128\mu(x)}\E[\mu(V_{x,\alpha})]} \\
        &=& e^{-\frac{35}{128\mu(x)}\delta^2\E[\mu(V_{x,\alpha})]}
        \le e^{-\frac{\delta^2}{4\mu(x)}\E[\mu(V_{x,\alpha})]}
    \end{eqnarray*}
\end{proof}

\begin{lemma} \label{apxlemma:V-alpha-chernoff-large-delta}
    For every $\delta > 0$, $\Pr[\mu(V_{x,\alpha}) \ge (2 + \delta) \E[\mu(V_{x,\alpha})]] \le e^{-\frac{1}{2\mu(x)}\delta\E[\mu(V_{x,\alpha})]}$.
\end{lemma}
\begin{proof}
    We use Lemma \ref{apxlemma:V-x-alpha-MGF} using $r=1$ to obtain $\E\left[e^{\frac{1}{1.2\mu(x)} \left(\mu(V_{x,\alpha}) - \E[\mu(V_{x,\alpha})]\right)}\right] \le e^{\frac{5}{8\mu(x)} \E[\mu(V_{x,\alpha})]}$.

    We now use Chernoff-Markov bound:
    \begin{eqnarray*}
        \Pr\left[\mu(V_{x,\alpha}) \ge (2 + \delta)\E[\mu(V_{x,\alpha})]\right]
        &=& \Pr\left[\mu(V_{x,\alpha}) - \E[\mu(V_{x,\alpha})] \ge (1 + \delta)\E[\mu(V_{x,\alpha})]\right] \\
        &=& \Pr\left[e^{\frac{1}{1.2\mu(x)}\left(\mu(V_{x,\alpha}) - \E[\mu(V_{x,\alpha})]\right)} \ge e^{\frac{1}{1.2\mu(x)}(1 + \delta)\E[\mu(V_{x,\alpha})]}\right] \\
        &\le& e^{-\frac{1}{1.2\mu(x)}(1 + \delta)\E[\mu(V_{x,\alpha})]}\E\left[e^{\frac{1}{1.2\mu(x)}\left(\mu(V_{x,\alpha}) - \E[\mu(V_{x,\alpha})]\right)}\right] \\
        &\le& e^{-\frac{1}{1.2\mu(x)}(1 + \delta)\E[\mu(V_{x,\alpha})]}e^{\frac{5}{8\mu(x)}\E[\mu(V_{x,\alpha})]} \\
        &\le& e^{-\frac{\delta}{1.2\mu(x)}\E[\mu(V_{x,\alpha})]}
        \le e^{-\frac{\delta}{2\mu(x)}\E[\mu(V_{x,\alpha})]}
    \end{eqnarray*}
\end{proof}

\subsection{Expectation inequalities}

\begin{lemma} \label{apxlemma:one-over-conditional-V-x-by-A-x}
    For every $0 < \alpha \le 1$, $\E\left[\frac{1}{\mu(x) + \mu(V_{x,\alpha})}\right] \le \frac{1}{\mu(x)} \cdot \left(e^{-\frac{1}{16}a} + \frac{1}{1 + \frac{1}{2}a}\right)$, where $a = \alpha/\gamma_x$.
\end{lemma}
\begin{proof}
    \begin{eqnarray*}
        && \E\left[\frac{1}{\mu(x) + \mu(V_x \cap A_\alpha)}\right] \\
        &\le& \Pr\left[\mu(V_{x,\alpha}) < \frac{1}{2}\E[\mu(V_{x,\alpha})]\right] \cdot \frac{1}{\mu(x)} + \Pr\left[\mu(V_{x,\alpha}) \ge \frac{1}{2}\E[\mu(V_{x,\alpha})]\right] \cdot \frac{1}{\mu(x) + \frac{1}{2}\E[\mu(V_{x,\alpha})]} \\
        (*) &\le& e^{-\frac{1}{16\mu(x)}\E[\mu(V_{x,\alpha})]} \cdot \frac{1}{\mu(x)} + 1 \cdot \frac{1}{\mu(x) + \frac{1}{2}\E[\mu(V_{x,\alpha})]} \\
        (**) &=& \frac{e^{-\frac{1}{16}a}}{\mu(x)} + \frac{1}{\mu(x) + \frac{1}{2}a\mu(x)}
        = \frac{1}{\mu(x)} \cdot \left(e^{-\frac{1}{16}a} + \frac{1}{1 + \frac{1}{2}a}\right)
    \end{eqnarray*}
    $(*)$: By Lemma \ref{apxlemma:V-alpha-chernoff-small-delta-down} (Chernoff $\Pr[\mu(V_{x,\alpha}) < (1-\delta)\E[\mu(V_{x,\alpha})]]$ with $\delta=\frac{1}{2}$),\\$(**)$: Since $\E[\mu(V_{x,\alpha})] = a\mu(x)$.
\end{proof}

At this point we recall and prove Lemma \ref{lemma:est-beta-expected-time}.
\lemmaZestZbetaZexpectedZtime*
\begin{proof}
    Recall that $w_x = \mu(x) + s_x$.
    
    Case I: $\mu(x) \ge \frac{1}{2}w_x$:
    \begin{eqnarray*}
        \E\left[\frac{\mu(A_\alpha \cup \{x\})}{\mu\left((V_x \cap A_\alpha) \cup \{x\}\right)}\right]
        \le \E\left[\frac{\mu(A_\alpha \cup \{x\})}{\mu(x)}\right]
        \le \E\left[\frac{1}{\mu(x)}\right] = \frac{1}{\mu(x)} \le \frac{2}{w_x}
    \end{eqnarray*}

    Case II: $\mu(x) < \frac{1}{2} w_x$ and hence $s_x \ge \frac{1}{2}w_x$. Let $a = \alpha / \gamma_x$ (hence $\E[\mu(V_{x,\alpha})] = a\mu(x)$).
    \begin{eqnarray*}
        \E\left[\frac{\mu(A_\alpha \cup \{x\})}{\mu((V_x \cap A_\alpha) \cup \{x\})}\right]
        &=& \E\left[\frac{\mu((V_x \cap A_\alpha) \cup \{x\}) + \mu(A_\alpha \setminus V_x)}{\mu((V_x \cap A_\alpha) \cup \{x\})}\right] \\
        &=& 1 + \E\left[\frac{\mu(A_\alpha \setminus V_x)}{\mu((V_x \cap A_\alpha) \cup \{x\})}\right] \\
        &=& 1 + \sum_U \Pr[V_x = U] \E\left[\frac{\mu(A_\alpha \setminus U)}{\mu((U \cap A_\alpha) \cup \{x\})}\right]
    \end{eqnarray*}
    
    Since $U$ is hard-coded, the random set $A_\alpha \setminus U$, which contains every element in $\Omega \setminus (U\cup \{x\})$ with probability $\alpha$ independently, is independent of the random set $A_\alpha \cap U$, which contains every element in $U$ with probability $\alpha$ independently. Hence,
    \begin{eqnarray*}
        \E\left[\frac{\mu(A_\alpha \cup \{x\})}{\mu((V_x \cap A_\alpha) \cup \{x\})}\right]
        &=& 1 + \sum_U \Pr[V_x = U] \E[\mu(A_\alpha \setminus U)] \E\left[\frac{1}{\mu((U \cap A_\alpha) \cup \{x\})}\right] \\
        &\le& 1 + \sum_U \Pr[V_x = U] \cdot \alpha \E\left[\frac{1}{\mu((U \cap A_\alpha) \cup \{x\})}\right] \\
        &=& 1 + \alpha \E\left[\frac{1}{\left(\mu(V_x) \cap A_\alpha\right) \cup \{x\})}\right] \\
        &=& 1 + \alpha \E\left[\frac{1}{\mu(x) + \mu(V_{x,\alpha})}\right]
    \end{eqnarray*}
    
    We can use Lemma \ref{apxlemma:one-over-conditional-V-x-by-A-x} to obtain that:
    \begin{eqnarray*}
        \E\left[\frac{\mu(A_\alpha \cup \{x\})}{\mu((V_x \cap A_\alpha) \cup \{x\})}\right] &\le& 1 + \alpha \cdot \frac{e^{-\frac{1}{16}a} + \frac{1}{1+\frac{1}{2}a}}{\mu(x)} \\
        \text{[$\alpha = a\mu(x)/s_x$]} &=& 1 + \frac{a\mu(x)}{s_x} \cdot \frac{e^{-\frac{1}{16}a} + \frac{1}{1+\frac{1}{2}a}}{\mu(x)} \\
        &=& 1 + \frac{1}{s_x} \cdot a \left(e^{-\frac{1}{16}a} + \frac{1}{1+\frac{1}{2}a}\right) \\
        (*) &\le& 1 + \frac{16e^{-1} + 2}{s_x}
        \le 1 + \frac{8}{s_x}
        \le \frac{9}{s_x}
        \le \frac{20}{w_x}
    \end{eqnarray*}
    $(*)$: since $ae^{-a/16} = 16(a/16)e^{-(a/16)} \le 16 (\sup_{t \ge 0} te^{-t}) = 16e^{-1}$ and $\frac{a}{1+a/2} \le 2$ for $a > 0$.
\end{proof}

\begin{lemma} \label{apxlemma:lbound-beta}
    For every $0 < \alpha \le 1$, $\E[\beta_{x,\alpha}] \ge \left(1 - e^{-a/9}\right)\left(1 - \frac{3}{3 + a} \right)$, where $a = \alpha/\gamma_x$.
\end{lemma}
\begin{proof}
    \begin{eqnarray*}
        \E\left[\frac{\mu(V_{x,\alpha})}{\mu(x) + \mu(V_{x,\alpha})}\right]
        &\ge& \Pr\left[\mu(V_{x,\alpha}) \ge \left(1 - \frac{2}{3}\right) \E[\mu(V_{x,\alpha})]\right] \cdot \frac{\left(1-\frac{2}{3}\right)\E[\mu(V_{x,\alpha})]}{\mu(x) + \left(1-\frac{2}{3}\right)\E[\mu(V_{x,\alpha})]} \\
        \text{[Lemma \ref{apxlemma:V-alpha-chernoff-small-delta-up}]} &\ge& \left(1 - e^{-\frac{(2/3)^2}{4\mu(x)}\E[\mu(V_{x,\alpha})]}\right) \cdot \frac{\frac{1}{3}\E[\mu(V_{x,\alpha})]}{\mu(x) + \frac{1}{3}\E[\mu(V_{x,\alpha})]} \\
        &=& \left(1 - e^{-\frac{1}{9\mu(x)} \cdot a\mu(x)}\right) \cdot \frac{\frac{1}{3} \cdot a\mu(x)}{\mu(x) + \frac{1}{3}a\mu(x)} \\
        &=& \left(1 - e^{-\frac{1}{9}a}\right) \cdot \frac{a}{3 + a}
        = \left(1 - e^{-\frac{1}{9}a}\right) \cdot \left(1 - \frac{3}{3 + a}\right)
    \end{eqnarray*}
\end{proof}

\begin{lemma} \label{apxlemma:ubound-beta}
    For every $0 < \alpha \le 1$, $\E[\beta_{x,\alpha}] \le \frac{2\sqrt{a^2 + a}}{1 + a + \sqrt{a^2 + a}}$, where $a = \alpha/\gamma_x$.
\end{lemma}
\begin{proof}
    Recall that $\E[\mu(V_{x,\alpha})] = a\mu(x)$. Let $k_a = 1 + \sqrt{1 + a^{-1}}$.
    \begin{eqnarray*}
        \E[\beta_{x,\alpha}]
        &=& \E\left[\frac{\mu(V_{x,\alpha})}{\mu(x) + \mu(V_{x,\alpha})}\right] \\
        &\le& \Pr[\mu(V_{x,\alpha}) \le k_a a \mu(x)] \cdot \frac{k_a a \mu(x)}{\mu(x) + k_a a \mu(x)} + \Pr[\mu(V_{x,\alpha}) > k_a a \mu(x)] \cdot 1
        \\
        &=& (1 - \Pr[\mu(V_{x,\alpha}) > k_a a\mu(x)]) \cdot \frac{k_a a \mu(x)}{\mu(x) + k_a a \mu(x)} + \Pr[\mu(V_{x,\alpha}) > k_a a \mu(x)] \cdot 1 \\
        &=& \frac{ak_a}{1 + ak_a} + \Pr[\mu(V_{x,\alpha}) > k_a a \mu(x)] \cdot \left(1 - \frac{ak_a}{1 + ak_a}\right)
    \end{eqnarray*}

    By Markov's inequality, $\Pr[\mu(V_{x,\alpha}) > k_a a \mu(x)] < \frac{1}{k_a}$, and hence:
    \begin{eqnarray*}
        \E[\beta_{x,\alpha}] &\le& \frac{ak_a}{1 + ak_a} + \frac{1}{k_a} \cdot \left(1 - \frac{ak_a}{1 + ak_a}\right) \\
        &=& \frac{1}{1+ak_a}\left(ak_a + \frac{1}{k_a}\right) \\
        &=& \frac{1}{1+a(1 + \sqrt{1 + a^{-1}})}\left(a(1 + \sqrt{1 + a^{-1}}) + \frac{1}{1 + \sqrt{1 + a^{-1}}}\right) \\
        &=& \frac{1}{1+a(1 + \sqrt{1 + a^{-1}})}\left(a(1 + \sqrt{1 + a^{-1}}) - a\left(1 - \sqrt{1 + a^{-1}}\right)\right) \\
        &=& \frac{2\sqrt{a^2 + a}}{1 + a + \sqrt{a^2 + a}} \\
    \end{eqnarray*}
\end{proof}

\begin{lemma} \label{apxlemma:beta-continuous-wrt-alpha}
    The function $\alpha \to \E[\beta_\alpha]$ is continuous in $0 < \alpha < 1$.
\end{lemma}
\begin{proof}
    Implied from Observation \ref{obs:set-exp-value-continuity} since $\beta_\alpha$ is bounded between $0$ and $1$ with probability $1$.
\end{proof}

We prove now Lemma \ref{lemma:good-alphas-new}, which we recall here.
\lemmaZgoodZalphasZnew*

\begin{proof}
    We apply Lemma \ref{apxlemma:lbound-beta} and Lemma \ref{apxlemma:ubound-beta} to obtain the following bounds:
    \begin{itemize}
        \item $\E[\beta_{x,2 \cdot \gamma_x}] \le \frac{2\sqrt{2^2 + 2}}{1 + 2 + \sqrt{2^2 + 2}} < 0.9$, and by Observation \ref{obs:beta-monotone-wrt-alpha}, $\E[\beta_{x,\alpha}] < 0.9$ for all $\alpha \le 2\gamma_x$ as well.
        \item $\E[\beta_{x,41\gamma_x}] \ge \left(1 - e^{-41/9}\right)\left(1 - \frac{3}{3 + 41}\right) > 0.92$, and by Observation \ref{obs:beta-monotone-wrt-alpha}, $\E[\beta_{x,\alpha}] > 0.92$ for all $\alpha \ge 41\gamma_x$ as well.
    \end{itemize}

    Also, we use these lemmas to obtain:
    \begin{itemize}
        \item $\E[\beta_{x,2.3 \cdot \gamma_x}] \le \frac{2\sqrt{(2.3)^2 + 2.3}}{1 + 2.3 + \sqrt{(2.3)^2 + 2.3}} < 0.91$.
        \item $\E[\beta_{x,38 \cdot \gamma_x}] \ge \left(1 - e^{-38/9}\right)\left(1 - \frac{3}{3 + 38}\right) > 0.91$.
    \end{itemize}
    Since the mapping $\alpha \to \E[\beta_{x,\alpha}]$ is continuous (Lemma \ref{apxlemma:beta-continuous-wrt-alpha}), the Intermediate Value Theorem guarantees the existence of $2.3\gamma_x < \alpha_x < 38\gamma_x$ for which $\E\left[\beta_{x,\alpha}\right] = 0.91$.
\end{proof}

\subsection{Technical analysis of the assessment function $h(\beta_{x,\alpha})$}

In this appendix we prove the technical lemmas of Section \ref{sec:est-mu-x-using-alpha}. Recall that we use $h(\beta) \!=\! \min\left\{\frac{\beta}{1 - \beta}, T \right\}$ for $T = 8 \ln \eps^{-1} + 100$.

\begin{restatable}{lemma}{lemmaZtinyZtailZmuZVZxZalphaZnew} \label{lemma:tiny-tail-mu-V-x-alpha-new}
    Recall that $T = 8\ln \eps^{-1} + 100$ and let $a=\alpha/\gamma_x$. If $1 \le a \le 50$ and $\eps < \frac{1}{10}$, then $\Ct\left[\mu(V_{x,\alpha}) \cond \mu(V_{x,\alpha}) > T\mu(x) \right] \le \frac{1}{10} \eps \cdot a\mu(x)$.
\end{restatable}

\begin{proof}
    Let $\hat{T} = 8\ln \eps^{-1} + 2a \le 8\ln \eps^{-1} + 100 = T$. 
    \begin{eqnarray*}
        \Pr\left[\mu(V_{x,\alpha}) > 2^t T\mu(x)\right]
        &\le& \Pr\left[\mu(V_{x,\alpha}) > 2^t \hat{T}\mu(x)\right] \\
        &=& \Pr\left[\mu(V_{x,\alpha}) > (2^t \hat{T}/a) \E[\mu(V_{x,\alpha})]\right] \\
        &=& \Pr\left[\mu(V_{x,\alpha}) > 2^t \left(\frac{8}{a} \ln \eps^{-1} + 2\right) \E[\mu(V_{x,\alpha})]\right] \\
        &\le& \Pr\left[\mu(V_{x,\alpha}) > \left(2 + 2^t \left(\frac{8}{a} \ln \eps^{-1}\right)\right) 
        \E[\mu(V_{x,\alpha})]\right] \\
        \text{[Lemma \ref{apxlemma:V-alpha-chernoff-large-delta} (Chernoff)]} &\le& e^{-\frac{2^{t+3} \ln \eps^{-1} / a}{2 \cdot \mu(x)} \cdot \E[\mu(V_{x,\alpha})]}
        = e^{-\frac{8 \cdot 2^t \ln \eps^{-1}}{2 a\mu(x)} \cdot a\mu(x)}
        = e^{-2^{t+2} \ln \eps^{-1}}
    \end{eqnarray*}

    We obtain that:
    \begin{eqnarray*}
        \Ct\left[\mu(V_{x,\alpha}) \cond \mu(V_{x,\alpha}) > T\mu(x) \right]
        &\le& \Ct\left[\mu(V_{x,\alpha}) \cond \mu(V_{x,\alpha}) > \hat{T}\mu(x) \right] \\
        &\le& \sum_{t=0}^\infty \Pr\left[2^t \hat{T} \mu(x) < \mu\left(V_x \cap A_\alpha\right) \le 2^{t+1} \hat{T} \mu(x) \right] \cdot 2^{t+1} \hat{T} \mu(x) \\
        &\le& 2\hat{T}\mu(x) \sum_{t=0}^\infty 2^t \Pr\left[\mu\left(V_x \cap A_\alpha\right) > 2^t \hat{T} \mu(x)\right] \\
        &\le& 2\hat{T}\mu(x) \sum_{t=0}^\infty 2^t \cdot e^{-2^{t+2} \ln \eps^{-1}} \\
        &\le& 2 \cdot (8 \ln \eps^{-1} + 2a) \mu(x) \sum_{t=0}^\infty e^{t \ln 2 - 2^{t+2} \ln \eps^{-1}}
    \end{eqnarray*}
    
    Since $t \ln 2 - 2^{t+2} \ln \eps^{-1} \le -(t+4)\ln \eps^{-1}$ for every $\eps < e^{-1}$ and $t \ge 0$, we obtain that:
    \begin{eqnarray*}
        \Ct\left[\mu(V_{x,\alpha}) \cond \mu(V_{x,\alpha}) > T\mu(x) \right] &\le& (16 \ln \eps^{-1} + 4a) \mu(x) \sum_{t=0}^\infty e^{-(t+4) \ln \eps^{-1}} \\
        &=& (16\ln \eps^{-1} + 4a) \mu(x) \cdot \eps^4\sum_{t=0}^\infty \eps^{t} \\
        \text{[$\eps < 1/2$]} &\le& (16 \ln \eps^{-1} + 4a) \mu(x) \cdot 2\eps^4 \\
        &=& \left(\frac{32 \ln \eps^{-1}}{a} + 8\right)\eps^3 \cdot \eps a \mu(x) \\
        \text{[$a \ge 1$]} &\le& (32 \ln \eps^{-1} + 8)\eps^3 \cdot \eps a \mu(x) \\
        \text{[$\eps < 1/10$]} &<& \frac{1}{10}\eps a\mu(x)
    \end{eqnarray*}
\end{proof}

\lemmaZEZhZbetaZapproxZEZbetaZoverZoneZminusZbetaZnew*
\begin{proof}
    Let $\beta_\mathrm{root} = 1 - \frac{1}{T+1}$ be the break-even point in the definition of $h(\beta) = \min\left\{ \frac{\beta}{1-\beta}, T \right\}$. This is the only non-differentiable point of $h$ in $(0,1)$.

    Let $a = \alpha/\gamma_x$, so that $\E[V_{x,\alpha}] = a\mu(x)$ and $\E\left[\frac{\beta_{x,\alpha}}{1-\beta_{x,\alpha}}\right] = a$.

    \begin{eqnarray*}
        0 \le \E\left[\frac{\beta_{x,\alpha}}{1-\beta_{x,\alpha}}\right] - \E\left[h(\beta_{x,\alpha})\right]
        &\le& \Ct\left[\frac{\beta_{x,\alpha}}{1-\beta_{x,\alpha}} \cond \frac{\beta_{x,\alpha}}{1-\beta_{x,\alpha}} > T \right] \\
        &=& \Ct\left[\frac{\mu(V_x \cap A_\alpha)}{\mu(x)} \cond \frac{\mu(V_x \cap A_\alpha)}{\mu(x)} > T \right] \\
        &=& \frac{1}{\mu(x)} \Ct\left[\mu(V_x \cap A_\alpha) \cond \mu(V_x \cap A_\alpha) > T\mu(x) \right] \\
        \text{[Lemma \ref{lemma:tiny-tail-mu-V-x-alpha-new}]} &\le& \frac{1}{\mu(x)} \cdot \frac{1}{10}\eps a \mu(x)
        = \frac{1}{10}\eps \E\left[\frac{\beta_{x,\alpha}}{1-\beta_{x,\alpha}}\right]
    \end{eqnarray*}

    By Observation \ref{obs:key-observation},  $\E\left[\frac{\beta_{x,\alpha}}{1 - \beta_{x,\alpha}}\right] = \alpha s_x / \mu(x)$, hence $\E[h(\beta_{x,\alpha})] = \left(1 \pm \frac{1}{10}\eps\right) \alpha s_x / \mu(x)$.
\end{proof}

\lemmaZhZerrZscaling*
\begin{proof}
    For $\beta < \frac{3}{T+1}$, note that $h'(\beta) \le \frac{1}{(1 - \beta)^2} \le 2$, hence $h(\beta \pm \delta) = h(\beta) \pm 2 \delta$ for $\beta < \frac{2}{T+1}$.

    Let $g(\beta) = \ln h(\beta)$. If $\frac{1}{T+1} < \beta < 1 - \frac{1}{T+1}$ then $g'(\beta) = \frac{h'(\beta)}{h(\beta)} = \frac{1}{\beta(1-\beta)} \le \frac{(T+1)^2}{T} \le T+3$. If $1 - \frac{1}{T+1} < \beta < 1$ then $g'(\beta) = 0$ since $h$ is fixed there. That is, $g$ is $(T+3)$-Lipschitz in $[\frac{1}{T+1}, 1]$.

    Hence, for every $\beta_1, \beta_2 \in [\frac{1}{T+1}, 1]$ for which $\abs{\beta_1 - \beta_2} \le \delta$, $\abs{g(\beta_1) - g(\beta_2)} \le (T+3)\delta \le \frac{1}{21}\eps$. This means that $\frac{h(\beta \pm \delta)}{h(\beta)} = e^{\pm \frac{1}{21}\eps} = 1 \pm \frac{1}{20}\eps$ in this range.
    
    Overall, $\abs{h(\beta \pm \delta) - h(\beta)} \le \max\left\{2\delta, \frac{1}{20}\eps h(\beta) \right\}$.
\end{proof}

\section{Long technical proofs}
\label{apx:tldr}

This paper contains some elementary statements whose proofs were omitted, such as encapsulations of long arithmetic calculations or simple inclusion-exclusions. To make it verifiable, we put here the proofs of these statements.

\lemmaZdtvZpermutationZleZtwiceZdhist*
\begin{proof}
    Let $\eps = \dhist(\mu ; \tau)$ and $\pi$ be a permutation that realizes this divergence. Let $H = \{ x : \mu(x) > (1 + \eps)\tau(\pi(x)) \}$ and $M = \{ x : \tau(x) < \mu(x) \le (1 + \eps)\tau(\pi(x)) \}$.
    \begin{eqnarray*}
        \dtv(\mu, \pi \tau)
        &=& \sum_{x : \mu(x) > \tau(\pi(x))} (\mu(x) - \tau(\pi(x))) \\
        &=& \sum_{x \in H} (\mu(x) - \tau(\pi(x))) + \sum_{x \in M} (\mu(x) - \tau(\pi(x))) \\
        &\le& \sum_{x \in H} \mu(x) + \sum_{x \in M} \eps \tau(\pi(x)) \\
        &=& \mu(H) + \eps \tau(M)
        \le \eps + \eps \cdot 1
        = 2\eps
    \end{eqnarray*}
\end{proof}

\obsZmedianZamplification*
\begin{proof}
    Let $X_i$ be the probability of the $i$th trial to be outside the desired range and $X = \sum_{i=1}^M X_i$ be the number of trials outside this range. If $X < \frac{1}{2}k$ then the median is inside the desired range. Hence, the probability that the median is wrong is bounded by $\Pr\left[Y \ge \frac{1}{2}k \right]$, where $Y \sim \Bin(k, 1/3)$.

    For parts \ref{median-amplification:9-to-5/6}, \ref{median-amplification:13-to-8/9} and \ref{median-amplification:47-to-99/100}, we explicitly bound the error probability:
    \[  \Pr\left[Y \ge \frac{1}{2}k \right]
        = \sum_{i=\ceil{k/2}}^k \binom{k}{i} \left(\frac{1}{3}\right)^i \left(\frac{2}{3}\right)^{k-i}
        = \frac{1}{3^k} \sum_{i=0}^{\floor{k/2}} \binom{k}{i} 2^i
    \]

    For a separation parameter $1 \le t \le \floor{k/2}$, we can bound the last expression using
    \begin{eqnarray*}
        \frac{1}{3^k} \sum_{i=0}^{\floor{k/2}} \binom{k}{i} 2^i
        &=& \frac{1}{3^k} \left(\sum_{i=0}^{t-1} \binom{k}{i} 2^i + 2^t \binom{k}{t} + \sum_{i=t+1}^\floor{k/2} \binom{k}{i} 2^i \right) \\
        &\le& \frac{1}{3^k} \left(\sum_{i=0}^{t-1} \binom{k}{t-1} 2^i + 2^t \binom{k}t + \sum_{i=t+1}^\floor{k/2} \binom{k}{i} 2^i \right) \\
        &\le& \frac{1}{3^k} \left(2^t \binom{k}{t-1} + 2^t \binom{k}t + \sum_{i=t+1}^\floor{k/2} \binom{k}{i} 2^i \right) 
        = \frac{1}{3^k} \left(2^t \binom{k+1}{t} + \sum_{i=t+1}^\floor{k/2} \binom{k}{i} 2^i \right)
    \end{eqnarray*}
    
    For part \ref{median-amplification:9-to-5/6} we separate in $t=3$:
    \[
        \frac{1}{3^9} \sum_{i=0}^4 \binom{9}{i} 2^i
        \le \frac{2^3 \binom{10}{3} + 2^4 \binom{9}{4}}{3^9}
        = \frac{960 + 2016}{19683}
        < \frac{1}{6}
    \]
    
    For part \ref{median-amplification:13-to-8/9} we separate in $t=5$:
    \[
        \frac{1}{3^{13}} \sum_{i=0}^6 \binom{13}{i} 2^i
        \le \frac{2^5 \binom{14}{5} + 2^6 \binom{13}{6}}{3^{13}}
        = \frac{64064 + 109824}{1594323}
        < \frac{1}{9}
    \]

    For part \ref{median-amplification:47-to-99/100} we separate in $t=23$:
    \[
        \frac{1}{3^{47}}\sum_{i=0}^{23} \binom{47}{i} 2^i
        \le \frac{2^{23} \binom{48}{23}}{3^{47}}
        \le \frac{(8.3887 \cdot 10^6) \cdot (3.0958 \cdot 10^{13})}{2.6588 \cdot 10^{22}}
        < \frac{1}{100}
    \]

    For part \ref{median-amplification:log-to-2c} and part \ref{median-amplification:log-to-24c} we use Chernoff bound:
    \[  \Pr[\mathrm{error}]
        \le \Pr\left[Y \ge \frac{1}{2}k \right]
        = \Pr\left[\Bin(k, 1/3) \ge \frac{1}{3}k + \frac{1}{6}k \right]
        \le e^{-2(k/6)^2 / k}
        =   e^{-k/18}
    \]

    In both parts, $k = \ceil{30 \ln c^{-1}} \ge 30 \ln c^{-1}$, hence $\Pr\left[Y \ge \frac{1}{2}k \right] \le e^{-(30/18)\ln c^{-1}} = e^{-(2/3) \ln c^{-1}} c$.

    \begin{itemize}
        \item Part \ref{median-amplification:log-to-2c}: if $c < 1/3$ then $e^{-(2/3) \ln c^{-1}} < \frac{1}{2}$.
        \item Part \ref{median-amplification:log-to-24c}: if $c < 1/150$ then $e^{-(2/3) \ln c^{-1}} < \frac{1}{24}$. \qedhere
    \end{itemize}
\end{proof}

\obsZmajorityZamplificationZadZhoc*
\begin{proof}
    For $k = 3$:
    \[
        \Pr\left[\Bin\left(k, \frac{5}{8}\right) \le \frac{1}{2}k\right]
        \le \Pr\left[\Bin\left(3, \frac{5}{8}\right) \le 1\right]
        = (3/8)^3 + 3 \cdot (5/8) \cdot (3/8)^2
        = \frac{162}{512}
        < \frac{1}{3}
    \]
    
    For $k \ge 45$:
    \[
        \Pr\left[\Bin\left(k, \frac{5}{8}\right) \le \frac{1}{2}k\right]
        \le e^{-2 \left(\frac{5}{8} - \frac{1}{2}\right)^2 k}
        = e^{-\frac{1}{32}k}
        \le e^{-45/32}
        < \frac{1}{4} \qedhere
    \]
\end{proof}

\lemmaZcZtruncationZadditiveZerrZtwoc*

\begin{proof}
    In this proof we use the contribution notation of Definition \ref{def:contribution}.
    
    Let:
    \begin{itemize}
        \item $L_\mu = \{ x \in \Omega : f_\mu(x) = 0 \}$.
        \item $L_\tau = \{ x \in \Omega : f_\tau(x) = 0 \}$.
        \item $H_\mu = \{ x \in \Omega : \mu(x) > \tau(x) \} \setminus (L_\mu \cup L_\tau)$.
        \item $H_\tau = \{ x \in \Omega : \tau(x) > \mu(x) \} \setminus (L_\mu \cup L_\tau)$.
        \item $M = \{ x \in \Omega : \tau(x) = \mu(x) \} \setminus (L_\mu \cup L_\tau)$.
    \end{itemize}

    Observe that:
    \begin{eqnarray*}
        \E_{x \sim \mu}\left[ \max\left\{0, 1 - \frac{f_\tau(x)}{\mu(x)} \right\} \right]
        &=& \Ct_{x \sim \mu}\left[1 \cond f_\tau(x) = 0 \right] + \Ct_{x \sim \mu}\left[\max\left\{0, 1 - \frac{\tau(x)}{\mu(x)} \right\} \cond f_\tau(x) \ne 0 \right] \\
        &=& \Ct_{x \sim \mu}\left[1 \cond x \in L_\tau \right] + \Ct_{x \sim \mu}\left[\max\left\{0, 1 - \frac{\tau(x)}{\mu(x)} \right\} \cond x \notin L_\tau \right] \\
        &=& \mu(L_\tau) + \Ct_{x \sim \mu}\left[\max\left\{0, 1 - \frac{\tau(x)}{\mu(x)} \right\} \cond x \notin L_\tau \right]
    \end{eqnarray*}
    
    Analogously,
    \[\E_{x \sim \tau}\left[ \max\left\{0, 1 - \frac{f_\mu(x)}{\tau(x)} \right\} \right] = \mu(L_\mu) + \Ct_{x \sim \tau}\left[\max\left\{0, 1 - \frac{\mu(x)}{\tau(x)} \right\} \cond x \notin L_\mu \right]\]
    
    By definition, $\dtv(\mu,\tau) = \frac{1}{2}\sum_{x \in \Omega} \abs{\mu(x) - \tau(x)}$. We use the partition $\Omega = L_\mu \cup L_\tau \cup H_\mu \cup H_\tau \cup M$.

    Set of elements that are negligible in at least one side:
    \begin{eqnarray*}
        \sum_{x \in L_\mu \cup L_\tau} \abs{\mu(x) - \tau(x)}
        &=& \sum_{x \in L_\mu} \abs{\mu(x) - \tau(x)} + \sum_{x \in L_\tau} \abs{\mu(x) - \tau(x)} - \sum_{x \in L_\mu \cap L_\tau} \abs{\mu(x) - \tau(x)} \\
        &=& (\tau(L_\mu) \pm \mu(L_\mu)) + (\mu(L_\tau) \pm \tau(L_\tau)) \pm (\mu(L_\mu \cap L_\tau) + \tau(L_\mu \cap L_\tau)) \\
        &=& (\tau(L_\mu) \pm c) + (\mu(L_\tau) \pm c) \pm 2c \\
        &=& \mu(L_\tau) + \tau(L_\mu) \pm 4c
    \end{eqnarray*}

    Observe that if $B$ is a set then $\mu(B) = \Ct_{x\sim \mu}[1 | x \in B]$. Hence,
    \begin{eqnarray*}
        \sum_{x \in L_\mu \cup L_\tau} \abs{\mu(x) - \tau(x)} &=& \Ct_{x\sim \mu}\left[1 \cond x \in L_\tau \right] + \Ct_{x\sim \tau}\left[1 \cond x \in L_\mu \right] \pm 4c \\
        \text{[$f_\tau(x) = 0$ for $x \in L_\tau$]} &=& \Ct_{x\sim \mu}\left[\max\left\{0, 1 - \frac{f_\tau(x)}{\mu(x)} \right\} \cond x \in L_\tau \right] \\
        \text{[$f_\mu(x) = 0$ for $x \in L_\mu$]} &&~ + \Ct_{x\sim \tau}\left[\max\left\{0, 1 - \frac{f_\tau(x)}{\mu(x)} \right\} \cond x \in L_\mu \right] \pm 4c
    \end{eqnarray*}

    Set of non-negligible elements in one side that are also heavier than in the other side:
    \begin{eqnarray*}
        \sum_{x \in H_\mu} \abs{\mu(x) - \tau(x)}
        &=& \sum_{x \in H_\mu} (\mu(x) - \tau(x)) \\
        &=& \sum_{x \in H_\mu} \left(\mu(x) \left(1 - \frac{\tau(x)}{\mu(x)} \right)\right) \\
        &=& \E_\mu \left[1_{H_\mu} \cdot \left(1 - \frac{\tau(x)}{\mu(x)}\right)\right] \\
        &=& \Ct_\mu \left[1_{H_\mu} \cdot \left(1 - \frac{\tau(x)}{\mu(x)}\right) \cond x \notin L_\tau \right] \\
        (*) &=& \Ct_\mu \left[\max\left\{0, 1 - \frac{\tau(x)}{\mu(x)}\right\} \cond x \notin L_\tau \right]
    \end{eqnarray*}
    $(*)$: holds since if $x \notin L_\tau, H_\mu$ then $\tau(x) \ge \mu(x)$ and the contribution is zero.
    
    By an analogous analysis,
    \[ \sum_{x \in H_\tau} \abs{\mu(x) - \tau(x)} = \Ct_\tau \left[\max\left\{0, 1 - \frac{\mu(x)}{\tau(x)}\right\} \cond x \notin L_\mu\right]\]

    For the equal-weights part $M$, clearly $\sum_{x \in M} \abs{\mu(x) - \tau(x)} = 0$.

    We sum the partial bounds to obtain:
    \begin{eqnarray*}
        2 \dtv(\mu,\tau)
        &=& \mu(L_\tau) + \Ct_\mu \left[\max\left\{0, 1 - \frac{\tau(x)}{\mu(x)}\right\} \cond x \notin L_\tau \right] 
        \\&&~ + \tau(L_\mu) + \Ct_\tau \left[\max\left\{0, 1 - \frac{\mu(x)}{\tau(x)}\right\} \cond x \notin L_\mu \right] \pm 4c \\
        &=& \E_\mu \left[\max\left\{0, 1 - \frac{f_\tau(x)}{\mu(x)}\right\} \right] + \E_\tau \left[\max\left\{0, 1 - \frac{f_\mu(x)}{\tau(x)}\right\} \right] \pm 4c
    \end{eqnarray*}
\end{proof}

\lemmaZtechnicalZsizesZofZAi*
\begin{proof}
    Observe that:
    \[
        \sum_{i=1}^k N_i
        = N_1 + \sum_{i=1}^{k-1} \ceil{(1 + 120\eps)N_i}
        \le N_1 + (1 + 120\eps) \sum_{i=1}^{k-1} N_i + (k - 1)
    \]
    
    Let $M_t = \sum_{i=1}^t N_t$. Then $M_1 = N_1$ and $M_t \le (N_1 + t - 1) + (1 + 120\eps)M_{t-1}$. By induction,
    \begin{eqnarray*}
        M_k
        &\le& \sum_{i=1}^{k-1} (1 + 120\eps)^{i-1} (N_1 + k - i) + (1 + 120\eps)^{k - 1} M_1 \\
        &\le& (N_1 + k) \sum_{i=1}^{k} (1 + 120\eps)^{i-1} \\
        &\le& (N_1 + k) \frac{(1 + 120\eps)^k}{120\eps}
        \le (N_1 + k) \frac{e^{120 \eps k}}{120\eps}
    \end{eqnarray*}

    We use $k \le \frac{\ln N}{240\eps}$ and $N_1 = 1$ to obtain:
    \[  M_k
        \le \left(1 + \frac{\ln N}{240\eps}\right) \frac{e^{\ln N / 2 + 1}}{120\eps}
        = \frac{e}{\eps^2} \left(\frac{1}{120}\eps + \frac{\ln N}{28800}\right) \sqrt{N}
        \underset{N \ge 2}\le \frac{\sqrt{N} \ln N}{\eps^2} \qedhere \]
\end{proof}

\newpage

\phantomsection
\addcontentsline{toc}{section}{Bibliography}

\bibliographystyle{alpha}
\bibliography{main}

\end{document}